\documentclass[12pt]{amsart}

\usepackage{hyperref}

\usepackage[left=1in,right=1in,top=1in,bottom=1in]{geometry}
\usepackage{amsfonts}
\usepackage{amsmath}
\usepackage{amssymb}
\usepackage{multirow}
\usepackage{fancyhdr}

\usepackage{enumerate}

\usepackage[]{algorithm2e}

\newcommand{\argmin}{\operatorname{argmin}}

\def\lsp{{\boldsymbol\ell}}
\renewcommand{\(}{\left(}
\renewcommand{\)}{\right)}

\newcommand{\abf}{{\bf a}}
\newcommand{\dbf}{{\bf d}}
\newcommand{\ebf}{{\bf e}}
\newcommand{\gbf}{{\bf g}}
\newcommand{\sbf}{{\bf s}}

\newcommand{\vbf}{{\bf v}}
\newcommand{\xbf}{{\bf x}}
\newcommand{\ybf}{{\bf y}}
\newcommand{\zbf}{{\bf z}}

\def\Abf{{\mathbf A}}
\def\Ibf{{\mathbf I}}

\def\Cbb{{\mathbb C}}
\def\Ebb{{\mathbb E}}

\def\Kbb{{\mathbb K}}
\def\Nbb{{\mathbb N}}
\def\Pbb{{\mathbb P}}
\def\Rbb{{\mathbb R}}

\def\Ccal{{\mathcal C}}

\def\Fcal{{\mathcal F}}
\def\Gcal{{\mathcal G}}
\def\Hcal{{\mathcal H}}
\def\Mcal{{\mathcal M}}
\def\Pcal{{\mathcal P}}

\def\Wcal{{\mathcal W}}

\usepackage{amsthm}
\newtheorem{acorol}{Corollary}
\newtheorem{adef}{Definition}
\newtheorem{atheorem}{Theorem}
\newtheorem{aprop}{Proposition}
\newtheorem{alemma}{Lemma}
\newtheorem{aconj}{Conjecture}
\theoremstyle{remark}

\newtheorem{rmk}{Remark}

\numberwithin{equation}{section}
\numberwithin{atheorem}{section}
\numberwithin{acorol}{section}
\numberwithin{rmk}{section}
\numberwithin{aprop}{section}

\usepackage{graphicx}
\usepackage{epstopdf}
\usepackage{subfigure}

\newcommand{\correction}[1]{{{{#1}}}}
\usepackage{todonotes}

\date{\today}

\begin{document}

\title{Local sparsity and recovery of fusion frame structured signals}

\author{Roza Aceska}
\address{Department of Mathematical Sciences, Ball State University}


\author{Jean-Luc Bouchot}
\address{School of Mathematics and Statistics, Beijing Institute of Technology}

\author{Shidong Li}
\address{Department of Mathematics, San Francisco State University}


\begin{abstract}
The problem of recovering signals of high complexity from low quality sensing devices is analyzed via a combination of tools from signal processing and harmonic analysis. 
By exploiting fusion frames, we introduce a compressed sensing framework in which we split the dense information into subchannels and fuse the local estimations. 
Each piece of information is measured by linear, potentially low quality sensors, and recovered via compressed sensing. 
Finally, by a fusion process within the fusion frames, we are able to recover accurately the original signal. 

We illustrate our findings with numerical experiments, first consider various artificial setups in which we show that splitting a signal via local projections allows for accurate, stable, and robust estimation. 
We verify that by increasing the size of the fusion frame, a certain robustness to noise is also achieved. 
While the computational complexity remains relatively low, we achieve stronger recovery performance compared to usual single-device compressed sensing systems. 
We finally show how our techniques can be applied in various signal processing tasks such as Doppler signal denoising, natural scene scanning and reconstruction, and MR Image reconstruction. 
In these examples, we empirically verify that visually good reconstruction are obtained, even in highly undersampled and noisy regimes. 
\end{abstract}

\maketitle


{\bf Keywords}
Compressed sensing; Fusion frame; Sparse signal approximation

\section{Problem statement}
In a traditional sampling and reconstruction system, the sensors are designed so that the recovery of the signal(s) of interest is possible. 
For instance, when considering the sparse recovery problem, one tries to find the sparsest solution $\hat{\xbf} \in \Kbb^N$ from the noisy measurements $\ybf = \Abf \xbf +\ebf \in \Kbb^m$ and $m \ll N$. 
Here $\Kbb$ denotes the field $\Rbb$ or $\Cbb$. 
This is done by solving the mathematical program
\begin{equation} \label{eq:generalCS}
    \tag{$\lsp^0$-min}
    \hat{\xbf} := \argmin\|\zbf\|_0, \quad \text{ subject to } \|\Abf \zbf - \ybf\|_2 \leq \eta. 
\end{equation}
This problem is NP-Hard and usually only approximately solved, for instance by solving its convex relaxation, known as the Basis Pursuit denoising
\begin{equation}
    \label{eq:bpdn}
    \tag{BPDN}
    \hat{\xbf} := \argmin\|\zbf\|_1, \quad \text{ subject to } \|\Abf \zbf - \ybf\|_2 \leq \eta. 
\end{equation}
It is known that for a given \emph{complexity} (measured by the sparsity in our context) $s$ of the signal $\xbf$, the number of random subgaussian linear measurements needs to grow as $m \gtrsim s \log(N/s)$ for $\widehat{\xbf}$ to be a good enough approximation to $\xbf$. 
Said differently, if the design of a sensor can be made at will, then knowing the complexity of the signal, here characterized by the sparsity, is sufficient for a stable and robust recovery. 
This paper looks at the problem of sampling and reconstructing potentially highly non-sparse signals when the quality of the sensors is constrained. 
We emphasize in passing that throughout this paper, the sought after signal will always be considered of high-complexity. 
We use the sparsity or density as a measure of complexity, but one could consider other models. 
We investigate problems where the number of measurements $m$ cannot be chosen based on the complexity of the signals to recover. 
In the context described above, one would have a limit on the sparsity of the vectors that can be recovered by $s \lesssim m/\log(N/m)$. 
These constraints can be due to many reasons such as cost -- e.g. using 10 sensors at a coarser resolution is cheaper than one at the finest --, frequency rate -- sensors at 2000 THz might not exist for a while --, legal regulation -- e.g. in nuclear medicine where one should not expose a patient to too high radiations at once. 
Problems arise when the signals being sampled are \emph{too dense} for the usual mathematical theories. 
When one thinks about compressed sensing, the size of the sensor required is driven by a certain measure of complexity of the signals considered. 
Allowing for the recovery of signals with higher level of complexity entails the use of better sensors. 
Here we look at the problem differently:  first,  we assume  constraints on the sensor design which are fixed due to some outside reasons.  Under these assumptions, we take on the following challenge: split  the information carried by the signal   in   a clever way so that a mathematical recovery is possible. 

This paper revisits the theory of fusion frames and applies it to the dense signal recovery problem. 
We show that by using advanced mathematical techniques stemming from applied harmonic analysis, it is possible to handle very high complexity signals in an efficient and stable manner.  
Before we dig into the more technical details, we present some real-world scenarios where our framework appears useful, if not essential. 

\subsection{Examples}\label{IntroExamplesSubSection}

\subsubsection{Unavailability of high quality observation devices}
\label{sssec:sampling}
A typical time-invariant bounded linear operator is always represented by a circulant matrix $\Abf$.  So suppose $\Abf$ represents a sensing device whose number of rows $m$ is physically limited by the sampling rate (or resolution) of the device.  
Moreover, consider that the sparsity of the sampled signal $\xbf$ is substantially larger than what a single observation by $\Abf$ could handle/recover by various compressed sensing techniques.

In this context, the limitations on the sensing devices combined with the (potentially) high number of non-zeros in the signals makes it impossible for a state-of-the-art algorithm to recover the unknown signal $\xbf$. As illustrated in Figure~\ref{fig_1}, we suggest to apply $n$ such devices in parallel after prefiltering. The fused compressed sensing technique introduced later allows to resolve the problem that otherwise a single device can not! 
Such scenarios actually exist and show the necessity of the fused compressed sensing technique presented below.

 For example,   suppose an application requires a sensing device of  capacity $X$, described by a sensing matrix $\tilde \Abf$. 
In case such a device is either very expensive, or not available, we may choose to combine $n$ parallel projections   $\{P_j\}_{j=1}^{n}$ prior to measuring, and use $n$ low-quality  sensing devices of capacity $\frac{1}{n}X$, each described by the sensing matrix $\Abf$. 
The (sparse) signal $\xbf$ is then subsequently recovered by various techniques 
via each channel and, through the theory of fusion frames~\cite{Casazza2004framessubspaces,Casazza2008ff,Cahill2012nonOrthFF}, merged into a single vector.
As long as $\{P_j\}_{j=1}^{n}$ are projections - or any filtering operations - with the property that $C\Ibf\le \correction{\sum_j P_j^* P_j}\le D\Ibf $ for some $0<C\le D<\infty$, such a fusion operation is always possible.

\begin{figure}[!htb]
\centering
\includegraphics[width=3in,height=1.5in]{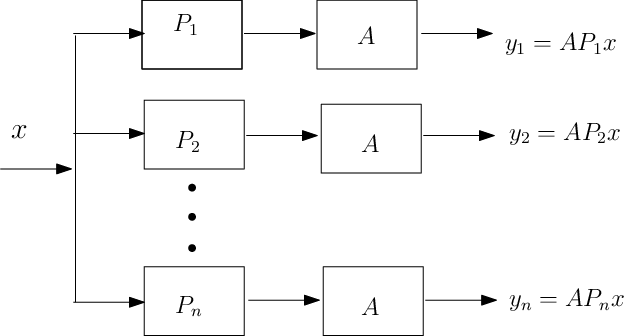}
\caption{Sparse linear array geometry}
\label{fig_1}
\end{figure}


In order for the fusion operator to be bounded away from $0$, in our work we  assume that the pre-filtering/projections $P_i$ 
 are \emph{complete}  in the sense that there won't be ``holes'' in the signal coverage for a whole class of signals; more precisely, we require that the collection of projections satisfy the fusion frame inequality \eqref{eq:ffInequality}.

Evidently, if we design such a fused compressed sensing technique, by enabling the subdivision of sparsity of $\xbf$ into individual subspaces, the sparse  recovery problem becomes a feasible one and can be resolved by multiple sensing devices with low quality or resolution, which are not only widely available but also economical.

The value of fused compressed sensing techniques presented in this paper is clearly reflected in this situation where an otherwise too expensive or impossible problem can now be resolved by using a number of lower resolution/sampling rate devices and by making a \emph{reasonable} number of observations, and processed by the fusion frame theory.

\subsubsection{SAR imaging and spatial filtering}
\label{ssec:examples:SAR}
In the Synthetic Aperture Radar (SAR) imaging process, a flying carrier (an airplane or a satellite) emits a sequence of radio waves to the field of observation (and then detects the reflections by the objects in the field). Each radio wave is sent through an antenna with a fixed aperture/size which physically presents an antenna beam (magnetic field) pattern, say $\{F_i\}$.

\begin{figure}[!htb]
\centering
\includegraphics[width=3in,height=2.5in]{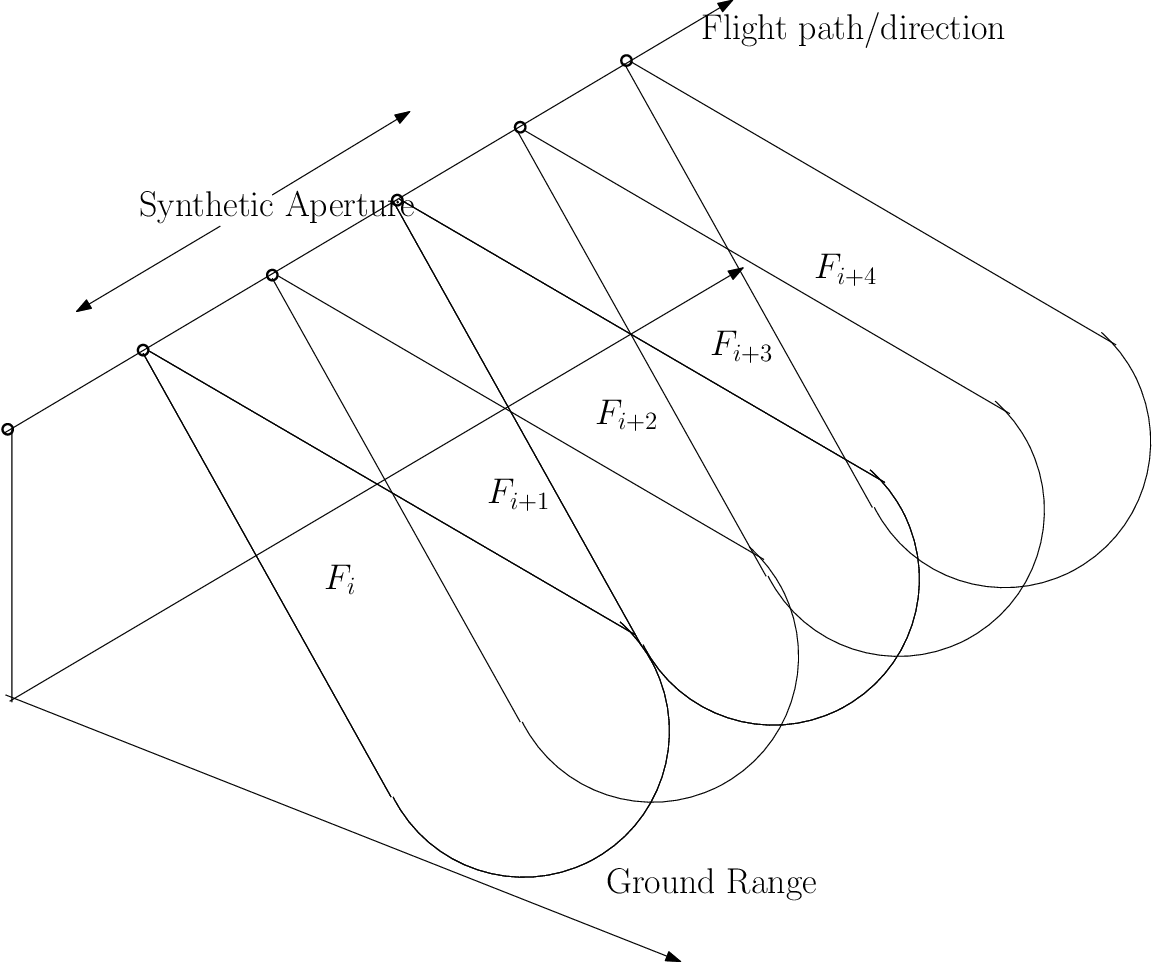}
\caption{Synthetic Aperture Imaging beam and data collection mechanism}
\label{fig_2}
\end{figure}

Consequently, as seen in Figure~\ref{fig_2}, during the $i^{th}$ data collection, the antenna beam pattern $F_i$ physically implements a spatial filtering operator.
And, naturally, between adjacent $i$'s (even among several adjacent $i$'s), the spatial antenna beams $F_i$'s have intentional overlaps.
The actual data observation, for each $i$, is modeled by $\ybf^{(i)} = \Abf F_i \xbf$, where $\xbf$ is the field image reflection coefficients that eventually form the SAR image of the field, and $\Abf$ is a fixed observation matrix determined by the SAR imaging mechanism, which is typically chirping (linearly changing frequencies in time) functions in two different dimensions \cite{Book_SARimaging_Curlander,Book_SARimaging_Franceschetti,Book_SARimaging_Cumming}.
This formulation has a natural ``distributed sparsity'' due to the spatial filtering operations inherent to the physical beams.

Here, to consider the spatial filtering effect, and to avoid the conventional ad-hoc ``alignment'' in the flight direction, the fused SAR imaging process needs to be considered \cite{Zeng2015fusedSAR}.
The new fused SAR imaging formulation is not only the should-be rigorous mathematical formulation of the SAR imaging process, but also potentially beneficial to image resolution and to the robustness of the (often) turbulent data collection process.

\subsection{Contributions}

In this  paper we combine mathematical tools from compressed sensing and fusion frame theory to break the limitations on the signal complexity induced by traditional recovery methods.

Our general approach can be described as the two following steps, or summarized in Algorithm~\ref{algo:framework}.
\begin{enumerate}
\item[1)] Estimate the local information $\widehat{\xbf^{(i)}}$, for $1 \leq i \leq n$ by any sparse recovery method, and
\item[2)] Approximate the fused solution $\widehat{\xbf} = S^{-1}\left(\sum_{i=1}^n \widehat{\xbf^{(i)}}\right)$, either directly or by virtue of the frame algorithm.
\end{enumerate}

\begin{algorithm}[H]
 \KwData{$\left(P_i\right)_{1 \leq i \leq n}$, $n$ projections such that $C\Ibf \leq \sum \correction{\|P_i\|_{2\to 2}^2} \leq D\Ibf$, an estimation of the measurement noise $\eta$}
 \KwResult{An estimation $\hat{{\bf x}}$ of the vector ${\bf x}$}
$k \leftarrow 1$\;
$\widehat{\xbf} = 0$\;
 \While{$k \leq n$}{
  Measure $k^{\text{th}}$ vector $\correction{\ybf^{(k)}} \leftarrow \Abf P_k \xbf + \ebf^{(k)}$\;
	Add local information $\widehat{\xbf} \leftarrow \widehat{\xbf} + \operatorname{argmin}\|\zbf\|_1$ s.t. $\|\Abf P_k \zbf - \ybf^{(k)}\|_2 \leq \eta$\;
	$k \leftarrow k+1$
 }
Fusion: $\widehat{\xbf} \leftarrow S^{-1}\widehat{\xbf}$
 \caption{Fused distributed sensing recovery framework}
\label{algo:framework}
\end{algorithm}

Note that the fusion operation either requires the inverse frame operator $S^{-1}$ or can be approximated using the (fusion-) frame algorithm, described at the end of Subsection~\ref{sssec:frames}.
An important aspect of the framework is that it can either be used sequentially or in a parallel manner. 
In the sequential approach, we get the local pieces of information one after the other. 
In this case, as described in Algorithm~\ref{algo:framework}, we have no need to save the local measurements and can update our guess in an online fashion. 
In the parallel approach, the local measurements are processed independently of each other, and the fusion is done by a central unit, once all the local information have been collected.

In particular, Proposition~\ref{prop:recoveryfixrank} later in the manuscript shows that the robustness is preserved independently from the number of subspaces considered, which can be phrased in simple terms as:

\begin{aprop}
Let $\Abf \in \Kbb^{m \times N}$, and let $\Wcal = (W_i,P_i)_{i=1}^n$, $n \geq 1$ be a fusion frame.
Assume the measurement vectors $\ybf^{(i)} = \Abf P_i \xbf + \ebf^{(i)}$ are corrupted by some (adversarial) independent noise uniformly bounded $\|\ebf^{(i)}\|_2 \leq \eta$. 
Then the solution $\widehat{\xbf}$ generated by Algorithm~\ref{algo:framework} satisfies
\begin{equation}
\|\xbf - \widehat{\xbf}\|_2^2 \leq K \eta^2.
\end{equation}
 \end{aprop}
In the more precise statement (Proposition~\ref{prop:recoveryfixrank}), we specify precisely the value of the constant $K = 1/C$ where $C$ denotes the lower frame bound. 
As a special case of Theorem~\ref{thm:fCS_RIPresults}, we proved the following
\begin{atheorem}
Let $\Wcal = (W_i,P_i)_{i=1}^n$ be a fusion frame system for $\Kbb^N$ with frame operator $S$ and frame bounds $0 < C \leq D < \infty$. 
Let $\Abf \in \Kbb^{m \times N}$ be a matrix satisfying some Partial RIP condition (see Definition~\ref{def:p-rip}) . 
Any $s$ distributed-sparse vector $\xbf$ whose sparsity is uniformly distributed among the $n$ subspaces (i.e. $s_i = s/n$ is the sparsity of $P_i\xbf$) can be recovered as 
\begin{equation*}
    \widehat{\xbf} = S^{-1}\left(\sum_{i=1}^n\widehat{\xbf^{(i)}}\right),
\end{equation*}
where the $\widehat{\xbf^{(i)}}$ are obtained as solutions to the $n$ local sparse recovery problems
\begin{equation*}
    \widehat{\xbf^{(i)}} := \argmin\|\zbf\|_0, \quad \text{ subject to } \|\Abf P_i \zbf - \ybf^{(i)}\|_2 \leq \eta_i.
\end{equation*}
Moreover, assuming the errors to be uniformly bounded by $\eta$, and that the RIP constants $\delta_1 = \delta_2 = \cdots = \delta_n = \delta < 4/\sqrt{41}$ are the same in all the subspaces, the solution approximate the true vector $\xbf$ in the following sense:
\begin{equation*}
 \|\widehat{\xbf} - \xbf\|_2 \leq \frac{n}{C}\Ccal \eta
\end{equation*}
where $\eta := \max_{i}\eta_i$ and $\Ccal \leq \frac{960\sqrt{2}}{\left( 16-41\delta^2 \right)^2}$ is a constant depending only on $\delta$. 
\end{atheorem}
Note that under the usual compressed sensing assumptions, this allows to have a number of measurements per sampled vector that decay linearly with the number of sensors. 

The paper is articulated as follows.
We review 
basics from compressed sensing and fusion frames in Section~\ref{GeneralToolssection}. 
In Section~\ref{motivatingExamplesSection} we explore several theoretical examples to motivate our research,  and to emphasize the importance of our findings. 
Our claims are empirically verified by numerical results.
In the flair of traditional compressed sensing, we extend the standard compressed sensing results 
to the case of recovery where we make explicit use of local redundancy in the fusion frame decomposition in Section~\ref{ShidongSection}. 
Finally, Section~\ref{sec:extensions} derives a mathematical theory allowing to work with fusion frames where sparsity is exploited along a subspace decomposition. 
These findings extend the traditional compressed sensing problem~\cite{foucart2013book} and show some similarities with recent developments in parallel acquisition~\cite{adcock2016CSparallel} and structured sensing~\cite{boyer2015compressed}.

\section{General tools and models}\label{GeneralToolssection}

\subsection{(Traditional) Compressed sensing}
  Compressed sensing  (CS) relies (see \cite{foucart2013book, fornas2011} and references therein) on the inherent sparsity of natural signals $\xbf$ for their recovery from seemingly few measurements $\ybf = \Abf\xbf + \ebf$ given some linear measurement matrix $\Abf \in \Kbb^{m \times N}$, with $m \ll N$. Here the vector $\ebf \in \Kbb^m$ contains the noise and is usually assumed to be bounded, $\|\ebf\|_2 \leq \eta$. With the sparsity assumption, CS aims at finding (approximate) solutions to the~\eqref{eq:generalCS} problem. 
General recovery guarantees ensure that the recovery is stable and robust, that is, the solution satisfies the following approximation bounds
\begin{align}
\|\widehat{\xbf} - \xbf\|_2 \leq \frac{C}{\sqrt{s}}\sigma_s(\xbf)_1+D\eta, \\
\|\widehat{\xbf} - \xbf\|_1 \leq C\sigma_s(\xbf)_1+D\sqrt{s}\eta,
\end{align}

where $\sigma_s(\xbf)_1 := \correction{\inf}_{\zbf: \|\zbf\|_0 \leq s}\|\xbf - \zbf\|_1$ defines the best $s$-term approximation of $\xbf$ in the $\lsp^1$ norm.

Plethora of conditions (some of which we look closer in Section~\ref{sec:extensions}) on $\Abf$ have been derived to ensure that the previous (or similar) estimates hold. These are based on restricted isometry constants ($\delta_{2s} < \sqrt{2}-1$~\cite{candes2008rip} or $\delta_{2s} < 4/\sqrt{41}$~\cite{foucart2013book}, $\delta_s < 1/3$), null space properties ($\Abf$ fulfills the robust and stable $\operatorname{NSP}(s,\rho,\tau)$ if $\|\vbf_S\|_1 \leq \rho \|\vbf_{\overline{S}}\|_1 + \tau \|\Abf \vbf\|_2$ for any vector $\vbf \in \Kbb^N$, and $S$ any index set with $|S| \leq s$), coherence ($\mu_1(s)+\mu_1(s-1) < 1$), quotient property, and so on.
Similar conditions and bounds can be found when greedy and thresholding algorithms are used to approximate~\eqref{eq:generalCS} (see~\cite{blumensath2009iht,foucart2011htp,bouchot2013ghtp,bouchot2014generalized,zhang2011omprip,needell2009cosamp} for some relatively recent results in this direction).

However, to ensure such recovery bounds, there is still a need for building adequate sensing matrices $\Abf$. 
So far, no deterministic matrices $\Abf$ can be built with reasonable numbers of rows (i.e. with decent, small enough, number of measurements) and we have to rely on randomness to build measurement matrices. 
It has been shown, for instance, that matrices with independent random subgaussian entries fulfill the $\operatorname{RIP}(s,\delta)$ provided that the number of measurements scales as $m \asymp \delta^{-2} s \log(N/s)$~\cite{baraniuk2008RIP}.
Similar results have been obtained (with different $\log$ factors) for structured random matrices or matrices from bounded orthonormal systems, see~\cite{rauhut2010structMat} for instance.

\subsection{Frames, fusion frames, and distributed signal processing}
\label{sssec:frames}
 By definition, a sequence $\Fcal = \{f_i\}_{i \in I}$ in a Hilbert space $\Hcal$ is a frame \cite{christensen2003framebook} for $\Hcal$ if there exist $0< A \leq B < \infty$ (lower and upper frame bounds) such that
  \begin{equation}
  \label{frameineq}
  A\Vert f \Vert^2 \leq \sum_{i \in I} \vert \langle f, f_i \rangle \vert^2 \leq B\Vert f \Vert^2 \; \text{for all} \;  f \in \Hcal.
  \end{equation}
 The representation space associated with $\Fcal$ is $\lsp^2(I)$ and its analysis and synthesis operators are respectively given by
 \[T(f) = \{\langle f, f_i \rangle \}_{i \in I}   \; \text{and} \;   T^*(\{c_i\}_{i \in I}) = \sum_{i \in I} c_i f_i,   \]  {for all} $f \in \Hcal$ and  $\{c_i\}_{i \in I} \in \lsp^2(I)$. It shows that the frame operator $S := T^*T$ is   a positive, self-adjoint and invertible operator; this means that recovery of any $f \in \Hcal$ is possible,  if $Sf$ is known; however, computing $S^{-1}$ can be computationally challenging.
Luckily, each frame $\Fcal$ is accompanied by at least  one so-called dual frame $  \Gcal= \{g_i\}_{i \in I}$,  which satisfies 
 \begin{equation}
   \label{framerepr}
    f = \sum_{i \in I}  \langle f, f_i \rangle  g_i = \sum_{i \in I}  \langle f, g_i \rangle  f_i \; \text{for all} \;  f \in \Hcal, 
   \end{equation} and ensures recovery of the function $f$. Whenever a frame is $A$-tight ($A=B$), the problem of function reconstruction is simplified, since the frame operator $S$ in this case is  a scalar multiple of the identity operator.

 {\bf {Fusion frames}}
have been  initially created \cite{Casazza2004framessubspaces} to model the setting of a wireless sensor network. Sensor networks are composed of wireless sensors with constraints in their processing power and transmission bandwidth, which reduces the costs but also affects the precision of the system. The sensors are distributed over a significantly large area of interest, to measure, for instance, pollution,  temperature, sound,   pressure, motion etc. The network is redundant, i.e., there is no orthogonality among sensors, so each sensor can be interpreted  as a frame element.     In addition, a large sensor network is split into redundant sub-networks; the local measurements within each sub-network are sent to a local sub-station, which submits  the gathered information further  to a central processing station for final reconstruction.

Every (local) sensor is represented by a single frame vector; that is, each sub-network is related to a frame for a subspace in a Hilbert space. 
The subspaces have to satisfy a certain overlapping property, which ensures that the overlaps are not too large.   
The reconstruction in such a system is done in two steps: first, within each subspace the conventional frame reconstruction is employed; then, the local pieces of information serve as the inputs for the fusion frame reconstruction, which reconstructs the initial signal completely.
  \begin{adef}[Fusion frames]
	\label{def:fusionframes}
  Given an  index set $I$, let $\Wcal: = \{W_i \, | \, i \in I \}$ be a family of closed subspaces in $\Hcal$. We denote the orthogonal projections\footnote{ an operator $P$ is an orthogonal projection if $P^2 = P$ and  $P^* = P $ } onto $W_i$ by $P_i$. Then  $\Wcal$ is a fusion frame, if there exist  $C, D >0$ such that
 \begin{equation}
\label{eq:ffInequality}
 C \Vert f \Vert^2 \leq  \sum_{i \in I}    \Vert P_i(f)  \Vert^2  \leq  D\Vert f \Vert^2  \; \text{for all} \; f \in \Hcal. 
\end{equation}
 \end{adef}
\begin{rmk}
	Fusion frames are often accompanied by  respective weights.
In the weighted case, the frame condition reads $C \Vert f \Vert^2 \leq  \sum_{i \in I}  v_i^2  \Vert P_i(f)  \Vert^2  \leq  D\Vert f \Vert^2  \; \text{for all} \; f \in \Hcal$ and some positive weights $(v_i)_{i \in I}$. 
	All (unweighted) results derived below apply mutatis mutandis to the weighted case.
\end{rmk}
	
 Given  a fusion frame  $\Wcal$  for a Hilbert space  $\Hcal$,   let  $\Fcal_{i}:=\{f_{ij} \, | \, j \in J_i \}$  be a frame for $W_i$, $i \in I$. Then   $\{ \(W_i,  \Fcal_i \) | i \in I \}$ is a fusion frames system for $\Hcal$. 
In fact fusion frames and frames are not unrelated. 
 It is known~\cite{Casazza2004framessubspaces} that the following statements are equivalent. 
%
 \begin{itemize}
 \item $\cup_{i \in I} \{f_{ij} \, | \, j \in J_i \}$ is a frame for $\Hcal$.
 \item $\{ W_i\}_{i \in I}$ is a fusion  frame for $\Hcal$.
 \end{itemize} 
In particular, if $\{ \( W_i, \Fcal_i \) \}_{i \in I}$ is a fusion  frame system for $\Hcal$ with frame bounds $C$ and $D$, then $\cup_{i \in I} \{f_{ij} \, | \, j \in J_i \}$ is a frame for $\Hcal$ with frame bounds $AC$ and $BD$.
 In the fusion frame theory, an input signal is represented by a collection of vector coefficients that represent the projection onto each subspace. 
The representation space used in this setting is
 \[  \( \sum_{i \in I} \oplus W_i \)_{\lsp^2} =
 \{   { \{ f_i \}_{i \in I} \, | \, f_i \in W_i } \} \]
 with $ \{\Vert f_i \Vert \} _{i \in I} \in \lsp^2(I)$.

  The analysis operator $T$ is then  defined by
 $$T(f) := \{ P_i(f)\}_{i \in I} \; \text{for all} \; f \in \Hcal, $$ while its adjoint operator is the synthesis operator $T^*: \( \sum_{i \in I} \oplus W_i \)_{\lsp^2} \to \Hcal$, defined by
 $$T^* (f) = \sum_{i} f_i, \; \text{where} \;  f = \{f_i\}_{i \in I} \in  \( \sum_{i \in I} \oplus W_i \)_{\lsp^2}. $$
 The   fusion frame operator $S = T^* T$   is given by
  $$S(f) =   \sum_{i \in I} P_i(f).$$ It is easy to verify that $S$ is a positive and invertible operator on $\Hcal$. In  particular, it holds $C \Ibf \leq S \leq D \Ibf$, with $\Ibf$ denoting the identity.

  If the dual frames $\Gcal_i$, $i \in I$, for each local frames are known then the fusion frame operator can be expressed in terms of the local (dual) frames \cite{Casazza2008ff}:
  \[ S = \sum_{i \in I} T^*_{\Gcal_i} T_{\Fcal_i} = \sum_{i \in I} T^*_{  \Fcal_i} T_{\Gcal_i}.\]

 For computational needs,
we may only consider the fusion frame operator in finite frame settings, where the fusion frame operator  becomes a sum of   matrices of each subspace frame operator. The evaluation of the fusion frame operator $S$  and its inverse $S^{-1}$ in finite frame settings are conveniently  straightforward.  By $F_i$ we denote  the frame matrices formed by the frame vectors from $\Fcal_i$, in a column-by-column format. Let $G_i$ be defined in the same way from the dual frame $\{ g_{ij}\}_{j \in J_i}$. Then the fusion frame operator is
$$  S(f) =  \sum_{i \in I} F_iG_i ^T f=  \sum_{i \in I}  G_iF_i ^T f. $$
 Hence a distributed fusion processing is feasible in an elegant way, since the reconstruction formula for all $f \in \Hcal$ is
 \[
  f = \sum_{i \in I}  S^{-1}  P_i(f). \]

 The standard distributed fusion procedure uses the local projections of each subspace. In this procedure, the local reconstruction takes place first in each subspace $W_i$, and the inverse fusion frame is applied to each local reconstruction and combined together:
 \begin{equation}
 \label{firstapproach}
 f = \sum_{i \in I}  S^{-1} P_i(f) =  \sum_{i \in I} S^{-1} \(\sum_{j \in J_i} \langle f,f_{ij} \rangle g_{ij} \) \; \text{for all} \;  f \in \Hcal.
\end{equation}
 Alternatively, one may use a reconstruction procedure acting globally, which is possible if the coefficients of signal/function decompositions are available:
\begin{equation}
 \label{secondapproach}  f =  \sum_{i \in I} \sum_{j \in J_i} \langle f,f_{ij} \rangle  (S^{-1}  g_{ij}) \; \text{for all} \;  f \in \Hcal.
 \end{equation}   The difference in   procedure \eqref{secondapproach}, compared with a global frame reconstruction lies in the fact that the (global) dual frame $\{ S^{-1}  g_{ij}\}$ is first calculated at the local level, and then fused into the global dual frame by applying the inverse fusion frame operator. This potentially makes the evaluation of (global) duals much more efficient.

As stated above, computing the inverse (fusion) frame operator is often a challenging task. 
Instead, one can approximate the solution $\widehat{\xbf}$ by employing the so-called frame algorithm. 
All we need to start the iterative algorithm is $S\widehat{\xbf}$, which we already have. We recall the relevant result   below for completeness. 

\begin{aprop}\label{CasLiIterative}\cite{Casazza2008ff}
       Let $(W_i)_{i\in I}$ be a fusion frame in $\Kbb^N$, with fusion frame operator $S=S_W$, 
and fusion frame bounds $C, D$. Further, let $\xbf \in  \Kbb^N$, and define the sequence $(\xbf_k)$
by
$\xbf_0 = 0$ and $\xbf_k = \xbf_{k-1} + \frac{2}{C+D}S(\xbf - \xbf_{k-1})$,  $k \geq  1$. Then we have $\xbf = \lim_{k\rightarrow \infty} \xbf_k$, with the
error estimate $\Vert \xbf - \xbf_{k}  \Vert \leq \left( \frac{D-C}{D+C}\right)^k \Vert \xbf  \Vert.  $
\end{aprop}

Concretely, using the fusion frame reconstruction, the updates read
\begin{equation*}
\xbf_k = \xbf_{k-1} + \frac{2}{C+D}\sum_{i=1}^n\xbf^{(i)} - \frac{2}{C+D}\sum_{i=1}^nP_i \xbf_{k-1}.
\end{equation*}
The middle term is computed once and for all. 
The following updates then only require some basic matrix-vector multiplications. 
Note that starting the algorithm with $\widehat{\xbf}_0 = 0$ yields $\widehat{\xbf}_1 = \frac{2}{C+D}S\widehat{\xbf}$.

Similarly, in Section~\ref{ShidongSection}, we have  $\widehat{\xbf}= S^{-1}\left(\sum_j \widehat{\xbf^{(j)}}\right)$, and  $\widehat{f}=D \widehat{\xbf}$. 
We can find $\widehat{\xbf}$ via the iterative approach of Proposition~\ref{CasLiIterative}, by starting the algorithm with $\widehat{\xbf_0 }= 0$, $\widehat{\xbf_1} = \frac{2}{C+D}S(\widehat{\xbf})=\frac{2}{C+D}\sum_j S \widehat{\xbf^{(j)}}$, and so on. 
The $n-$th term approximation of $\widehat{f}$ is then estimated with $\widehat{f}_n=D \widehat{\xbf}_n$.

\subsection{Signal recovery in fusion frames}
This work describes an approach for sensing and reconstructing signals in a fusion frame structure. 
As presented above, given some local information $\xbf^{(i)} := P_{i}(\xbf)$, for $1 \leq i \leq n$, a vector can easily be reconstructed by applying the inverse fusion frame operator
\begin{equation}
\label{eq:reconstructionFormula}
\xbf := S^{-1}S(\xbf) = S^{-1}\sum_{i=1}^n P_{i}(\xbf) = S^{-1}\(\sum_{i = 1}^n \xbf^{(i)}\), \text{ where $\xbf^{(i)} := P_{i}(\xbf)$, for $1 \leq i \leq n$}.
\end{equation}
Throughout the work, we assume that the projected vectors are sampled independently from one another, with $n$ devices modeled by the same sensing matrix $\Abf$. 
Formally, the problem is as follows \emph{From the measurements $\ybf^{(i)} = \Abf \xbf^{(i)} + \ebf^{(i)}$, reconstruct an estimation $\widehat{\xbf^{(i)}}$ of the local information to compute the approximation $\widehat{\xbf}$ of the signal $\xbf$}.

From this point on, there are two ways of thinking about the problem. 
In a first scenario, the signals are measured in the subspace and the recovery of the $\xbf^{(i)}$ are done locally. 
In other words, it accounts for solving $n$~\eqref{eq:generalCS} problems (or their approximations via~\eqref{eq:bpdn} for instance) in the subspaces, and then transmitting the estimated local signals to a central unit taking care of the fusion via Equation~\eqref{eq:reconstructionFormula}.
The other approach consists of transmitting the local observations $\ybf^{(i)} = \Abf P_{i}\xbf$ to a central processing station which takes care of the whole reconstruction process. 
In this case, the vector $\xbf$ can be recovered by solving a unique \eqref{eq:generalCS} problem directly with, letting $\Ibf_n$ denote the $n$ dimensional identity matrix,
\begin{equation}
\label{eq:centralReconstruction}
\ybf = \( \begin{array}{c} \ybf^{(1)} \\ \vdots \\ \ybf^{(n)} \end{array}\) = \Ibf_n \otimes \Abf \left[ \begin{array}{c} P_{1} \\ \vdots \\ P_{n} \end{array} \right] \xbf
\end{equation}
While the latter case is interesting on its own (see for instance~\cite{adcock2016CSparallel}), we investigate here some results for the first case. 
Our results can be investigated and generalized further, integrating ideas  where the measurements matrix (here, the sensors) vary locally, as is the case in~\cite{adcock2016CSparallel} or driven with some structured acquisition (see for instance~\cite{boyer2015compressed}).

We would like to put our work in context. 
This paper is not the first one to describe the use of fusion frames in sparse signal recovery. 
However, it is inherently different from previous works in~\cite{boufounos2011csfs,ayaz2014ff}. 
In~\cite{boufounos2011csfs} the authors provide a framework for recovering sparse fusion frame coefficients. 
In other words, given a fusion frame system $\{ \(W_i, P_i \) \}_{i=1}^n$ a vector is represented on this fusion frame as a set of $n$ vectors $(\xbf^{(i)})_{i=1}^n$ where each of the $\xbf^{(i)}$ corresponds to the coefficient vector in subspace $W_i$.
The main idea of the authors is that the original signal may only lie in few of the $n$ subspaces, implying that most of the $\xbf^{(i)}$ should be $0$. To rephrase the problem, we can say that the vector has to lie in a sparse subset of the original fusion frame.
While~\cite{boufounos2011csfs} is concerned with some recovery guarantees under (fusion-) RIP and average case analysis, the paper~\cite{ayaz2014ff} gives uniform results for subgaussian measurements and derive results on the minimum number of vector measurements required for robust and stable recovery.

We, on the other hand, exploit structures (or sparsity) locally\footnote{This sparsity assumption is needed only when a given frame component has a large dimension. We hope that in most of the cases we can control this dimension to be small enough avoiding the necessity of sparse assumptions. See the following section for more details.}. 
We do not ask that only a few of the subspaces be active for a given signal, but that the signal only has a few active components per subspace. 
This justifies the use of the local properties of fusion frame systems.

\section{Examples and applications}
\label{motivatingExamplesSection}
This section introduces some theoretical examples where the application of our fusion frame-based recovery shows increased performance over traditional recovery techniques. 
In particular, we justify in Proposition~\ref{prop:recoveryfixrank} that the fusion frame recovery may provide better robustness against noise.

\subsection{Orthogonal canonical projections}

Consider the problem of recovering $\xbf \in \Kbb^N$ from the measurements:
\begin{equation*}
\ybf^{(i)} = \Abf_i \xbf + \ebf^{(i)}, \quad 1 \leq i \leq n,
\end{equation*}
where $\Abf_i$'s are defined as $\Abf_i = \Abf P_i$ and $P_i$ is the orthogonal projection onto $\Omega_i$, with $\Omega_i \subseteq \{1,\cdots, N\}$.
We assume that the noise vectors $\ebf^{(i)}$ are uncorrelated and independent, with $\|\ebf^{(i)}\|_2 \leq \eta_i$. 
$\Abf$ corresponds to a matrix of linear measurements in $\Kbb^{m \times N}$.
Note that this framework is an example of recovery from Multiple Measurement Vectors (\emph{MMV}) (see~\cite{duarte2005distributed,gribonval2008atoms,rauhut2008csRdeundant} and references therein) where the signal is the same in every measurements,
and the matrices are different for each set of measurements. 
We consider the problem of recovering local vectors $\xbf^{(i)}$, for $1 \leq i \leq n$, where the $\xbf^{(i)}$ are defined as $\xbf^{(i)} := P_i\xbf$.
In this case, we have that the sparsity of each local signal is at most $N_i := \operatorname{rank}(P_i) = |\Omega_i|$.
Once the local pieces of information $\widehat{\xbf^{(i)}}$ are recovered, the original signal $\xbf$ can be estimated from a \emph{fusion frame like} reconstruction:
\begin{equation}
 \label{eq:ffreconstructionfromProj}
 \widehat{\xbf} = S^{-1}\( \sum_{i=1}^n \widehat{\xbf^{(i)}} \) = \sum_{i=1}^n S^{-1}\widehat{\xbf^{(i)}},
\end{equation}
where the fusion frame operator $S$ is defined in the usual way as $S: \Kbb^N \to \Kbb^N$, $S(\xbf) = \sum_{i = 1}^n P_i(\xbf)$.
The problem is to recover the unknown vector $\xbf$ from the measurements $\{\ybf^{(i)}\}_{i=1}^{n}$ by solving local problems.
Clearly, a necessary condition for uniform recovery is that we have $\bigcup_{i=1}^n \Omega_i = \{1,\cdots,N\}$ in a deterministic setting, or with high probability in a probabilistic setting.

The main idea of our approach is to solve the (very) high-dimensional, high-complexity, and demanding problem $\ybf = A \xbf + \ebf$ by combining results obtained from $n$  problems that are much easier to solve.
We investigate first the case where the orthogonal projections do not overlap, i.e. $\Omega_i \cap \Omega_j = \varnothing$, for any $i,j \in \{1,\cdots,n\}$ with $i \neq j$ and $\bigcup_{i=1}^n \Omega_i = \{1,\cdots, N\}$. 
The results developed are reminiscent of some work on partial sparse recovery, when part of the support is already known~\cite{bandeira2013partialNSP}.
In a second step we analyze the use of random and deterministic projections that may overlap for the recovery of the signal $\xbf$.

\begin{rmk}
\label{rmk:nbOfVariables}
It is important to pause here and understand the problem above. 
Knowing the support of the projections, one could work on the local vectors $\xbf^{(i)}$ of smaller sizes $N_i < N$. 
Doing this would allow faster computations of the solutions but also has the implication that the sparsity (per subspace) has to remain small. 
On the other hand, considering large vectors with many zero components allows to keep a sparsity rather large (compared to the subspace dimension). 
We choose to investigate mainly the latter case as we try to break the limitation on the complexity for a given sensing matrix.
\end{rmk}

\subsubsection{Decomposition via direct sums}
\label{ssec:direcsums}
It is clear that the dimension $N$ of the ambient space will play an important role. Indeed, if we consider no overlap between the projections, we have two drastically different extreme cases.
These scenarios are important as they shed light on the main ideas behind our recovery results.

\underline{Case $n=N$} This case is characterized by (up to a permutation of the indices) $\Omega_i = \{i\}$ for all $1 \leq i \leq N$.
In other words, each (set of) measurements $\ybf^{(i)} = \Abf P_i\xbf + \ebf^{(i)}$ gives information about a single entry $x_i$ of the input vector $\xbf$ via an overdetermined linear system.
In this case, assuming $\Abf$ has full rank, we can compute an estimate $\widehat{x_i}$ of the entry $x_i$ as the solution to the $\lsp^2$ minimization problem:
\begin{equation}
 \label{eq:l2minNproj}
 \widehat{x_i} := \argmin_{x \in \Rbb} \|\ybf^{(i)} - x \abf_i\|_2^2 = \frac{\abf_i^T\ybf^{(i)}}{\abf_i^T\abf_i}
\end{equation}
where $\abf_i$ denotes the $i^{th}$ column of the matrix $\Abf$.
The $N$ independent $\lsp^2$ minimizations ensure that the final solution $\widehat{\xbf}$ satisfies the bound
\begin{equation}
 \label{eq:l2bound}
 \|\xbf - \widehat{\xbf}\|_2 = \sqrt{\sum_{i=1}^N\frac{|\langle \abf_i, \ebf^{(i)}\rangle|^2}{\|\abf_i\|_2^4}} \leq \sqrt{\sum_{i=1}^N \frac{ \|\ebf^{(i)}\|_2^2 }{\|\abf_i\|_2^2} } \leq \sum_{i=1}^N \frac{\|\ebf^{(i)}\|_2}{\|\abf_i\|_2}.
\end{equation}
In particular, in the noiseless scenario ($\ebf^{(i)} =0$, for all $i$), the recovery is exact. In terms of (normalized) matrices from Gaussian measurements, which ensures that $\Abf$ has full rank -- even numerically -- with high-probability~\cite{chen2005condition,rudelson2009smallestSingValue}, it holds $\Ebb\|\abf_i\|_2 = 1$ and it follows $\Ebb\|\xbf - \widehat{\xbf}\|_2 \leq \sum_{i=1}^N \|\ebf^{(i)}\|_2$.

Notice that even though we have only one fixed matrix $\Abf \in \Kbb^{m \times N}$, (obviously) we are no longer limited to the recovery of sparse vectors, which is the motivation for using fusion frames as described below.
However, this has a price as the (actual) number of measurements drastically increased to $m \cdot N$.
This is clearly too many measurements in any practical scenarios, but with similar ideas, we can reach practical applications such as using multiple lower sampling rate devices in order to resolve a harder problem, typically required for much higher rate of sampling devices (see Section~\ref{sssec:sampling}). 
Note also that all the results here should be put in the context of MMV problems, and not with the usual CS setup.

\underline{Case $n=2$}
We consider now two sets $\Omega_1$ and $\Omega_2$ such that $\Omega_1 \cup \Omega_2 = \{1,\cdots, N\}$ and $\Omega_1 \cap \Omega_2 = \varnothing$. 
Without loss of generality we can assume that $\Omega_1 = \{1, \cdots, N_1\}$ and $\Omega_2 = \{N_1+1,\cdots,N_1+N_2\}$ with $N_1+N_2 = N$.
We measure twice the unknown vector $\xbf$ as $\ybf^{(1)} = \Abf_1 \xbf + \ebf^{(1)}= \Abf P_{1} \xbf + \ebf^{(1)}$ and $\ybf^{(2)} = \Abf_2 \xbf + \ebf^{(2)} = \Abf P_{2} \xbf + \ebf^{(2)}$. If both $N_1$ and $N_2$ are smaller than $m$, then two $\lsp^2$ minimizations recover $\widehat{\xbf^{(1)}}$ and $\widehat{\xbf^{(2)}}$ independently and it follows that the solution $\widehat{\xbf} = \widehat{\xbf^{(1)}}+\widehat{\xbf^{(2)}}$ obeys the following error bound:
\begin{equation}
 \|\widehat{\xbf} - \xbf\|_2^2 \leq \|\Abf_{\Omega_1}^+\ebf^{(1)}\|_2^2 + \|\Abf_{\Omega_2}^+\ebf^{(2)}\|_2^2.
\end{equation}
The reconstruction is again perfect in the noiseless case.
Moreover, as in the previous case, there is no need for a sparsity assumption on the original vector $\xbf$.
The oversampling ratio is not very large as we only have a total of $2m$ measurements.

If however one at least of $N_1$ or $N_2$ (say $N_1$) is larger than $m$ then we need to use other tools as we are now solving an underdetermined linear system.
Driven by ideas from CS, we can assume the vector to be sparse on $\Omega_1$. The recovery problem becomes
\begin{equation}
 \label{eq:l1CS}
 \text{Find } \widehat{\xbf^{(1)}}	\text{ that minimizes } \|\zbf\|_1, \quad
			\text{ subject to } \|\Abf_1 \zbf - \ybf^{(1)}\|_2 \leq \eta_1 
\end{equation}

The $\lsp^1$ minimization problem is introduced as a convex relaxation of the NP-hard $\lsp^0$ minimization.
Note that the constraints apply only on the support $\Omega_1$ of the unknown vector.
Hence the usual sparsity requirements encountered in the CS literature need not to apply to the whole vector.
Unfortunately, the sparsity assumption being applied independently on $\Omega_1$ and/or $\Omega_2$ restricts ourselves to non-uniform recovery guarantees only, at least when considering the full set of sparse signals.
The following definition will become handy in the analysis.
\begin{adef}[Partial null space property]
\label{def:prnsp}
A matrix $\Abf \in \Kbb^{m \times N}$ is said to satisfy a partial robust and stable null space property of order $s$ with respect to a subset $\Omega \subset \{1, \cdots, N\}$ with $N \geq |\Omega | > m$, $P$ the orthogonal projection onto $\Omega$, and constants $0 < \rho < 1$, $0 < \tau$ if 
\begin{equation}
\|(P\vbf)_S\|_1 \leq \rho \|(P\vbf)_{\bar{S}}\|_1 + \tau \|\Abf\vbf\|_2
\end{equation}
holds for any vector $\vbf$ and any set $S \subset \Omega$ such that $|S| \leq s$.
\end{adef}
\begin{rmk}
The previous definition is a weakening of the usual robust null space property, where what happens in the complement of the set $\Omega$ is irrelevant.
It is worth noticing that it also coincides with the usual definition when $\Omega = \{1, \cdots, N\}$.
\end{rmk}
The partial robust null space property of the measurement matrix $\Abf$ on $\Omega_1$ ensures that the recovered local vector $\widehat{\xbf^{(1)}}$ obeys the following error bound~\cite[Theorem 4.19]{foucart2013book}:
\begin{equation}
\label{eq:errorBoundMixedLargeSmallRanks}
\|\widehat{\xbf^{(1)}} - P_1\xbf\|_1 \leq \frac{2(1+\rho)}{1-\rho}\sigma_{s}(\xbf^{(1)})_1 + \frac{4\tau}{1-\rho}\|\ebf^{(1)}\|_2.
\end{equation}
Equation~\eqref{eq:errorBoundMixedLargeSmallRanks} is obtained by modifying the proofs of~\cite[Theorem 4.19, Lemma 4.15]{foucart2013book} and adapting them to the presence of the projection $P_1$. 
A formal proof of a more general statement is given in the proof of the later Theorem~\ref{thm:rdnspRecovery}.
This yields the following direct consequence
\begin{aprop}
 \label{prop:recovery2proj}
 Let $N$ and $m$ be positive integers with $N > m$ and $\Abf \in \Kbb^{m \times N}$. 
There exist two integers $N_1$ and $N_2$ such that, up to a permutation, $\Omega_1 = \{1, \cdots, N_1\}$ and $\Omega_2 = \{N_1+1, \cdots N_1+N_2\}$.
 Assume that $N_2 \leq m$ and $N_1 > m$.
 Assume that the matrix $\Abf$ satisfies the partial robust null space property of order $s$ with respect to $\Omega_1$ with constants $0 < \rho < 1$ and $\tau > 0$ and that $\Abf_{\Omega_2}$ has full rank.
 For $\xbf \in \Kbb^N$, let $\widehat{\xbf} := \widehat{\xbf^{(1)}} + \widehat{\xbf^{(2)}}$ with $\widehat{\xbf^{(1)}}$ solution to Problem~\eqref{eq:l1CS} and $\widehat{\xbf^{(2)}}$ solution to the overdetermined $\lsp^2$ minimization problem on $\Omega_2$. Then the solution obeys:
 \begin{equation}
  \label{eq:boundOverUnderRecovery}
  \|\xbf - \widehat{\xbf}\|_2 \leq \|\Abf_{\Omega_2}^+\ebf^{(2)}\|_2 + \frac{2(1+\rho)}{1-\rho}\sigma_{s}(\xbf^{(1)})_1 + \frac{4\tau}{1-\rho}\|\ebf^{(1)}\|_2.
 \end{equation}
 Moreover, the total number of measurements amounts to $2m$ for the recovery of an $N_2+s$ sparse vector.
\end{aprop}
Recovering the vector $\xbf^{(2)}$ does not create any problem, as long as the pseudo-inverse has a reasonable $\lsp^2$ norm, i.e. as soon as the rank of the projection is reasonably small (see discussion below and Theorem~\ref{thm:minSingValue} in particular).
The recovery of the vector $\xbf^{(1)} \in \Kbb^{N}$ is ensured provided the number of subgaussian measurements $m$ scales as~\cite[Corollary 9.34]{foucart2013book}
\begin{equation}
 m \gtrsim 2s \ln(eN/s)\( 1+\rho^{-1} \)^2.
\end{equation}
Proposition~\ref{prop:recovery2proj} slightly reformulated yields the following result:
\begin{acorol}
 Let $\Abf := \Abf P_{1} + \Abf P_{2} \in \Kbb^{m \times N}$ with $\Omega_1 \cap \Omega_2 = \varnothing$ and $\Omega_1 \cup \Omega_2 = \{1, \cdots, N\}$ and $N_1 = |\Omega_1|, N_2 = |\Omega_2|$.
 If $\Abf$ is a subgaussian matrix and $\Abf_{\Omega_2}$ has full rank, then provided that
 \begin{equation*}
  m \gtrsim 2s \ln(eN/s)\( 1+\rho^{-1} \)^2,
 \end{equation*}
 any vector $\xbf \in \Theta := \Sigma_{s}^{N_1} + \Kbb^{N_2} \subset \Sigma_{s+N_2}^N$ can be recovered with the bound~\eqref{eq:boundOverUnderRecovery}.
\end{acorol}
Here, $\Sigma_s^{N_1}$ denotes the set of $s$ sparse vectors in $\Kbb^{N_1}$.
The previous Corollary provides a uniform recovery result, with a proviso that we restrict the model for the vector $\xbf$.

\subsubsection{Controlled deterministic projections}
\label{ssec:controlledDetProj}

We look here at a scenario where the sets $\Omega_i$ (and hence the projections $P_i$) are completely and deterministically controlled.
We consider an integer $n > 0$ and subsets $\Omega_i \subset \{1,\cdots,N\}$ with $|\Omega_i| = N_i$, $\cup_{i=1}^n \Omega_i = \{1,\cdots,N\}$ and $N_i \leq m$, and assume again that $\Abf$ has full rank.
In this case, we can always enforce the disjointness of the support $\Omega_i \cap \Omega_j = \varnothing$, for any $i \neq j$.
This yields the trivial recovery of the input signal from its local information vectors:
\begin{equation}
 \widehat{\xbf} := \sum_{i = 1}^n \widehat{\xbf^{(i)}}.
\end{equation}
Here, all the local vectors are recovered via generalized Moore-Penrose pseudo inverses of the sub-matrices $\Abf_{\Omega_i}$ composed of only the columns supported on $\Omega_i$:
\begin{equation}
 \widehat{\xbf^{(i)}} = \Abf_{\Omega_i}^+ \ybf^{(i)} = \xbf^{(i)} + \Abf_{\Omega_i}^+ \ebf^{(i)}.
\end{equation}
The error estimate follows directly:
\begin{equation}
 \label{eq:errorControlledDeterministic}
 \|\xbf - \widehat{\xbf}\|_2^2 \leq \sum_{i = 1}^n \|\Abf_{\Omega_i}^+ \ebf^{(i)}\|_2^2
\end{equation}

Note that it holds $\|\Abf^+\ebf\|_2 \leq \|\Abf^+\|_{2 \to 2} \|\ebf\|_2$ and that $\|\Abf^+\|_{2 \to 2} = 1/\operatorname{min}_{1 \leq i \leq r}(\sigma_i)$, with $r$ the rank of $\Abf$ and $\{\sigma_i\}_{i=1}^r$ its singular values.
As a consequence, when dealing with ``nice'' matrices, the bound~\eqref{eq:errorControlledDeterministic} is reasonable.
In a CS setup, normalized sensing matrices generated at random from a Gaussian distribution have the property of being \emph{well-behaved}, as suggested by~\cite[Theorem 9.26]{foucart2013book}:
\begin{atheorem}
\label{thm:minSingValue}
Let $\widetilde{\Abf}$ be an $m \times s$ Gaussian matrix with $m > s$ and $\Abf := \frac{1}{\sqrt{m}}\widetilde{\Abf}$ be its variance-normalized counterpart. 
Then for $t > 0$,
 \begin{align*}
  \Pbb\(\sigma_\text{max}(\Abf) \geq 1 + \sqrt{s/m} + t\) &\leq e^{-\frac{mt^2}{2}} \\
  \Pbb\(\sigma_\text{min}(\Abf) \leq 1 - \sqrt{s/m} - t\) &\leq e^{-\frac{mt^2}{2}},
 \end{align*}
where $\sigma_{max}(\Theta)$ and $\sigma_{min}(\Theta)$ are the maximum and the minimum singular values of $\Theta$, respectively.
\end{atheorem}
For small projection ranks, and $\Abf$ as in Theorem~\ref{thm:minSingValue}, it holds $s \leq r < m$, and therefore $\Pbb(\|\Abf_{\Omega_i}^+\|_{2 \to 2} \geq \frac{1}{1-\sqrt{\frac{r}{m}} - t}) \leq e^{-mt^2/2}$ which justifies that the bound~\eqref{eq:errorControlledDeterministic} is small.

\subsubsection{Rank-controlled projections}
\label{ssec:rankRdmSupport}
This scenario differs from the previous one by the fact that we may control only the ranks $N_i$'s of the projections but let the support of projections be random.
This example is motivated by the SAR applications where the rank of the projections is controlled by the sensing device itself. 
In this case, for uniform recovery of any vector $\xbf$ we need to ensure that, with high probability, the whole support $\{1,\cdots,N\}$ is covered by the random projections.
We assume that all the projections have the same rank $r$ for simplicity of calculations, but similar ideas apply if variations in the rank of the projections were needed.
We pick uniformly at random $n$ sets $\Omega_i$ of size $r$ in $\{1,\cdots,N\}$. It holds
\begin{equation}
\label{eq:probaElementNotInOmega} 
\Pbb\left[ \exists i \in \{1,\cdots,N\}: i \notin \Omega := \cup_{j = 1}^n \Omega_j \right] = N \Pbb\left[ 1 \notin \Omega \right] = N\(\Pbb\left[ 1 \notin \Omega_1 \right] \)^n 
 = N\( \frac{N-r}{N} \)^n.
\end{equation}
This gives that $\Pbb\left[ \exists i \in \{1,\cdots,N\}: i \notin \Omega := \cup_{j = 1}^n \Omega_j \right] \leq \varepsilon$ whenever
\begin{equation}
 \label{eq:nbMeasVectors}
 n \geq \frac{\log(N/\varepsilon)}{\log(N/(N-r))}.
\end{equation}

This direct consequence follows.
\begin{aprop}
 \label{prop:recoveryfixrank}
 Let $\Abf \in \Kbb^{m \times N}$ with $m < N$, $r \leq m$, $\varepsilon > 0$, and $n$ be such that Equation~\eqref{eq:nbMeasVectors} holds. 
In addition, assume that every submatrix with $r$ columns extracted from $\Abf$ has full rank. 
Let $\Omega_1, \Omega_2, \cdots, \Omega_n$ be $n$ subsets in $\{1,\cdots,N\}$ of size $r$ chosen independently and uniformly at random.
 Any vector $\xbf$ recovered from the fusion of local recoveries as $\widehat{\xbf} := S^{-1} \sum_{i=1}^n \widehat{\xbf^{(i)}}$ from the local measurements $\ybf^{(i)} = \Abf P_i\xbf + \ebf^{(i)}$, $1 \leq i \leq n$ satisfies, with probability at least $1-\varepsilon$ 
 \begin{equation}
  \label{eq:errorBoundProp2} \| \xbf - \widehat{\xbf} \|_2^2 \leq \frac{1}{C} \sum_{i=1}^n \|\Abf_{\Omega_i}^+ \ebf^{(i)}\|_2^2
 \end{equation}
 where $S(\xbf) := \sum_{i=1}^n P_i \xbf$ and $\widehat{\xbf^{(i)}} := A_{\Omega_i}^+ \ybf^{(i)}$ and $C$ denotes the lower frame bound of $S$.
\end{aprop}
The lower frame bound is defined for general fusion frame operators as in Definition~\ref{def:fusionframes}. 
The error bound above is applicable in case of any fusion/filtering scenario. 
In our particular set-up, we calculate exactly this bound.
\begin{alemma}
Given $n$ projections $P_i$ of rank $r$, the lower frame bound of the fusion frame operator $S := \sum P_i$ is given by
\begin{equation}
\label{eq:lowerFrameBound}
C := \min_{1 \leq k \leq N} \Mcal(k) =: \Mcal,
\end{equation}
where $\Mcal(k)$ is the multiplicity of index $k$. 
Formally, let $\Gamma_k := \{ j: 1\leq j \leq n: k \in \Omega_j \}$, then $\Mcal(k) := |\Gamma_k|$.
\end{alemma}

\begin{proof}
That $C \leq \Mcal$ is clear. 
Indeed, let $\xbf = \ebf_{k_0}$ where $k_0$ is one of the indices achieving the minimum and $\ebf_{k_0}$ denotes the canonical vector. 
Then it holds $S(\ebf_{k_0}) = \sum_{j \in \Gamma_{k_0}} \ebf_{k_0}$ and it follows that $\|S(\xbf)\|_2^2 \leq \Mcal^2$.

On the other hand, for any $\xbf \in \Kbb^N$, the fusion frame operator can be rewritten as $S(\xbf) = \sum_{i=1}^N \Mcal(i)x_i$. 
Consequently it holds that for any $\xbf \in \Kbb^N, \|S(\xbf)\|_2^2 \geq \Mcal^2 \sum_{i=1}^N|x_i|^2$ and hence $C^2 \geq \Mcal^2$.
\end{proof}

As a consequence, it is trivial to see that for any $1 \leq k \leq N$, $\Mcal(k)$ is a non-decreasing function of the number of projections $n$. Hence, the greater the number of projections (the redundancy, in fusion frame terms) the greater the bound $C$ and therefore the smaller the error in the recovery, according to Equation~\eqref{eq:errorBoundProp2}. 
This however can only hold true, as long as the noise per measurement vectors remains small.

In the case of subsets selected independently uniformly at random, we can have a precise statement. 
The lower frame bound is defined as the minimum number of occurrences of any index $k \in \{1,\cdots, N\}$. 
Let us consider $\Mcal(k)$ and $\Gamma_k$ as in the previous lemma. 
Let $k \in \{1,\cdots,N\}$ be any index and $\Omega_j$ denote a draw of a random set. 
It holds $\Pbb[k \in \Omega_j] = r/N$. 
The subsets being independent of each other, $\Mcal(k)$ is a binomial random variable with probability of success $P = r/N$ for each of the $n$ trials. 
It follows that $\Pbb[\Mcal(k) = l] = {n \choose l} P^l(1-P)^{n-l}$ for all $1 \leq k \leq N$ and $1 \leq l \leq n$.

Putting everything together, we get that $C := \min \limits_{k \in \{1,\cdots, N\}}\Mcal(k)$ is a random variable such that 
\begin{align}
\Pbb[C \geq l] &= \Pbb[\forall 1 \leq k \leq N, \Mcal(k) \geq l] \\ 
   &= \( \Pbb[\Mcal(k) \geq l] \)^N = \( 1 - \sum_{j=1}^{l-1}\Pbb[\Mcal(k) = j] \)^N = \( 1 - F(l-1,n,P) \)^N.
\end{align} 

These expressions resemble the calculations involved in Equations~\eqref{eq:probaElementNotInOmega} and ~\eqref{eq:nbMeasVectors}, however with a much more complicated probability distribution. 
To avoid unnecessarily tedious calculations that would infer the readability of the paper, we choose to leave the following result as a conjecture.
\begin{aconj}
\label{conj:expectedLowerFramebound}
Let $\Omega_1, \Omega_2, \cdots, \Omega_n$ be $n$ subsets of $r$ elements taken uniformly at random in $\{1,\cdots,N\}$. 
Let $\Gamma_k := \{ j: 1\leq j \leq n: k \in \Omega_j \}$, and $\Mcal(k) := |\Gamma_k|$.
Then $C := \min_{1 \leq k \leq N} \Mcal(k) =: \Mcal$ grows at least linearly with $n$, i.e., there exists constants $c > 0$ and $b \geq 0$, independent of $n$, such that 
\begin{equation}
\label{eq:linearGrowthC}
\Ebb[C] \geq  cn - b.
\end{equation}
\end{aconj}
A direct consequence of this is that one can always find a multiplicative constant $c' \leq c$ such that for a number of projections $n \geq b/(c-c')$, $\Ebb[C] \geq  c'n$ (the case $c = c'$ can only be true for $b = 0$, in which case, there is no need for this remark).
Though the proof is not given, the result is backed up with numerical results shown in Figure~\ref{fig:lowerFrameBounds}.\footnote{It is also backed up with particular cases, when the ranks are exactly half the dimension, see http://math.stackexchange.com/questions/1135253/mean-value-of-minimum-of-binomial-variables} 
The graphs always show the expected linear growth, independently of the parameters of the problem. 

\begin{figure}[htb]
\centering
\subfigure[Large projections]{
\includegraphics[width=0.35\linewidth]{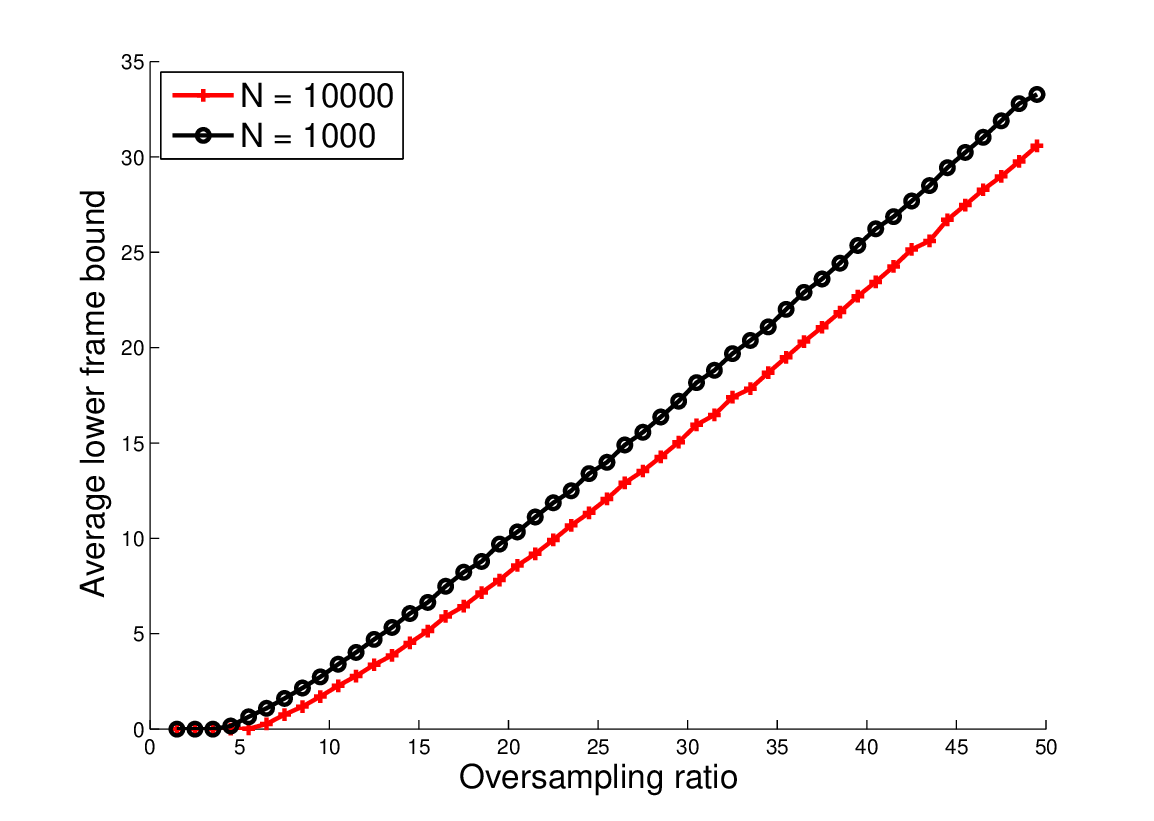}
\label{fig:lowFrameLargeProj}
}
\subfigure[Small projections]{
\includegraphics[width=0.35\linewidth]{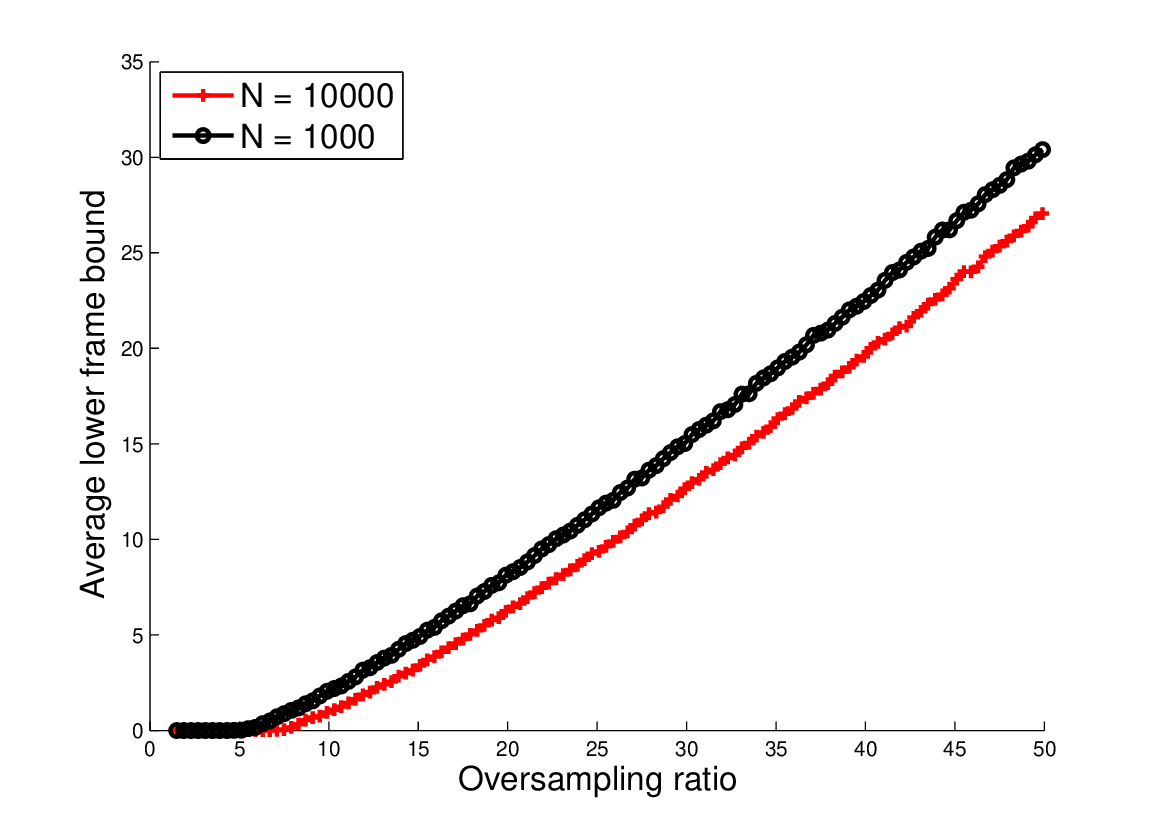}
\label{fig:lowFrameSmallProj}
}
\caption{Linear growth of the expectation of the lower frame bound for various use cases. The averages are calculated from 300 random experiments. The left graphs show the average lower frame bound when dealing with rather large projections (with rank $r = N/2$) while the right graphs are in the case of much smaller projections ($r = N/5$). The oversampling ratio is calculated as a linear function of the number of projections and is independent of the ambient dimension $N$; $\zeta = rn/N$.}
\label{fig:lowerFrameBounds}
\end{figure}

%
%
Assuming the conjecture to be true, Equation~\eqref{eq:errorBoundProp2} then simplifies to 
\begin{equation}
\Ebb\|\widehat{\xbf} - \xbf\|_2^2 \leq \frac{1}{c'n}\sum_{i=1}^n\|\Abf_{\Omega_i}^+\ebf^{(i)}\|_2^2 \leq \frac{\delta \nu}{c'},
\end{equation}
where the expectation is taken over the draw of the projections, for a certain constant $c' > 0$, where we let $\delta = \max_i\|\Abf_{\Omega_i}^+\|_{2 \to 2}$ and $\nu := \max_i\|\ebf^{(i)}\|_2^2$.
It shows that in expectation, the error should not grow as the number of measurements increases. 
Moreover, numerical results illustrated in Figure~\ref{fig:nbProjections} suggest that the expected reconstruction error decreases. 
An intuition for this behavior is provided in a simple example in Equations~\eqref{eq:simpleCalculations1} and ~\eqref{eq:simpleCalculations2}.
From a compressed sensing perspective, it has been shown that the least favorable case appears as the case when the vectors are repeated again and again. 
This translates in our scenario to the case where the increase of the measurements consist solely in repeating the same projections (which would not produce a fusion frame). 
Considering the example of a SAR imaging process developed in Section~\ref{ssec:examples:SAR}, this would correspond to a plane flying over a given region multiple times, with the spatial filtering being exactly the same.



\begin{rmk}
 It is important to note here that there is no sparsity assumption. It is possible to recover {\bf any} vector $\xbf$ with high probability, as long as the number of measurement vector scales reasonably with the dimension of the input space.
\end{rmk}

\subsubsection{Constrained number of measurements and ranks of projections}
In certain scenarios, the physical measurement devices limit our freedom in choosing the ranks and size of the measurements.
We hence deal here with fixed ranks $N_i$'s for the projections $P_i$.
Assume, without loss of generality, that the first $0 \leq k \leq n$ projections have ranks lower than the number of measurements $m$.
For these subspaces, the usual $\lsp^2$ minimization procedure yields perfect (or optimal in the noisy case) recovery of the local information $\xbf^{(i)}$, for $1 \leq i \leq k$.
For the remaining $n-k$ subspaces, the pseudo-inverse is not sufficient, and we use tools from  CS.
Here again we use the partial robust null space property from Definition~\ref{def:prnsp} on every subspace $\Omega_j$, for $k+1 \leq j \leq n$.
Combining results from Propositions~\ref{prop:recoveryfixrank} and~\ref{prop:recovery2proj} yields the following corollary as a consequence:
\begin{acorol}
\label{cor:constrainedRecovery}
Let $m < N$ and $A \in \Kbb^{m \times N}$. 
Assume that we are given a set of $n$ sets $\{\Omega_i\}_{i=1}^n$ in $\{1,\cdots,N\}$ of respective sizes $\{N_i\}_{i=1}^n$ with $N_1 \leq N_2 \leq \cdots N_n$.
Assume in addition that $\cup_{i}\Omega_i = \{1,\cdots,N\}$. 
Moreover, there exists a unique $0\leq k \leq n$ such that $N_k \leq m$ and $N_{k+1} > m$ (with the convention that $N_0 = 0$ and $N_{n+1} > m$). 
Assume that the submatrices $\Abf_{\Omega_i}$, for $1 \leq i \leq k$, have full rank.
If in addition the measurement matrix satisfies a partial robust null space property of order $s$ with respect to every subset $\Omega_i$, $i \geq k+1$, then, the approximation $\widehat{\xbf}$ defined as $\widehat{\xbf}:= S^{-1}\sum_{i=1}^n \widehat{\xbf^{(i)}}$ with $\widehat{\xbf^{(i)}} := \Abf_{\Omega_i}^+ \ybf^{(i)}$, for $1 \leq i \leq k$ and $\widehat{\xbf^{(i)}}$ solutions to the $\lsp^2$ constrained $\lsp^1$ minimization problems~\eqref{eq:bpdn}, for $k +1 \leq i \leq n$ fulfills the following estimate:
\begin{equation}
\|\xbf - \widehat{\xbf}\|_2 \leq \frac{1}{\Mcal}\( \sum_{i=1}^k \| \Abf_{\Omega_i}^+ \ebf^{(i)} \|_2 + \frac{2(1+\rho)}{1-\rho}\sum_{i = k+1}^n \sigma_s(\xbf^{(i)})_1 + \frac{4\tau}{1-\rho}\sum_{i=k+1}^n\|\ebf^{(i)}\|_2 \)
\end{equation}
\end{acorol}
For simplicity we have assumed that the NSP is valid on every subspaces for the same set of parameters. 
To avoid notational encumbrance we do not write results where the $\rho$, $\tau$, and $s$ may depend on the subspace considered. 

\begin{rmk}
This is a direct application of the previous results. 
It is obtained by recovering every local pieces of information independently from one another.
It is possible to improve these estimates by a sequential procedure where we first estimate the $x_j$ for $j \in \Omega_i$, for some $1 \leq i \leq k$, and then using this reliable estimate to improve the accuracy of the $\lsp^1$ minimization program.
\end{rmk}

A final remark considers the case where the rank of the projection is also random. We can think of the sets $\Omega_i$, for $1 \leq i \leq n$ as Binomial random variables with probability of success $p$. In this case, the rank of the projection $P_i$ is controlled by the expectation of this random variable. An analysis similar to the previous one can be carried over to ensure recovery of any vector with high probability.

\subsection{Numerical examples}
\label{ssec:numerics}
This Section describes numerical recovery results on the case described above. 
We first show some particular examples of recovery, when dealing with dense signals and show that we can break the traditional sparsity limit by adding a few sensors. 
The next example shows the behavior of the quality of the reconstruction when slowly adding sensors. 
Finally, the last subsection validates our Conjecture and show that the noise tends to decrease while adding projections, and that the overall quality of recovery scales linearly with the noise level. 
In another contribution, we have also verified that we can recover a very dense Fourier spectrum, by using our approach with Fourier measurement matrices, see~\cite{ABL2017Conm}. 
\subsubsection{Examples of recovery}

\begin{figure}[htb]
\centering
\subfigure[Large ranks]{
\includegraphics[width=0.30\linewidth]{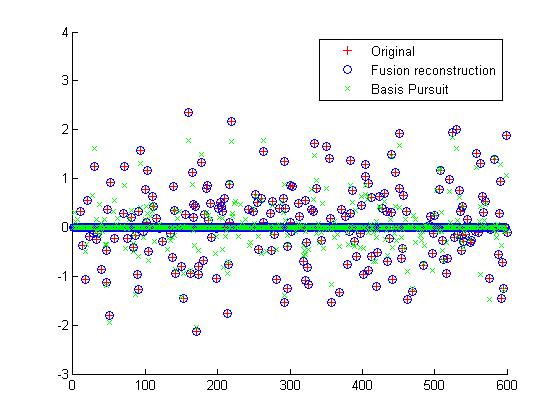}
\label{fig:lRSS}
}
\subfigure[Small ranks, small sparsity]{
\includegraphics[width=0.30\linewidth]{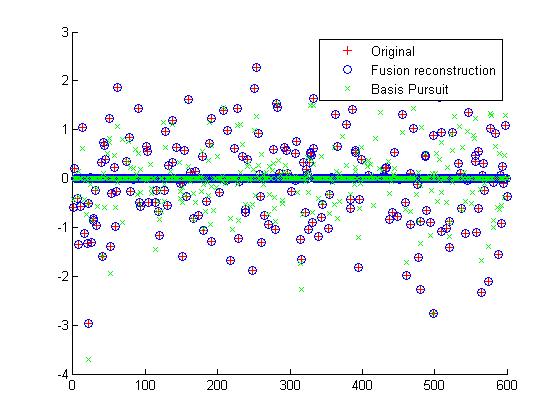}
\label{fig:sRSS}
}
\subfigure[Small ranks large sparsity]{
\includegraphics[width=0.30\linewidth]{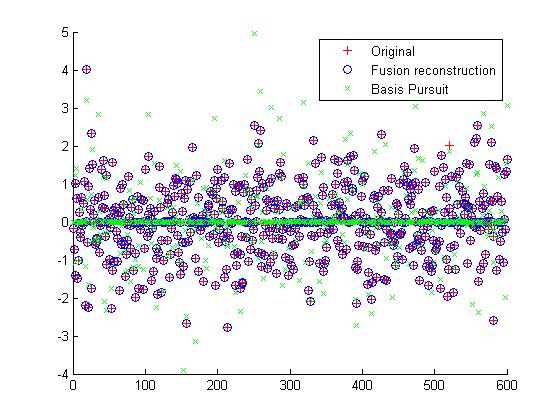}
\label{fig:sRLS}
}
\caption{Examples of fusion reconstruction.}
\label{fig:examples}
\end{figure}

The examples of reconstructions depicted on Figures~\ref{fig:examples} were created by generating $s$-sparse Gaussian vectors at random with $s=200$ for Figures~\ref{fig:lRSS} and~\ref{fig:sRSS}, and $s = 500$ for Figure~\ref{fig:sRLS}. 
For each case, the $250 \times 600$ matrix was generated at random with entries i.i.d. from a Gaussian distribution. 
In the first example, the rank is set to $r = 300 > m$ and hence, $\lsp^1$ minimization (traditional basis pursuit) is used on every subspaces. 
This yields a total of $13$ projections generated at random.
For the two last examples, the ranks of the projections are set to $r = 200 < m$. In this case, only a classical $\lsp^2$ inversion is needed on every subspace. 
Here, since the ranks are smaller, the number of projections has to be increased to $22$ in order to ensure that the whole set $\{1,\cdots, N\}$ is covered with high probability. 
By looking carefully at the last example, one can see that an index with non-zero magnitude has not been selected (around index 520). These cases, however, are rare. 
In every figure, the red '+' crosses represent the true signal, the blue circles represent the reconstructed signal from our fusion approach, and the green 'x' correspond to the reconstruction with traditional basis pursuit. 
While the very last set-up (very high number of non-zero components) is clearly not suitable for usual compressed sensing, the advantage of projections and fusions can be seen even in the first two (where the number of non-zeros remains relatively small). 
Another aspect to look at is that when dealing with small ranks, the solutions to the $\lsp^2$ minimization problems are computed efficiently. As a consequence, even when the number of projections increases, the calculation of the recovered $\widehat{\xbf}$ is still orders of magnitude faster. 
\subsubsection{Recovery and number of projections}
\label{ssec:constraintedProj}
Given any scenario introduced above, it is expected from Conjecture \ref{conj:expectedLowerFramebound} that increasing the number of projections (of a given rank) will increase the lower frame bound~\eqref{eq:lowerFrameBound} at least linearly. 
This linear increase of the lower frame bound compensates for the potential increase in the reconstruction error. 
One could however hope for better results according to the numerical evidence illustrated in Figure~\ref{fig:nbProjections}, where the average case error seems to be reduced as the number of projections increases.
Given a set of $n$ projections generating a fusion frame, the least favorable case, studied in~\cite{eldar2010average} for the case of a single filtering operation, when adding an extra $n$ projections, is when the exact same $\Omega_i$ are repeated.
The exact behavior of the reconstruction error with respect to the number of projections is still under research. 

%
\begin{figure}[htb]
\centering
\includegraphics[width=0.75\linewidth]{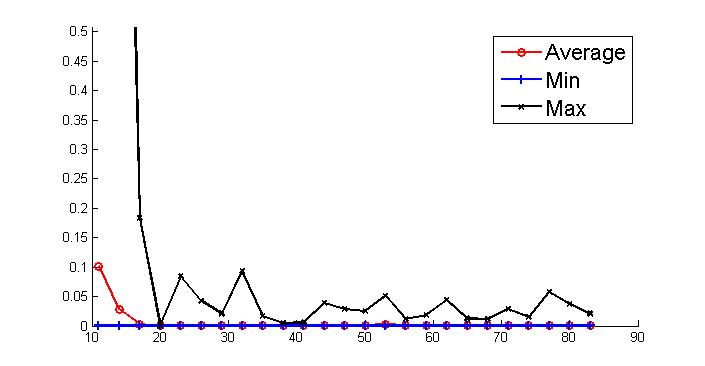}
\caption{$\lsp^2$ error of recovery for various number of projections.}
\label{fig:nbProjections}
\end{figure}
For a total of 10 generations of random Gaussian matrices of size $300 \times 1000$, and $10$ random Gaussian $200$ sparse vectors (note: that is a setting where usual compressed sensing has no recovery guarantees), we let the number of random projections of rank $r = 500$ go from $60\%$ up to $5$ times the minimum given by Equation~\eqref{eq:nbMeasVectors} ($17$ in our setting, for a success of covering the whole input set of $\varepsilon = 99\%$).
As can be seen from the figure, the recovery is unlikely to be correct as long as the number of projections remains small. The trend is generally towards the perfect recovery - though this can only be guarantee with high accuracy, which explains the little spikes. Moreover, while the maximum might be somewhat bigger, it appears that 1) it remains within reasonable bounds (depending on the application) and 2) the average of the error among the $100$ tests per number of projections is very small, suggesting that only a few of the recoveries are off. 
A way to understand this is to consider the case, where we have two projections that are slightly overlapping. Given the global set $\Omega = \{1, \cdots, N\}$, assume we split it into $\Omega_1$ and $\Omega_2$ with some overlap, i.e. $\Omega = \Omega_1 \cup \Omega_2$ and $\Omega_{12} := \Omega_1 \cap \Omega_2 \neq \varnothing$. We independently solve
\begin{align}
\widehat{\xbf^{(i)}} &:= \argmin \|\xbf \|_1 \\ 
							&\text{s.t. } \|\Abf P_i \xbf - \ybf^{(i)}\|_2 \leq \eta_i
\end{align}
for $1 \leq i \leq 2$. For now we consider only the noiseless case, i.e. $\varepsilon_1 = \varepsilon_2 = 0$ which is equivalent to solving two basis pursuit optimization problems. 
We can decompose the solutions recovered as $\widehat{\xbf^{(i)}} = \widehat{\xbf^{(i)}}|_{\Omega_i \backslash \Omega_{12}} + \widehat{\xbf^{(i)}}|_{\Omega_{12}}$. 
If we assume, for simplicity, that the first component is recovered exactly, we have that $\widehat{\xbf^{(1)}} = \xbf^{(1)}$. 
If however the second component is not recovered accurately, we can write $\widehat{\xbf^{(2)}} = \xbf^{(2)} + \ebf^{(2)}$ (for further generalizations, we can always write $\widehat{\xbf^{(i)}} = \xbf^{(i)} + \ebf^{(i)}$, for $1 \leq i \leq n$ with $\ebf^{(i)} = 0$ in case of successful recovery on the set $\Omega_i$). 
In this very simple example, we have $S(\xbf) = \sum_{i=1}^n \xbf^{(i)} = \xbf^{(1)} + \xbf^{(2)}$. 
Similarly, the inverse fusion frame operator can be seen as the point wise empirical average of the evidences (i.e. the local recoveries): 
\begin{equation}
\label{eq:simpleCalculations1}
S^{-1}(\xbf) := \xbf^{(1)}|_{\Omega_1 \backslash \Omega_{12}} + \xbf^{(2)}|_{\Omega_2 \backslash \Omega_{12}} + \frac{\xbf^{(1)}|_{\Omega_{12}} + \xbf^{(2)}|_{\Omega_{12}}}{2}.
\end{equation}
Note that the $2$ in the denominator denotes the number of times that the indices in $\Omega_{12}$ are picked in the creation of the (random) projections.
In other words, increasing the number of projections increases potentially the denominator under a given index and hence reduces the error. 
\begin{equation}
\label{eq:simpleCalculations2}
\widehat{\xbf} - \xbf = S^{-1}\( \widehat{\xbf^{(1)}} + \widehat{\xbf^{(2)}} - \xbf^{(1)} - \xbf^{(2)} \) = \ebf^{(2)}|_{\Omega_2 \backslash \Omega_{12}} + \frac{\ebf^{(2)}|_{\Omega_{12}}}{2}
\end{equation}
As a consequence, the error (on the set $\Omega_{12}$) is bounded by the maximum of the error from every recovery problem divided, as it corresponds to the average error on the domain. 
\subsubsection{Robustness to noise}

\begin{figure}[htb]
\centering
\subfigure[Minimum number of projections]{
\includegraphics[width=0.30\linewidth]{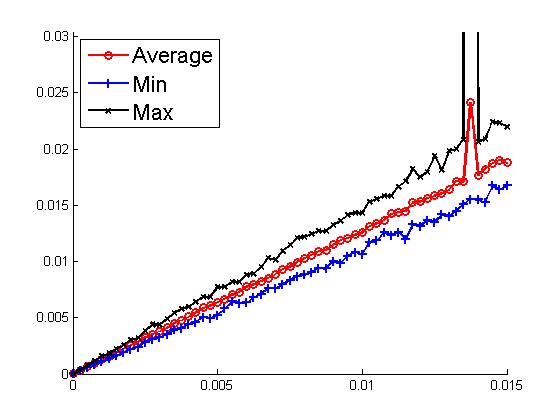}
\label{fig:noise_minProj}
}
\subfigure[Double number of projections]{
\includegraphics[width=0.30\linewidth]{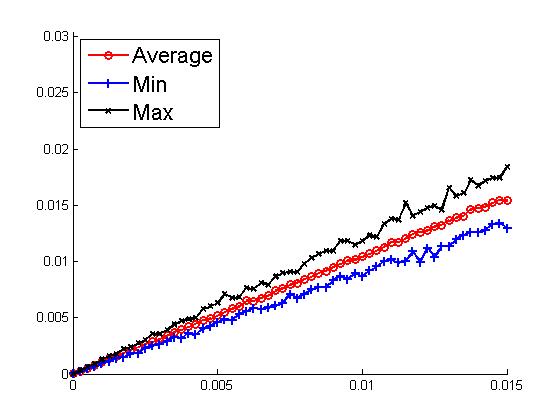}
\label{fig:noise_doubleProj}
}
\subfigure[Basis pursuit denoising]{
\includegraphics[width=0.30\linewidth]{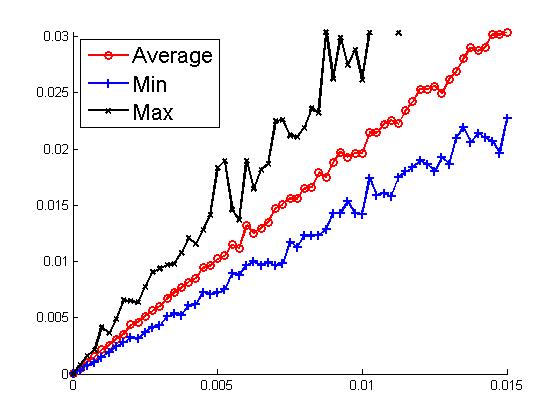}
\label{fig:noise_l1}
}
\caption{Behavior facing noise}
\label{fig:noise}
\end{figure}

In the set of experiments illustrated in Figures~\ref{fig:noise}, we have compared the robustness of our algorithm to traditional basis pursuit denoising when dealing with noise. 
We have set the parameters in a regime where~\eqref{eq:bpdn} works for fair comparison. 
For a $500$ dimensional ambient space and $300$ dimensional measurement space, we generated 5 random $80$-sparse Gaussian vectors for each of the 10 randomly generated Gaussian matrices. 
For every single test, some additive Gaussian noise with a given $\lsp^2$ norm $\theta$ has been added. 
This means that for the case of the distributed approach, $n$ independent noise components of norm $\theta$ have been added in total; one for each channel. 
The figures depict the evolution of the error of the $\lsp^2$ norm when the energy of the noise component increases. 
The black line represents the maximum of the error, the blue one the minimum, and the red one the average over every samples. 
The left graph shows the result when the number of projections is exactly set as in Equation~\eqref{eq:nbMeasVectors} while the second one doubles this number. 
The third figure shows the result when using the usual Basis Pursuit Denoising. 
It is important to notice that the higher peak appearing on the first figure is due to an index with non zero component from the support of the original vector $\xbf$ not being selected at all during the random projections. 
This however does not contradict the high probability of recovery. 
The improvement in the noise behavior from the first figure to the second shows how the fusion frame operator tends to average out the local errors to yield a better estimate (as described in the previous section). 
Finally, all of the algorithms scale linearly with the norm of the noise per measurements (as suggested by Equation~\eqref{eq:errorBoundProp2}) and even if the total noise is increased as the number of projections is increased, the recovery tends to be improved by considering more projections.


\section{Compressed sensing theory with fusion frames}
\label{sec:extensions}

Our results so far ensure that robust and stable recovery of dense signals is possible, by smartly combining local information.
In this section we show  that the ideas developed in the previous Section can also be considered as an extension of the traditional CS.
In particular, it gives solid mathematical foundations to our work. 
It is important to note that similar ideas have been developed in parallel in~\cite{adcock2016CSparallel}. 
There, the authors introduced a model similar to~\eqref{eq:centralReconstruction}, albeit asking that the fusion frame be tight with $C = D= 1$, but reconstruct the whole signal globally, without the use of the fusion process. 
Other authors~\cite{boyer2015compressed} looked at CS with structured acquisition and \correction{structured} sparsity. 
In simple words, they prove that adapting the sampling matrices to some prior knowledge of the sparsity pattern allowed for larger applicability of the CS framework. 

In this section we describe the recovery of a signal $\xbf$ by means of CS in the local subspaces and fusion processing. 
The local pieces of information are computed as solutions to the problems
\begin{equation}
\tag{$\Pcal_{1,\eta}$}
\label{eq:fcsbpdn}
\begin{array}{rl} 
\displaystyle
\min_{\correction{\zbf} \in \Kbb^N} & \|\zbf\|_1 \\ \text{s.t. } & \|\Abf P_i\zbf - \ybf^{(i)}\|_2 \leq \eta_i. \end{array}
\end{equation}

In the noiseless case, the problem is solved by the basis pursuit
\begin{equation}
\tag{$\Pcal_{1,0}$}
\label{eq:fcsbp}
\begin{array}{rl} \min_{\xbf \in \Kbb^N} & \|\zbf\|_1 \\ \text{s.t. } & \Abf P_i\zbf = \ybf^{(i)}. \end{array}
\end{equation}

\begin{alemma}
\label{lemma:suppSolution}
Let $\widehat{\zbf^{(i)}}$ be a solution to the noisy~\eqref{eq:fcsbpdn} or noiseless~\eqref{eq:fcsbp} basis pursuit problem. 
Then $\widehat{\zbf^{(i)}} \in W_i$. 
\end{alemma}
\begin{proof}
Let $\widehat{\zbf{(i)}}$ be a solution and let $\tilde{\zbf} = P_i\widehat{\zbf^{(i)}}$. Then $\|\tilde{\zbf}\|_1 \leq \|\widehat{\zbf^{(i)}}\|_1 + \|(I-P_i)(\widehat{\zbf^{(i)}}))\|_1$. From the optimality of $\widehat{\zbf^{(i)}}$ and the admissibility of $\tilde{\zbf}$, it follows that $\|\tilde{\zbf}\|_1 = \|\widehat{\zbf^{(i)}}\|_1$ and $\widehat{\zbf^{(i)}} = P_i(\widehat{\zbf^{(i)}}))$. 
\end{proof}

\subsection{Signal models and recovery conditions}

\subsubsection{Extension of the sparsity model}
The traditional sparsity model is not appropriate in this setting. 
As an example, let us consider that all the $s$ nonzero components of a vector $\xbf$ fall within a certain subspace (say $W_1$), there is, a priori, no hope to improve the recovery performance compared to a single sensor/subspace problem. 
Indeed, in this case, the recovery is ensured (locally) by CS methods if the number of observations $m$ scales as 
\begin{equation*}
m \asymp s\log(N/s).
\end{equation*}
Since we are dealing here only with an identical sensing model, this yields a total number of observations $m_T$ scaling as
\begin{equation*}
m_T \asymp ns\log(N/s).
\end{equation*}
This may be acceptable if we consider only very few subspaces but may explode in certain cases. 
Therefore the model of distributed sparsity is more appropriate. 
\begin{adef}
A signal $\xbf \in \Kbb^N$ is said to be $\sbf = (s_1, \cdots s_n)$-distributed sparse with respect to a fusion frame $\Wcal = (W_i, P_i)_{1\leq i \leq n}$, if $\|P_{i}(\xbf)\|_0 \leq s_i$, for every $1 \leq i \leq n$. 
$\sbf$ is called the \emph{sparsity pattern} of $\xbf$ with respect to $\Wcal$. 
\end{adef}
We denote by $\Sigma_\sbf^{(\Wcal)}$ the set of all $\sbf$-distributed sparse vectors with respect to the family of subspaces $(W_i)_{i}$. 
We let $s = \|\sbf\|_1$ denote the global sparsity of the vector, with respect to $\Wcal$. 
In the case that the sparsity of the signal is uniformly distributed among the subspaces ($s_i = s/n$), the usual CS recovery guarantees ensure us that
\begin{equation*}
m \asymp s_i \log(N/s_i)
\end{equation*}
observations per subspace are required for a stable and robust recovery of the signal. 
This accounts for a total number of measurements scaling as $m_T \asymp n s_i \log(N/s_i) = s \log(N/s_i)$. 
In other words, we are able to recover similar sparsities as in the classical CS framework, while using only \emph{cheap} sensors.
But in opposition to the classical theory most of the computations can be easily carried in a distributed setting, where only pieces of the information are available. 
Only the fusion process requires all the local information to compute the final estimation of a signal. 

Locally, it only requires solving some very small CS system, which can be done faster than solving the original one. 
This is also the findings found in parallel in~\cite{adcock2016CSparallel}, where it is concluded that the number of measurements per sensor decreases linearly with the number of sensor. 
We describe a similar problem, while looking at it from a different perspective. 
In particular, we try to find the sparsity patterns that may be recovered for a given sensor design. 
The motivation for this problem comes from the applications in SAR imaging where the sensor is given and the same everywhere, and where we may not have any control on the number of observations per subspace. 
As it will become useful later, we also need to introduce the local best approximations.
\begin{adef}
Let $\Wcal = (W_i,P_i)_{i=1}^n$ be a fusion frame, and let $\xbf \in \Kbb^N$. 
For $p > 0$ and a sparsity pattern $\sbf = (s_1, \cdots, s_n)$ with $s_i \in \Nbb$ for all $1 \leq i \leq n$, the $\lsp^p$ errors of best $\sbf$-term approximations are defined as the vector
\begin{equation*}
\sigma_{\sbf}^{\Wcal}(\xbf)_p := \( \sigma_{s_1}(P_{1}\xbf), \sigma_{s_2}(P_{2}\xbf), \cdots, \sigma_{s_n}(P_{n}\xbf) \)^T
\end{equation*}
\end{adef}

\subsubsection{Partial properties}
The \emph{null space property} (NSP) has been used throughout the past decade in the CS literature as a necessary and sufficient condition for the sparse recovery problem via~\eqref{eq:generalCS}. 
A matrix $\Abf$ is said to satisfy the (robust) null space property with parameters $\rho \in (0,1)$ and $\tau > 0$ relative to a set $S \subset \{1, \cdots, N\}$ if 
\begin{equation*}
\|\vbf_S\|_1 \leq \rho\|\vbf_{\overline{S}}\|_1 + \tau \|\Abf \vbf\|_{2}, \quad \text{for all } \vbf \in \Kbb^{N}.
\end{equation*}
More generally, we say that the matrix $\Abf$ satisfies the NSP of order $s$ if it satisfies the NSP relative to all sets $S$ such that $|S| \leq s$. 
We extend here this idea to the context of distributed sparsity with respect to fusion frames, as already mentioned in Definition~\ref{def:prnsp}. 
Here we talk about a sparsity pattern and ask that the NSP property be valid for all local subspaces up to a certain (local) sparsity level. 
\begin{adef}[Robust and stable partial null space property (RP-NSP)]
\label{def:rpnsp}
Let $n$ be an integer and $\Wcal = (W_i, P_i)_{i=1}^n$ be a fusion frame for $\Kbb^N$.
Let $\sbf = (s_1,\cdots,s_n)$ be a sequence of non negative numbers representing the sparsity pattern with respect to $\Wcal$.
For a number $q \geq 1$, a sensing matrix $\Abf \in \Kbb^{m \times N}$ is said to satisfy the $\lsp^q$-RP-NSP with pattern $\sbf$ with respect to $\Wcal$ and with constants $\rho_1, \cdots, \rho_n \in (0,1)$ and $\tau_1, \cdots, \tau_n >0$ if 
\begin{equation*}
\|(P_{i}\vbf)_{S_i}\|_q \leq \frac{\rho_i}{s^{1-1/q}}\|(P_{i}\vbf)_{\overline{S_i}}\|_1 + \tau_i \|\Abf \vbf\|_{2}, \quad \text{for all } \vbf \in \Kbb^N, 1 \leq i \leq n, S_i \subset W_i, \text{ and } |S_i| \leq s_i.
\end{equation*} 
\end{adef}
This definition is reminiscent of the work on sparse recovery with partially known support~\cite{bandeira2013partialNSP}.
The difference here being that there is no need to enforce a condition on the vector $\vbf$ to lie in the range of the other subspaces. 
In a sense, this is taken care of by the fusion process and the fact that we have multiple measurement vectors. 

\begin{rmk}
Note that we could simplify the definition by asking that the parameters be uniform and independent of the local subspace. Namely, introducing $\tau := \max_{1 \leq i \leq n} \tau_i$ and $\rho := \max_{1 \leq i \leq n}\rho_i$, the above definition becomes
\begin{equation*}
\left\|\(P_{W_i}\vbf\)_{S_i}\right\|_q \leq \frac{\rho}{s^{1-1/q}}\left\|\(P_{W_i}\vbf\)_{\overline{S_i}}\right\|_1 + \tau \|\Abf \vbf\|_{2}, \quad \text{for all } \vbf \in \Kbb^N, 1 \leq i \leq n, S_i \subset W_i, \text{ and } |S_i| \leq s_i.
\end{equation*} 
\end{rmk}

A stronger, but easier to verify, condition often used as a sufficient recovery condition is the by-now well known \emph{Restricted Isometry Property} (RIP). 
Informally speaking a matrix is said to satisfy the $RIP(s,\delta)$ if it behaves almost like an isometry (up to a constant $\delta$) on every $s$-sparse vector $\vbf \in \Sigma_s$. 
Formally speaking, $\Abf \in \Kbb^{m \times N}$ satisfies $RIP(s,\delta)$, for some $s \geq 2$ and $\delta \in (0,1)$ if 
\begin{equation}
\label{eq:traditionalRIP}
(1-\delta)\|\vbf\|_2^2 \leq \|\Abf\vbf\|_2^2 \leq (1+\delta)\|\vbf\|_2^2, \quad \text{for every } \vbf \in \Sigma_s.
\end{equation}
The lowest $\delta$ satisfying the inequalities is called the \emph{restricted isometry constant}. 
Once again, we want to derive similar properties on our sensing matrix for the distributed sparse signal model.

\begin{adef}[Partial-RIP (P-RIP)] 
\label{def:p-rip}
Let $\Wcal = (W_i,P_i)_{i=1}^n$  be a fusion frame, and let $\Abf \in \Kbb^{m \times N}$. 
Assume that $\Abf$ satisfies the $RIP(s_i,\delta_i)$ on $W_i$, with $\delta_i \in (0,1)$, $i \in I=\{1,\cdots, n\}$.
Then, we say that $\Abf \in \Kbb^{m \times N}$ satisfies the partial RIP with respect to $\Wcal$, with bounds $\delta_1, \cdots, \delta_n$ and sparsity pattern $\sbf =(s_1, \cdots, s_n)$.
\end{adef}
In other words, $\Abf$ satisfies a P-RIP conditions, if it satisfies RIP-like conditions on every subset of vectors in $\operatorname{range}(P_i)$.

\begin{rmk}
This definition is consistent with the definition of the classical RIP in the sense that the case $n = 1$ (only one projection, one subspace) recovers the usual RIP. 
\end{rmk}

The P-RIP can be written in a form similar to the traditional RIP, Equation~\eqref{eq:traditionalRIP}.

\begin{aprop} 
Let   $\Wcal=(W_i)_{i \in I}$ be a fusion frame  (with frame bounds $0 < C \leq D < \infty$). Let  $\Abf \in \Kbb^{m \times N}$   satisfy the P-RIP with respect to $\Wcal$,  with bounds $\delta_1, \cdots, \delta_n$ and sparsity pattern $\sbf =(s_1, \cdots, s_n)$,  and let $C_o= C \min_i \{  1-\delta_i  \}$, $D_o= D\max_i \{ 1+\delta_i \}$. Then, for any $\vbf \in \Kbb^N$,   
\begin{equation*}\label{pr13}
 C_o  \|\vbf \|_2^2 \leq  \sum_{i} \|\Abf\vbf_i \|_2^2 \leq D_o \|  \vbf \|_2^2. 
\end{equation*}
\end{aprop}

\begin{proof}
Using the fusion frame inequality, and inequalities \eqref{eq:traditionalRIP} for all $i \in I$, we obtain 
  $$C \min_i \{  1-\delta_i  \}   \|\vbf \|_2^2 \leq   \min_i \{  1-\delta_i  \} \sum_{i} \| \vbf_i \|_2^2 \leq   \sum_{i} (1-\delta_i)\|\vbf_i\|_2^2\leq \sum_{i} \|\Abf\vbf_i \|_2^2  $$  $$ \leq    \sum_{i}( 1+\delta_i ) \| \vbf_i \|_2^2 \leq  \max_i \{ 1+\delta_i \}\sum_{i} \| \vbf_i \|_2^2   \leq D\max_i \{ 1+\delta_i \} \|  \vbf \|_2^2.$$
\end{proof}

\begin{atheorem}
\label{HFadaptDRIP}
Let $\varepsilon > 0$. 
Let $\Wcal = (W_i,P_i)_{i=1}^n$ be a fusion frame for $\Kbb^N$, $N \geq 1$. 
Let $\Abf \in \Kbb^{m \times N}$ be a subgaussian matrix with parameters $\beta, k$. 
Then, there exists a constant $C = C_{\beta,k}$ such that the P-RIP constants of $\frac{1}{\sqrt{m}}\Abf$ satisfy $\delta_{s_i} \leq \delta_i$, for $1 \leq i \leq n$ with probability at least $1-\varepsilon$, provided
\begin{equation*}
m \geq C \min(\delta)^{-2}\( \max(s_i)\ln(eN/\max(s_i)) + \ln(2\varepsilon n) \)
\end{equation*}
\end{atheorem}

Before we prove this result, we recall a standard RIP result for subgaussian matrices (Theorem 9.2 in \cite{foucart2013book}):
\begin{atheorem}\label{HF} Let $\Abf$  be an $m\times N$ subgaussian random matrix. Then there exists a constant $C>0$ (depending only on subgaussian parameters $\beta$, $k$) such that the RIP constant of $\frac{1}{\sqrt{m}} \Abf$ satisfies $\delta_s \leq \delta$ with probability at least $1-\varepsilon$, if 
\[ m \geq C \delta^{-2} \(s \ln(eN/s) +\ln(2\epsilon^{-1})\).\]
\end{atheorem}

\begin{proof}(of Theorem~\ref{HFadaptDRIP})

 For some constants $\delta_i \in (0,1)$, let $E$ be the event ``$\frac{1}{\sqrt{m}}\Abf$ does not satisfy P-RIP with respect to $\Wcal$ with constants $\delta_1, \cdots, \delta_n$''.
Applying a union bound it follows that
\begin{equation*}
\Pbb(E) = \Pbb(\exists i \in \{1, \cdots, n\}: \delta_{s_i} > \delta_i) \leq \sum_{i = 1}^n \Pbb(\delta_{s_i} > \delta_i) = \sum_{i = 1}^n \varepsilon_i
\end{equation*}

 For some $\varepsilon_i$ such that $\varepsilon_1 + \cdots + \varepsilon_n = \varepsilon \in (0,1)$,
 since the sensing matrix is the same for every sensors, there exists a unique $C > 0$ (depending on $\beta,k$) such that $\frac{1}{\sqrt{m}}\Abf$ satisfies the RIP locally within the subset $W_i$ and for a sparsity $s_i$ with constant $\delta_{s_i} \leq \delta_i$, provided that $m \geq  \max_{1\leq i \leq n} m_i$, with 
 \begin{equation*}\label{condGausrip}
   m_i \geq  C \delta_i^{-2} (s_i \ln(eN/s_i) +\ln(2\epsilon_i^{-1}) ).
 \end{equation*}  
 Additionally, since the function $s \to s\ln(eN/s)$ is monotonically increasing on $(0,N)$ and $s \ll N$, 
 $\max_{1 \leq i \leq n}m_i \geq C \max_{1 \leq i \leq n}\delta_i^{-2}\( \max_{1 \leq i \leq n}s_i\ln(eN/\max_{1 \leq i \leq n}s_i) + \ln(2\varepsilon_i^{-1}) \)$.
 Considering $\varepsilon_1 = \cdots = \varepsilon_n = \varepsilon/n$ concludes the proof.
\end{proof}
Note: all Gaussian and Bernoulli random matrices are subgaussian random matrices, so Theorem~\ref{HFadaptDRIP}  holds true for  Gaussian and Bernoulli random matrices.

\subsection{Recovery in general fusion frames settings}
With the tools introduced above, we show that any signals with sparsity pattern $\sbf$ can be recovered in a stable and robust manner via the fusion frame approach described in the previous sections. 

\subsubsection{RP-NSP based results}
Our first recovery guarantee results generalizes Thm. 3.9 from~\cite{ABL2017Conm} and is based on the robust partial NSP, introduced in Definition~\ref{def:rpnsp}.

\begin{atheorem}
\label{thm:rdnspRecovery}
Let $\Abf \in \Kbb^{m \times N}$ and $\Wcal = \( W_i, P_i \)_{i=1}^n$ a fusion frame with frame bounds $0 < C \leq D < \infty$ and frame operator $S$. 
Let $\left(\ybf^{(i)}\right)_{i=1}^n$ be the linear measurements $\ybf^{(i)} = \Abf P_i \xbf + \ebf^{(i)}$, $1\leq i \leq n$ for some noise vectors $\ebf^{(i)}$ such that $\|\ebf^{(i)}\|_2 \leq \eta_i$. 
Denote by $\widehat{\xbf^{(i)}}$ the solution to the local Basis Pursuit problems~\eqref{eq:fcsbpdn} and let $\widehat{\xbf} = S^{-1} \sum_i \widehat{\xbf^{(i)}}$.
If the matrix $\Abf$ satisfies the $\lsp^1$-RP-NSP with sparsity pattern $\sbf = (s_1,\cdots,s_n)$ with constants $0 < \rho_1, \cdots, \rho_n < 1$ and $\tau_1, \cdots, \tau_n > 0$ with respect to $\Wcal$, then the estimation $\widehat{\xbf}$ approximates $\xbf$ in the following sense: 
\begin{equation}
\label{eq:rdnspRecovery}
\|\widehat{\xbf} - \xbf \|_2 \leq \frac{2}{C}\( \langle \vec{\rho}, \sigma_{\sbf}^{\Wcal}(\xbf)_1 \rangle + \langle \vec{\tau}, \vec{\eta} \rangle\),
\end{equation}
where $\vec{\rho} = \(\frac{1+\rho_i}{1-\rho_i}\)_{i=1}^n$, $\vec{\tau} = \(\frac{2 \tau_i}{1-\rho_i}\)_{i=1}^n$, and $\vec{\eta} = \( \eta_i \)_{i=1}^n$.
\end{atheorem}

\begin{proof}
The solution is given by the fusion process $\widehat{\xbf} = S^{-1}\( \sum_{i=1}^n \widehat{\xbf^{(i)}} \)$ with $\widehat{\xbf^{(i)}}$ the solutions to the local problems~\eqref{eq:fcsbpdn}. 
It holds
\begin{equation*}
\|\xbf - \widehat{\xbf}\|_2 = \left\|S^{-1}\(\sum_{i = 1}^nP_{i}\xbf - \sum_{i = 1}^n\widehat{\xbf^{(i)}}\)\right\|_2 \leq C^{-1}\sum_{i = 1}^n \left\|P_{i}\xbf - \widehat{\xbf^{(i)}}\right\|_2 \leq C^{-1}\sum_{i = 1}^n \left\|P_{i}\xbf - \widehat{\xbf^{(i)}}\right\|_1.
\end{equation*}
For a particular $i \in \{1, \cdots, n\}$, we estimate the error on the subspace $W_i$ in the $\lsp^1$ sense. 
We follow the proof techniques from~\cite[Section 4.3]{foucart2013book} with the adequate changes. 
With $\vbf := P_i\xbf - \widehat{\xbf^{(i)}}$ and $S_i \subset W_i$ the set of best $s_i$ components of $\xbf$ supported on $W_i$, the $\lsp^1$-RP-NSP yields 
\begin{equation*}
\|(P_{i}\vbf)_{S_i}\|_1 \leq \rho_i\left\|\(P_{i}\vbf\)_{\overline{S_i}}\right\|_1 + \tau_i\|\Abf \vbf\|_2.
\end{equation*}
Combining with~\cite[Lemma 4.15]{foucart2013book} stating 
\begin{equation*}
\left\|\(P_{i}\vbf\)_{\overline{S_i}}\right\|_1 \leq \left\|P_{i}\widehat{\xbf^{(i)}}\right\|_1 - \left\|P_{i}\xbf\right\|_1 + \left\|\(P_{i}\vbf\)_{S_i}\right\|_1 + 2 \left\|\(P_{i}\xbf\)_{\overline{S_i}}\right\|_1.
\end{equation*}
we arrive at
\begin{equation*}
(1-\rho_i)\left\|\(P_{i}\vbf\)_{\overline{S_i}}\right\|_1 \leq \|P_{i}\widehat{\xbf^{(i)}}\|_1 - \|P_{i}\xbf\|_1 + 2 \left\|\(P_{i}\xbf\)_{\overline{S_i}}\right\|_1 + \tau_i\|\Abf\vbf\|_2 .
\end{equation*}
Applying once again the $\lsp^1$-RP-NSP, it holds
\begin{align*}
\left\|P_{i}\vbf\right\|_1 &= \left\| (P_{i}\vbf)_{S_i} \right\|_1 + \left\| \(P_{i}\vbf\)_{\overline{S_i}} \right\|_1 \leq \rho_i\left\| \(P_{i}\vbf\)_{\overline{S_i}} \right\|_1 + \tau_i\|\Abf \vbf\|_2 + \left\| \(P_{i}\vbf\)_{\overline{S_i}} \right\|_1 \\ 
 &\leq \(1+\rho_i\)\left\| \(P_{i}\vbf\)_{\overline{S_i}} \right\|_1 + \tau_i\|\Abf\vbf\|_2 \\ 
 &\leq \frac{1+\rho_i}{1-\rho_i}\( \|P_{i}\widehat{\xbf^{(i)}}\|_1 - \|P_{i}\xbf\|_1 + 2 \left\|\(P_{i}\xbf\)_{\overline{S_i}}\right\|_1\) + \frac{4\tau_i}{1-\rho_i}\|\Abf\vbf\|_2. 
\end{align*}
We now remember Lemma~\ref{lemma:suppSolution} and notice that $P_i \widehat{\xbf^{(i)}} = \widehat{\xbf^{(i)}}$. 
$\widehat{\xbf^{(i)}}$ being the optimal solution to~\eqref{eq:fcsbpdn}, it is clear that $\|\widehat{\xbf}\|_1 \leq \|P_{i}\xbf\|_1$ from what we can conclude that
\begin{equation*}
\left\|P_i \xbf - \widehat{\xbf^{(i)}}\right\|_1 = \|P_{i}\vbf\|_1 \leq 2\frac{1+\rho_i}{1-\rho_i}\sigma_{\sbf}^{\Wcal}(\xbf)_{1,i} + \frac{4\tau_i}{1-\rho_i}\|\Abf\vbf\|_2.
\end{equation*}
Summing up the contributions for all $i$ in $\{1, \cdots, n\}$ and applying the inverse frame operator finishes the proof.
\end{proof}

Similarly, assuming $\lsp^q$-RP-NSP, one can adapt the proof techniques from~\cite[Theorems 4.22, 4.25]{foucart2013book} to the local problems. 
This yields the following result
\begin{atheorem}
\label{thm:rdnspRecoveryl2NSP}
Let $\Abf \in \Kbb^{m \times N}$ and $\Wcal = \( W_i, P_i \)_{i=1}^n$ a fusion frame with frame bounds $0 < C \leq D < \infty$ and frame operator $S$. 
Let $\left(\ybf^{(i)}\right)_{i=1}^n$ be the linear measurements $\ybf^{(i)} = \Abf P_i \xbf + \ebf^{(i)}$, $1\leq i \leq n$ for some noise vectors $\ebf^{(i)}$ such that $\|\ebf^{(i)}\|_2 \leq \eta_i$. 
Denote by $\widehat{\xbf^{(i)}}$ the solution to the local Basis Pursuit problems~\eqref{eq:fcsbpdn} and let $\widehat{\xbf} = S^{-1} \sum_i \widehat{\xbf^{(i)}}$.
If the matrix $\Abf$ satisfies the $\lsp^2$-RP-NSP with sparsity pattern $\sbf = (s_1,\cdots,s_n)$ with constants $0 < \rho_1, \cdots, \rho_n < 1$ and $\tau_1, \cdots, \tau_n > 0$ with respect to $\Wcal$, then the estimation $\widehat{\xbf}$ approximates $\xbf$ in the following sense: 
\begin{equation}
\label{eq:rdnspRecoveryl2NSP}
\|\widehat{\xbf} - \xbf \|_p \leq \frac{1}{C}\( \frac{\langle \vec{\rho}, \sigma_{\sbf}^{\Wcal}(\xbf)_1 \rangle}{s^{1-1/p}} + \frac{\langle \vec{\tau}, \vec{\eta} \rangle}{s^{1/2-1/p}}\), \quad 1 \leq p \leq 2,
\end{equation}
where $\vec{\rho} = \(\frac{2(1+\rho_i)^2}{1-\rho_i}\)_{i=1}^n$, $\vec{\tau} = \(\frac{3-\rho_1}{1-\rho_i}\tau_i\)_{i=1}^n$, and $\vec{\eta} = \( \eta_i \)_{i=1}^n$.
\end{atheorem}

\subsubsection{P-RIP based recovery}
One can show that the P-RIP is sufficient for stable and robust recovery by combining Theorem~\ref{thm:rdnspRecovery} with the following result, showing the existence of random matrices satisfying the RP-NSP, from an RIP argument.
\begin{atheorem}\label{secondthm}
Let $\Abf \in \Kbb^{m \times N}$ be a matrix satisfying the P-RIP($2\sbf$,$\delta$), with $\sbf = (s_1,\cdots,s_n)$ and $\delta = (\delta_1,\cdots,\delta_n)$ and $\delta_i < 4/\sqrt{41}$, for all $1 \leq i \leq n$. 
Then, $\Abf$ satisfies the $\lsp^2$-RP-NSP with constants $(\rho_i,\tau_i)_{i=1}^n$ where 
\begin{equation}
\label{eq:RIPtoNSP}
\begin{array}{l}
\rho_i := \frac{\delta_i}{\sqrt{1-\delta_i^2} - \delta_i/4} < 1 \\
\tau_i := \frac{\sqrt{1+\delta_i}}{\sqrt{1-\delta_i^2} - \delta_i/4}.
\end{array}
\end{equation}
\end{atheorem}

\begin{proof}
The proof of this results consists in simply applying~\cite[Theorem 6.13]{foucart2013book} to every subspaces independently.
\end{proof}

From this result and Theorem~\ref{thm:rdnspRecoveryl2NSP}, it follows.
\begin{atheorem}
\label{thm:fCS_RIPresults}
Let $\Wcal = (W_i,P_i)_{i=1}^n$ be a fusion frame for $\Kbb^N$ with frame operator $S$ and frame bounds $0 < C \leq D < \infty$. 
Let $\Abf \in \Kbb^{m \times N}$ be a matrix satisfying the P-RIP($2\sbf$, $\delta$) where $\sbf = (s_1, \cdots, s_n$) and $\delta = (\delta_1, \cdots, \delta_n)$ with $\delta_i < 4/\sqrt{41}$, for all $1 \leq i \leq n$. 
Then any distributed-sparse vector $\xbf \in \Sigma_\sbf^{(\Wcal)}$ can be  recovered 
by solving $n$ \eqref{eq:bpdn} problems.

Assuming the noise in each \eqref{eq:bpdn} problem is controlled by  $\|\ebf^{(i)}\|_2 \leq \eta_i$, $1\leq i \leq n$, and set $vec{\eta} = \( \eta_i \)_{i=1}^n$. 
Let $\widehat{\xbf} = S^{-1}\sum \widehat{\xbf^{(i)}}$. Then
 
\begin{equation*}
 \|\widehat{\xbf} - \xbf \|_2 \leq \frac{1}{C}  \sum_{i=1}^n \alpha_i\frac{\sigma_{\sbf}^{\Wcal}(\xbf)_{1,i}}{\sqrt{s_i}} + \beta_i\eta_i
\end{equation*}
where $\alpha_i$ and $\beta_i$ depend only on the RIP constants $\delta_i$.
\end{atheorem}


 
\section{Local sparsity in general dictionaries and frames}
\label{ShidongSection}

The redundancy inherent to frame structures (and their generalization to frames of subspaces) makes them appealing to signal analysis task. 
So far, we have used the redundancy of the fusion frame process in order to increase the global sparsity of the original vector $\xbf$ as well as increase the robustness to noise. 
We investigate now the use of local dictionaries in order to use the redundancy within subspaces, using the local frames for the representation of the partial information. 
A common scenario in applications is when $f \in \Hcal$ has a sparse frame representation $f=D\xbf$, i.e. $\xbf$ is sparse, and the multiple measurements are given by
\[
  \ybf^{(i)} = \Abf P_i f + \ebf^{(i)} = \Abf P_i D \xbf + \ebf^{(i)}, \ \ \ 1 \leq i \leq n.
\]
Here $P_i$ can be any projection onto a subspace of $\mathcal{H}$. 
In practical applications such as in SAR radar imaging, $P_i$ can just be a projection or spatial filter onto $W_i\equiv \operatorname{span}\{\dbf_k\}_{k\in \Omega_i}$, where $\dbf_k$ is the $k^{th}$ column of $D$.
Such an operation can potentially reduce the number of nonzero entries of $\xbf$ in the $i^{th}$ observation, when the respective column vectors $\dbf_k$ are in $\operatorname{ker} P_i$.
In particular, let us denote by $\Gamma_i\equiv\{k \; | \; k\not\in\Omega_i,\ \dbf_k\in W_i\}$, and $\Lambda_i\equiv \{l \; | \; \dbf_l\in \operatorname{ker} P_i\}$.
\begin{eqnarray*}
P_if &=& P_i D \xbf =P_i\left(\{\dbf_k\}_{k\in\Omega_i}, \{\dbf_j\}_{j\in\Gamma_i}, \{\dbf_l\}_{l\in\Lambda_i}\right) \xbf \\
  &=& \left(\{\dbf_k\}_{k\in\Omega_i}, \{y_j=P_i \dbf_j\}_{j\in\Gamma_i}, \{0's\}_{l\in\Lambda_i}\right)\left(\begin{array}{c}
                                        \xbf_{\Omega_i} \\
                                        \xbf_{\Gamma_i} \\
                                        \xbf_{\Lambda_i}
                                      \end{array}\right) \\
  &=& \left(\{\dbf_k\}_{k\in\Omega_i}, \{y_l=P_i\dbf_l\}_{l\in\Gamma_i} \right)\left(\begin{array}{c}
                    \xbf_{\Omega_i} \\
                    \xbf_{\Gamma_i}
               \end{array}\right)\\
  &=& D_i \xbf^{(i)},
\end{eqnarray*}
where $D_i\equiv \left(\{\dbf_k\}_{k\in\Omega_i}, \{y_k=P_i\dbf_k\}_{k\in\Gamma_i} \right)$, and
\[
\xbf^{(i)}\equiv \left(\begin{array}{c}
                    \xbf_{\Omega_i} \\
                    \xbf_{\Gamma_i}
               \end{array}\right).
\]

As a result, the $i^{th}$ measurement becomes
\begin{equation*}
\label{eqn_Ba}
 \ybf^{(i)} = \Abf D_i \xbf^{(i)} + \ebf^{(i)}, \quad 1 \leq i \leq n,
\end{equation*}
or
\begin{equation*}\label{eqn_Bb}
 \ybf^{(i)} = \Abf f^{(i)} + \ebf^{(i)}, \quad f^{(i)}=D_i\xbf^{(i)}, \quad 1 \leq i \leq n.
\end{equation*}
Note that the first version suggests the use of $\lsp^1$ synthesis methods, while the second one looks at $\lsp^1$ analysis tools. 
$\lsp^1$ synthesis corresponds to the usual sparse recovery, via a dictionary $D$:
\begin{equation*} 
\label{eq:dict_l0}
\min_{\zbf} \|\zbf\|_0, \quad \text{ subject to } \|\Abf D\zbf - \ybf\|_2 \leq \eta.
\end{equation*}
The solution $\widehat{f}$ is later computed as $\widehat{f} = D\widehat{\xbf}$. 
In the $\lsp^1$ analysis approach, we do not care for a particular (sparse) representation of $f$. 
We just ask for this representation to have high fidelity with the data:
\begin{equation}
\label{eq:dict_analysis}
\min_{\gbf} \|D^*\gbf\|_0, \quad \text{ subject to } \|\Abf \gbf - \ybf\|_2 \leq \eta
\end{equation}
Here, $D^*$ denotes the canonical dual frame. 
Both approaches are further detailed in the next sections. 
We comment that if the choice of $P_i$ is allowed, one strategy is again to use random projections by randomly selecting the index set $\Omega_i$ to set the subspaces $W_i=\operatorname{span}\{\dbf_j\}_{j\in\Omega_i}$.

\subsection{Recovery via general \texorpdfstring{$\lsp^1$}{TEXT}-analysis method}
As introduced above, we try to recover $f$ that has {\bf a} sparse representation by solving Problem~\eqref{eq:dict_analysis}. 
While the problem is written in terms of the canonical dual frame, there is no obligation in using this particular dual frame. 
One may instead optimize the dual frame considered and use the sparsity-inducing dual frame \cite{liu2012optimaldual}, computed as part of the optimization problem:
\begin{equation}\label{eqn_B1}
  \widehat{f^{(i)}} =\argmin_{g,\, D\tilde{D}_i^*=\Ibf}\|\tilde D^*_i g\|_1 \quad \text{s.t. } \|\Abf g -\ybf^{(i)}\|_2\le \eta_i, \quad 1 \leq i \leq n.
\end{equation} 
The sparsity-inducing frame $\tilde D_i$ can be uniform across all $i$ but not necessarily.
The following result is known to hold for any dual frame~\cite{liu2012optimaldual}.
\begin{atheorem}\label{thm1}
  Let $D$ be a general frame of $\mathbb{R}^{N}$ with frame bounds $0<A\leq B<\infty$.  Let $\tilde{D}$ be an alternative
  dual frame of $D$ with frame bounds $0<\tilde{A}\leq \tilde{B}<\infty$, and let $\rho=s/b$. Suppose that the matrix $\Abf$ satisfies the following D-RIP condition
\begin{equation}\label{eqn_C1}
     \(1-\sqrt{\rho B \tilde{B}}\)^2 \cdot \delta_{s+a} +
\rho B \tilde{B}\cdot\delta_{b} < 1 - 2\sqrt{\rho B
\tilde{B}}
\end{equation}
for some positive integers $a$ and $b$ satisfying $0< b-a\leq
3a$. Let $\widehat{f}$ be the solution to the typical $\lsp^1$-analysis problem
\[
  \displaystyle \widehat{f} =\argmin_{g}\|\tilde D^* g\|_1 \quad \text{s.t. } \|\Abf g -\ybf\|_2\le \eta.
\]
Then
  \begin{equation}
    \label{eq15} \Vert \widehat{f}-f \Vert_{2} \leq \alpha\eta +
    \beta\frac{\Vert\tilde{D}^{*}f-(\tilde{D}^{*}f)_{s_i}\Vert_{1}}{\sqrt{s}},
  \end{equation}
 where $\alpha$ and $\beta$ are some constants and $(\tilde{D}^{*}f)_{s}$
 denotes the vector consisting the $s$ largest entries in magnitude of
 $\tilde{D}^{*}f$.
\end{atheorem}

In particular, the bound~\eqref{eq15} applied to the local pieces of information $\widehat{f^{(i)}}$ analyzed with the local frames $\tilde{D}^*_i$, for $1 \leq i \leq n$, obtained as solution to Problem~\eqref{eqn_B1} yields the following error bound for the reconstruct signal $\widehat{f} = S^{-1}\sum_{i=1}^n\widehat{f^{(i)}}$:
\begin{align*}
\|\widehat{f}-f\|_2 &\leq \|S^{-1}\|_{2 \to 2}\sum_{i=1}^n \|\widehat{f^{(i)}}-f^{(i)}\|_2 \leq \alpha \eta + \beta\( \sum_{i=1}^n \frac{\|\tilde{D}^*_i f^{(i)}-(\tilde{D}^*_i f^{(i)})_{s_i}\|_1}{\sqrt{s_i}} \),
\end{align*}
where $\alpha\equiv \operatorname{sup}_i\alpha_{i}\|S^{-1}\|_{2 \to 2}$,  $\eta\equiv \sum_i \eta_i$, and $\beta \equiv \operatorname{sup}_i{\beta_{i}}\|S^{-1}\|_{2 \to 2}$
This bound is reminiscent of the traditional bounds in CS where the (local) error decays as $\sigma_s(\xbf)_1/\sqrt{s}$ except that the contributions of each subspace are added together and normalized by the norm of the inverse frame operator.

\subsection{Recovery via \texorpdfstring{$\lsp^1$}{TEXT}-synthesis method}
The gap between recovery via $\lsp^1$ synthesis and $\lsp^1$ analysis has long been studied. 
Interestingly, it can be shown~\cite{liu2012performance} that, when using the sparsity-inducing frames described above~\eqref{eqn_B1}, both approaches are equivalent. 
Denote by $\tilde D_{i,o}$ the resultant optimal dual frame and suppose that $\Abf$ satisfies a D-RIP property~\eqref{eqn_C1}. 
Then it follows from~\cite{liu2012performance} 
\begin{equation*}\label{eqn_B1a}
\|\widehat{f^{(i)}} - f^{(i)}\|_2\leq \alpha_{i}\eta_i + \beta_{i}\frac{\|\tilde D^*_{i, o} f^{(i)}-(\tilde D^*_{i,o} f^{(i)})_{s_i}\|_1}{\sqrt{s_i}}
\end{equation*}
for some positive constants $\alpha_{i}$ and $\beta_{i}$.
Considering the fusion of the local information $\widehat{f}=S^{-1}\left(\sum_{i = 1}^n \widehat{f^{(i)}}\right)$, the following
result holds true:

\begin{aprop}
Let $S$ be the invertible fusion frame operator. Suppose $\Abf$ satisfies condition~\eqref{eqn_C1}.
Then the fused solution $\widehat{f}$ has an error bound given by
\begin{equation*}
\label{eqn_B2}
\|\widehat{f} - f\|_2\leq \alpha\eta + \beta\left(\sum_{i = 1}^n \frac{\|\tilde D^*_{i,o} f^{(i)}-(\tilde D^*_{i,o} f^{(i)})_{s_i}\|_1}{\sqrt{s_i}}\right),
\end{equation*}
where $\alpha\equiv \operatorname{sup}_i\alpha_{i}\|S^{-1}\|_{2 \to 2}$,  $\eta\equiv \sum_i \eta_i$, and $\beta \equiv \operatorname{sup}_i{\beta_{i}}\|S^{-1}\|_{2 \to 2}$.
\end{aprop}

\begin{proof}
Write $f =S^{-1}S f =S^{-1}\sum_{i=1}^n P_i f =S^{-1}\sum_{i=1}^n f^{(i)}$.  
Direct computation shows
\begin{eqnarray*}
\|\widehat{f}-f\|_2 &=& \|S^{-1}\sum_{i = 1}^n (\widehat{f^{(i)}} - f^{(i)})\|_2 \\
&\leq & \|S^{-1}\|_{2 \to 2}\sum_{i = 1}^n \|(\widehat{f^{(i)}} - f^{(i)})\|_2 \\
&\leq & \|S^{-1}\|_{2 \to 2}\sum_{i=1}^n \left(\alpha_{i}\eta_i + \beta{i}\frac{\|\tilde D^*_{i, o} f^{(i)}-(\tilde D^*_{i,o} f^{(i)})_{s_i}\|_1}{\sqrt{s_i}}\right)  \\
&\leq & \operatorname{sup}_i{\alpha_{i}}\|S^{-1}\|_{2 \to 2}\sum_{i=1}^n \eta_i + \\
& & \operatorname{sup}_i{\beta_{j}}\|S^{-1}\|_{2 \to 2} \left(\sum_{i=1}^n \frac{\|\tilde D^*_{i, o} f^{(i)}-(\tilde D^*_{i,o} f^{(i)})_{s_i}\|_1}{\sqrt{s_i}}\right).
\end{eqnarray*}
The result follows directly by setting $\alpha \equiv \operatorname{sup}_i{\alpha_{i}}\|S^{-1}\|_{2 \to 2}$,  $\eta \equiv \sum_i \eta_i$, and $\beta \equiv \operatorname{sup}_i{\beta_{j}}\|S^{-1}\|_{2 \to 2}$.
\end{proof}

We comment that this result is not surprising due to the equivalence between the two problems described in this section when dealing with the sparsity-inducing dual frames.

\section{Examples}
We provide in this Section some examples of applications of our approach where it yields, without any fine tunning of the parameters,  comparable -- if not better -- results with common methods (which are precisely targeted for the given problems).

\subsection{Wavelet frames and recovery of Doppler signals}
This section is intended as an illustrative example and proof-of-concept of the tools developed so far. 
In particular, we want to show that using a fairly \emph{poor} quality device, we are capable to recover a signal of fairly high complexity. 
For the sake of reproducible research, all the experiments presented in this section can be obtained and reproduced from of the named authors' Github page\footnote{See https://github.com/jlbouchot/FFCS for all the self implemented files. These files require, as described in the README  file, to have the access to a Haar matrix function and to have CVX ~\cite{CVX1,CVX2} installed.}.
In this experiment, we try to recover a noisy Doppler signal (see Fig.~\ref{fig:noisyDoppler}) using sensors with very few measurements. 

\begin{figure}
    \centering
    \subfigure[Original noisy Doppler signal]{
        \includegraphics[width=0.8\textwidth]{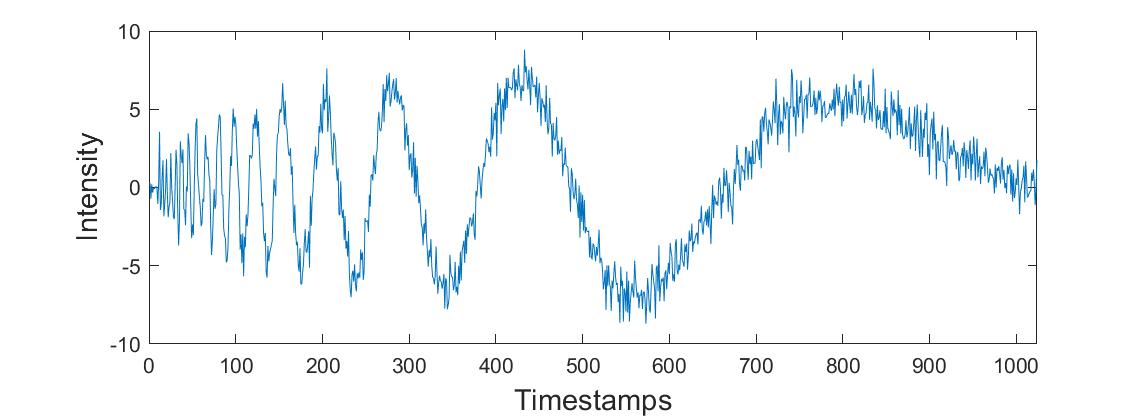}
        \label{fig:noisyDoppler}
    } \\
    ~ 
    \subfigure[Recovery from a distributed sensing approach -- $L^2$ error: $37.18$]{
        \includegraphics[width=0.8\textwidth]{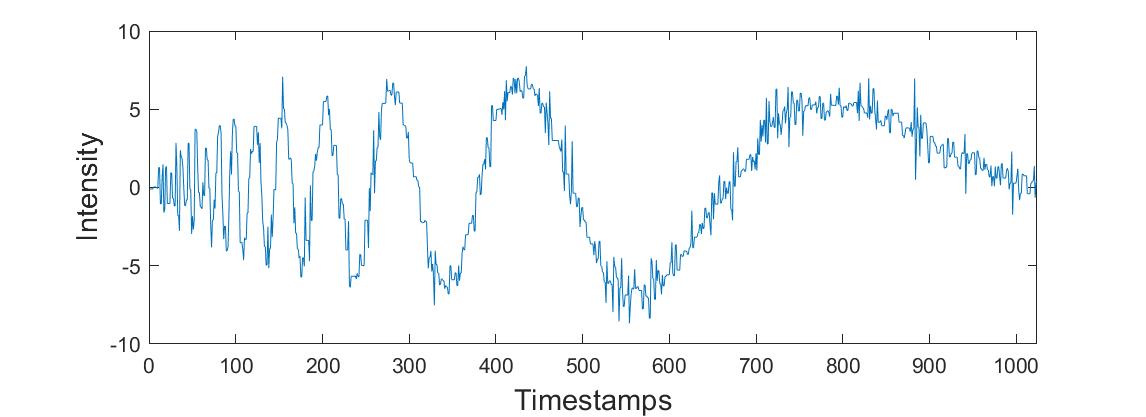}
        \label{fig:haarFF}
    } \\
    ~ 
    \subfigure[Reconstruction with an $\ell^1$ analysis method -- $L^2$ error: $82.72$]{
        \includegraphics[width=0.8\textwidth]{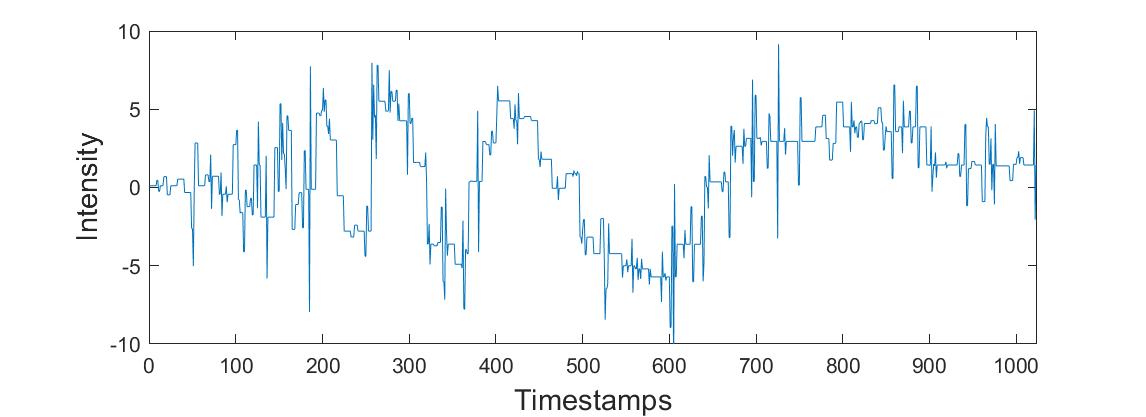}
        \label{fig:haarCS}
    } 
    \caption{A noisy Doppler signal (~\ref{fig:noisyDoppler}) and its reconstructions via a Haar wavelet frame approach \ref{fig:haarFF} and a direct $\ell^1$ analysis on the Haar system~\ref{fig:haarCS}}
		\label{fig:proofofconcept}
\end{figure}

To produce the Figures, we have considered a $1024$ dimensional noisy Doppler signal (obtained directly via Matlab's Wavelet toolbox). 
We compare a traditional $\ell^1$ analysis approach, where the sparsifying basis is chosen as the Haar system. 
Similarly, we consider the Haar wavelet systems as subspaces for our fusion frame approach. 
Given that a signal $s(t)$ can be decomposed as 
$$ s(t) = a_l(t) + \sum_{j \geq l}d_j(t) $$
where $a_l$ corresponds to the approximation coefficients at level $l$ and $d_j$ corresponds to the details at level $j$.
We therefore consider our projections $P_j$ to be onto the subspaces of details $W_j$ at each level (using the usual definition from multiresolution analysis), and a remaining low frequency approximation part. 
This gives us (in our setting) a grand total of $n = 9$ projections, $8$ of which ($P_1, \cdots, P_8$) correspond to detail coefficients, and the last one ($P_0$) corresponds to the low-frequency approximation. 
This system being orthogonal, it gives us the easy reconstruction formula
\begin{equation*}
\widehat{{\bf x}} = \sum_{i = 0}^8 \widehat{\xbf^{(i)}}.
\end{equation*}
All the local information are recovered with an $\ell^1$ analysis procedure, where the sparsifying basis are precisely chosen to be the Haar wavelets at the appropriate levels. 
Finally, the sparsity is set uniformly to $25$ on each subspace and a random Gaussian matrix with $m = 174$ rows is created and some normal noise with variance $0.05$ is added to the measurement vectors, all drawn independently from one another. 

Fig.~\ref{fig:haarFF} shows the reconstructed signal using our $\ell^1$-analysis fusion frame approach described in the previous section, while Fig.~\ref{fig:haarCS} is the reconstructed signal obtained from a traditional $\ell^1$-analysis approach. 
Both of them have been reconstructed using measurements obtained from the exact same matrix. 
It is important to note however, that we have not tried to optimize our parameters in this example. 
In particular, bearing in mind that the subspaces of higher-frequency details have higher dimensions (in fact, $\operatorname{rank}(P_j) = \mathcal{O}(2^j)$, for $1 \leq j \leq 8$), a better choice would be to set the local sparsity to be higher in these subspaces to improve our results. 
This level dependent-sampling is nothing new in the community and our results corroborate those from other approaches~\cite{adcock2017sparsityinlevel}. 
The Figures clearly show the reconstruction capabilities of our approach even with a sub-optimal setting and not making use of any kind of redundancy. 
Note that mostly the high frequency components of the signal are missing, which can be overcome by a better setting of the sparsity per subspace. 

\subsection{Natural image processing}
This section and the next one empirically studies the applicability of our proposed approach~\footnote{All experiments can be obtained from https://github.com/jlbouchot/FFCS}.
They detail some basic yet illustrative examples of the use of our fused sensing framework in imaging sciences. 
They prove both the easiness to implement all sorts of applications and the generality of the suggested approach.
Further examples can be found in~\ref{appendix}.

The first imaging application is concerned with natural images. 
To this end, we consider the traditional \emph{cameraman image} depicted on the left of Fig.~\ref{fig:camera:light}. 
We simulate a moving field of view in the $x$ direction, obtained by generating $n = 8$ traveling sine waves, as shown in Fig.~\ref{fig:camera:light}. 
This is motivated for its use in multi-channel MR Imaging~\cite{Pruessmann1999Sense,Ma2017SoS}, for which an example is given in the next section.

\begin{figure}[htbp]
\includegraphics[width=0.90\linewidth]{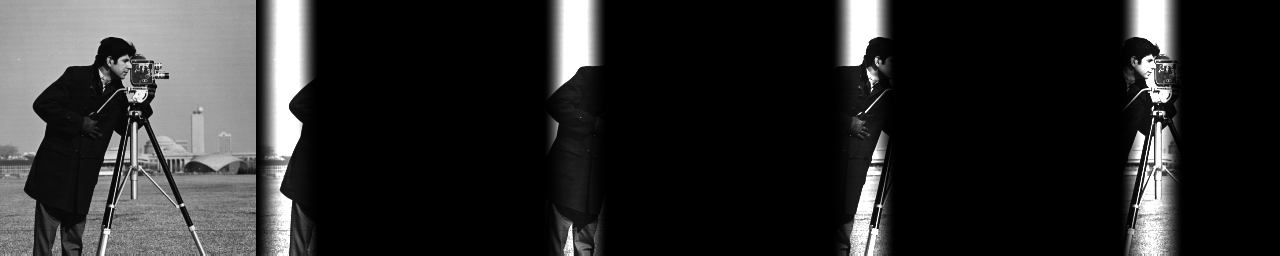}
\caption{Original cameraman image with the first four scene illuminations considered.}
\label{fig:camera:light}
\end{figure}

For this experiment, we compare the fused compressed sensing method described above with the traditional $\ell^1$ analysis approach and the usual $\ell^1$ minimization -- albeit knowing that the traditional compressed sensing is not suited for this task. 
The minimization is computed with a self-implemented primal-dual algorithm which we use for all the following experiments.
We exploit the fused sensing method based on the $\ell^1$ analysis model presented in Section~\ref{ShidongSection} using a Daubechies-4 wavelet decomposition as sparsifying dictionary. 
The measurements are obtained by subsampling the Fourier domain at random according to a Gaussian distribution. 
We consider a subsampling ratio of $4.84\%$, which means we are taking $3174$ samples from a $256 \times 256$ dimensional image. 
All the measurements are assumed to contain some randomly generated additive Gaussian noise with variance $0.05$. 
Fig.~\ref{fig:camera} compares the results obtained from our approach against the single sensor approaches in the $\ell^1$ analysis and traditional compressed sensing. 

\begin{figure}[htb]
\centering
\subfigure[Recovery using traditional compressed sensing]{
\includegraphics[width=0.30\linewidth]{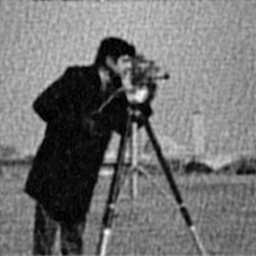}
\label{fig:camera:CS}
}
\subfigure[Recovery using an $\ell^1$ analysis approach with $db4$ wavelets]{
\includegraphics[width=0.30\linewidth]{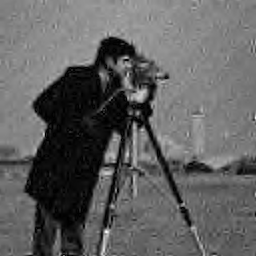}
\label{fig:camera:l1A}
}
\subfigure[Fused compressed sensing recovery via $\ell^1$ analysis model]{
\includegraphics[width=0.30\linewidth]{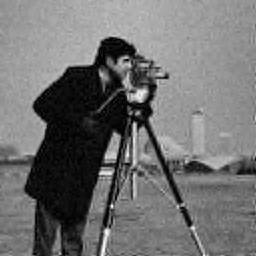}
\label{fig:camera:FCS} 
} \\ 
\subfigure[Pointwise error using traditional compressed sensing]{
\includegraphics[width=0.30\linewidth]{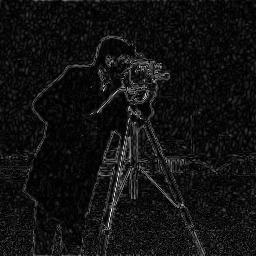}
\label{fig:camera:CS_error}
}
\subfigure[Pointwise error via $\ell^1$ analysis recovery with $db4$ wavelets]{
\includegraphics[width=0.30\linewidth]{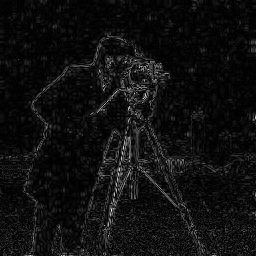}
\label{fig:camera:l1A_error}
}
\subfigure[Pointwise error of Fused $\ell^1$ analysis model]{
\includegraphics[width=0.30\linewidth]{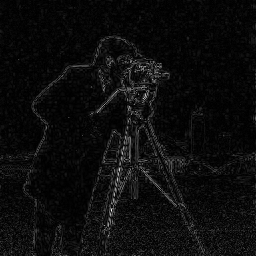}
\label{fig:camera:FCS_error}
}
\caption{Recovery results from single and multi-sensor set ups in natural images. Our approach allows to recover fine details whilst preserving the main structures of the subject.}
\label{fig:camera}
\end{figure}
As can be noted, even in this drastic undersampling situation, our method is capable of handling accurate reconstructions.
It can be seen that the details on the man are kept while making the noise less obvious. 
It is particularly visible in certain small details such as the hands of the cameraman, or the white parts of the feet of the tripod. 
The background also appears to be better defined in the fused image than in the classical and $\ell^1$ analysis methods. 

\begin{table}[htb]
\centering
\begin{tabular}{l|ccc}
\hline
               & Compressed sensing & $\ell^1$ analysis & Fused $\ell^1$ analysis \\ \hline
SSIM           &     $0.5473$       &     $0.5967$      &      $0.6838$           \\
PSNR           &     $23.076$       &     $22.893$      &      $24.342$           \\
$\ell^2$ error &     $17.965$       &     $18.348$      &      $15.530$           \\ \hline
\end{tabular}
\caption{Accuracy of the various recovery methods of the natural scene measured with the Structural SIMilarity index, Peak Signal to Noise Ratio, and pointwise $\ell^2$ error.}
\end{table}

\subsection{MR Image reconstruction}
The second imaging application is concerned with medical imaging. 
Following the ideas from ~\cite{Ma2017SoS,Pruessmann1999Sense}, we simulate a multi-channel MRI sensing set-up, first using the same sine waves as in the previous experiment (results displayed in Fig~\ref{fig:phantom_sine}), and then using spherical beam pattern (see in Fig.~\ref{fig:phantom_gauss} for the results). 
We generate the classical Shep-Logan phantom, and set its dimension to $1024 \times 1024$. 

The first experiment, depicted in Fig.~\ref{fig:phantom_sine} shows the results obtained by our method when using a Daubechies 4 sparsifying dictionary. 
Our results are compared to the single sensor $\ell^1$ analysis, and the results obtained by the Sum-of-Squares recovery~\cite{Pruessmann1999Sense,Ma2017SoS}. 

\begin{figure}[htb]
\centering
\subfigure[Recovery using the single sensor $\ell^1$ analysis]{
\includegraphics[width=0.30\linewidth]{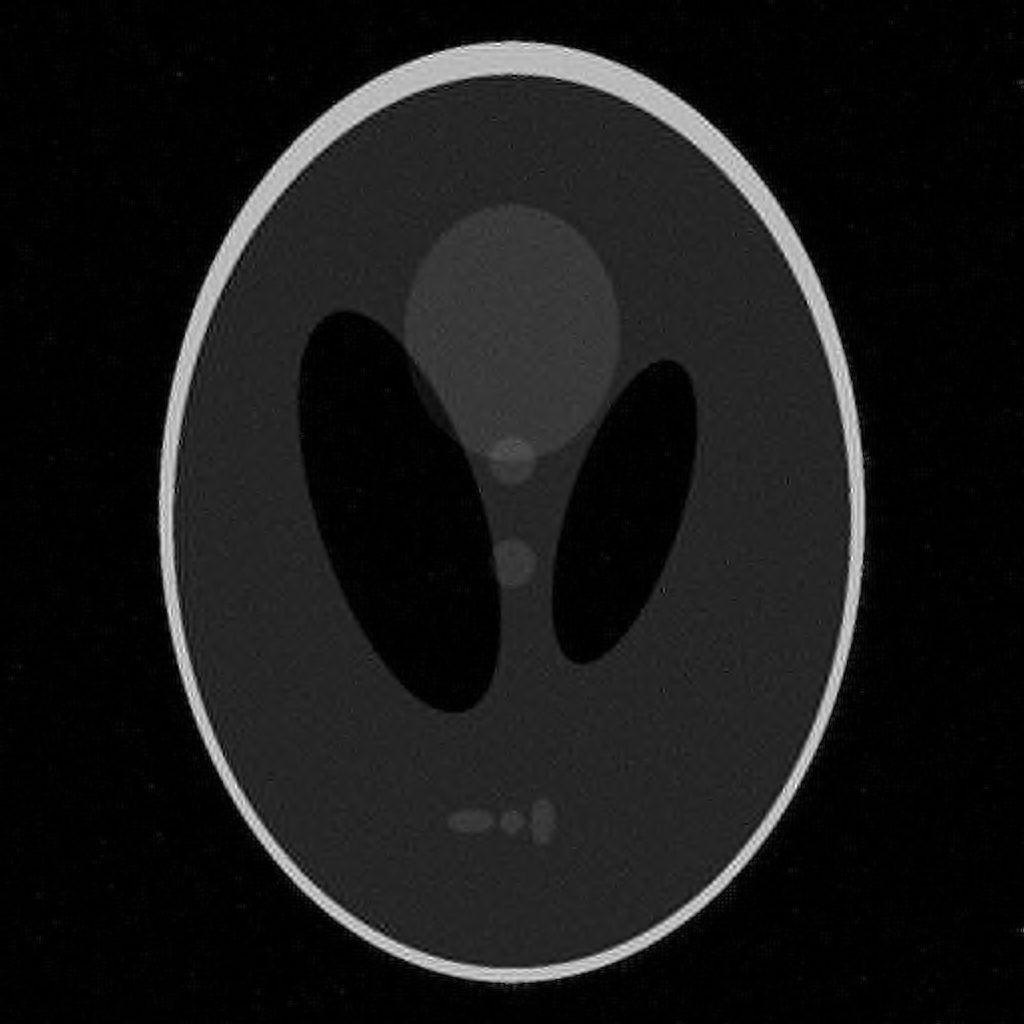}
\label{fig:phantom_sine:l1A}
}
\subfigure[Recovery using Sum-of-Squares approach]{
\includegraphics[width=0.30\linewidth]{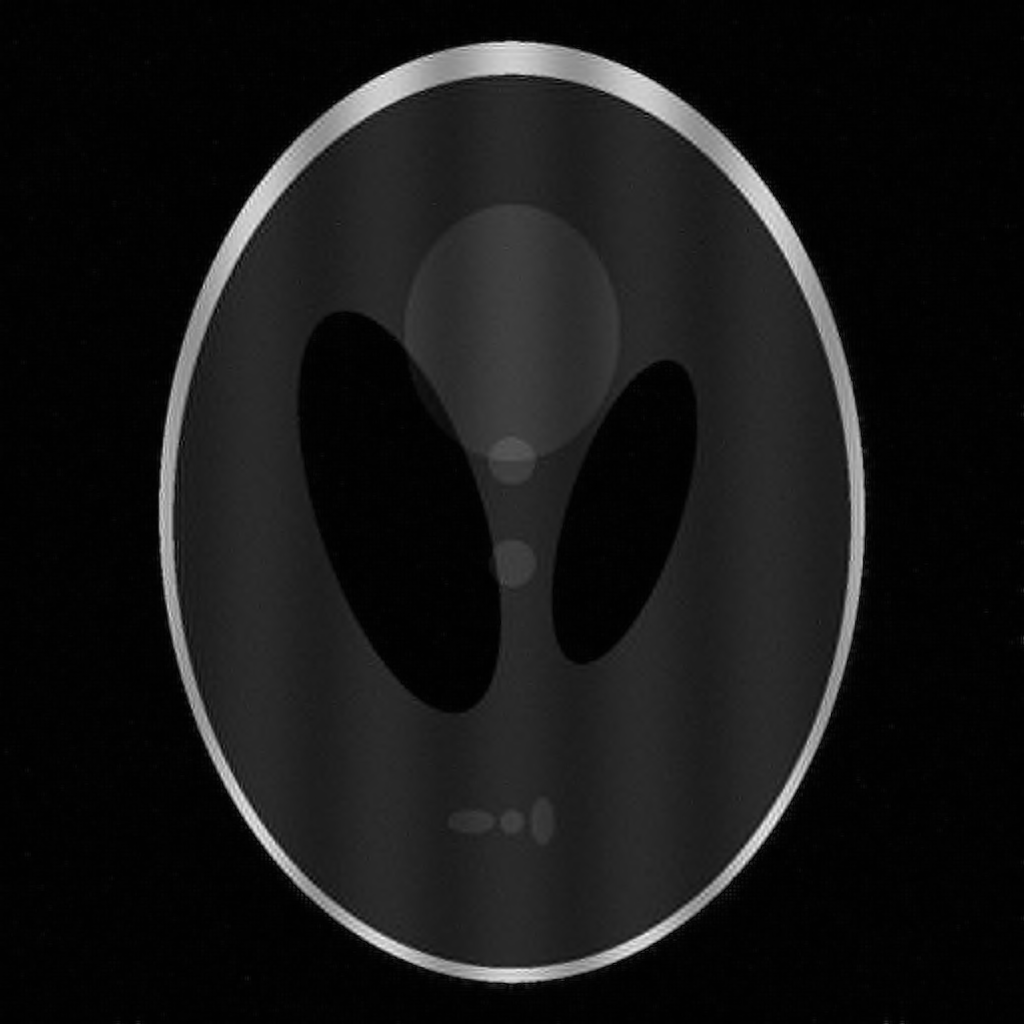}
\label{fig:phantom_sine:sos}
}
\subfigure[Fused compressed sensing recovery via $\ell^1$ analysis model]{
\includegraphics[width=0.30\linewidth]{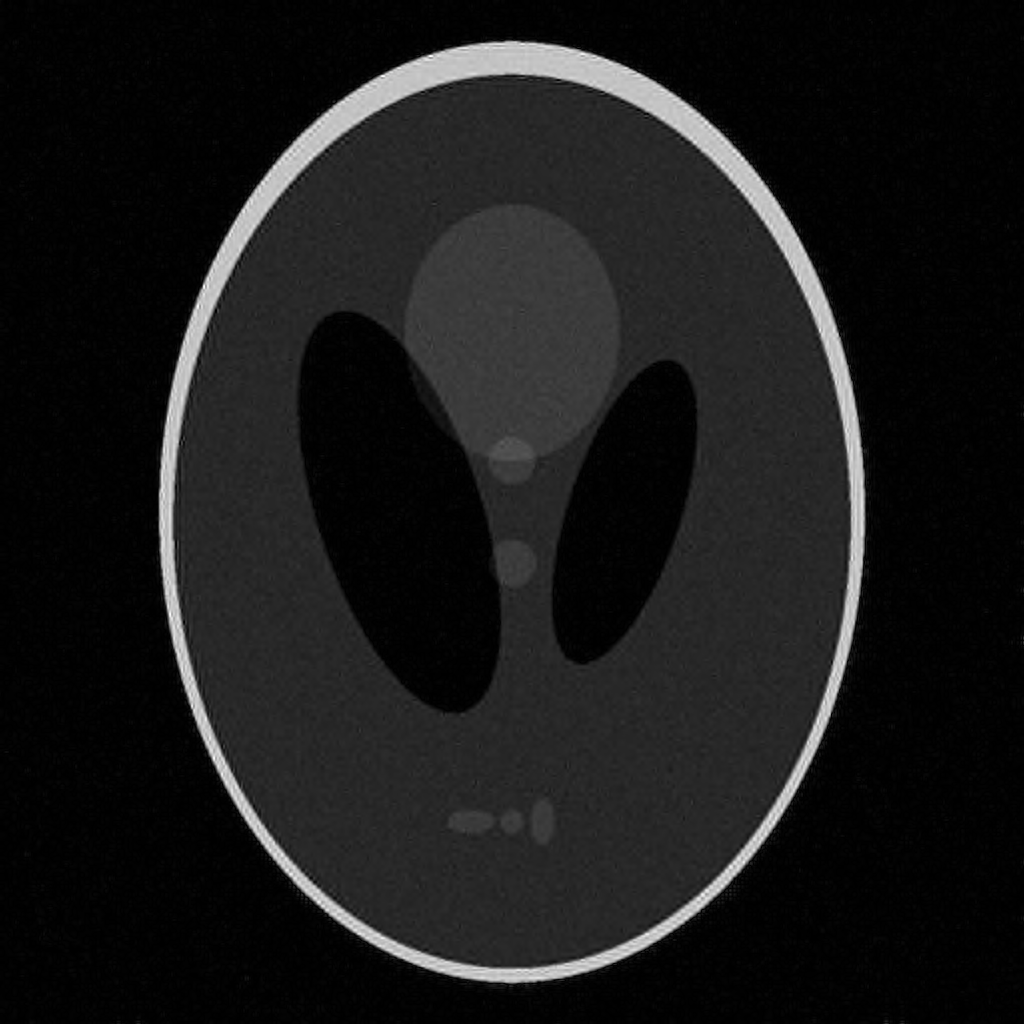}
\label{fig:phantom_sine:fcs}
} \\
\subfigure[Pointwise error via single sensor $\ell^1$ analysis]{
\includegraphics[width=0.30\linewidth]{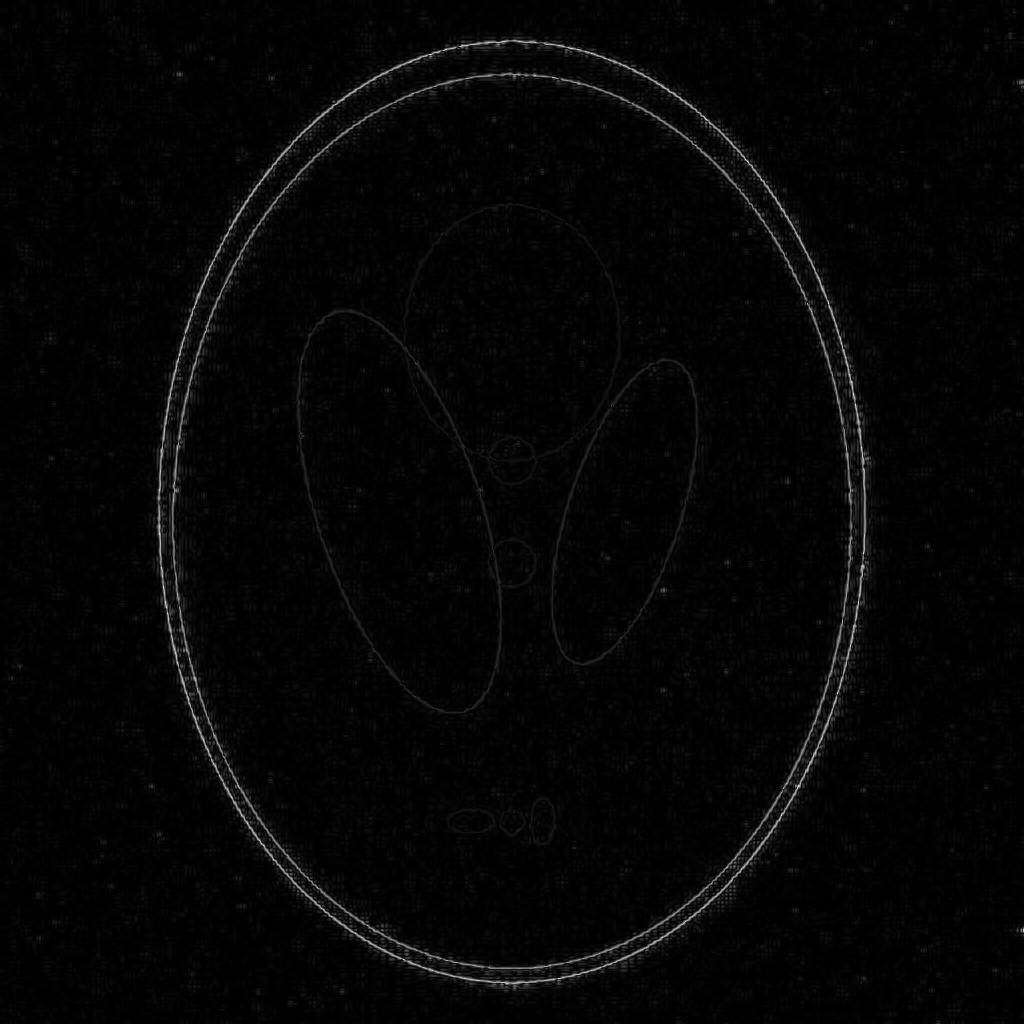}
\label{fig:phantom_sine:l1A_error}
}
\subfigure[Pointwise error Sum-of-Squares approach]{
\includegraphics[width=0.30\linewidth]{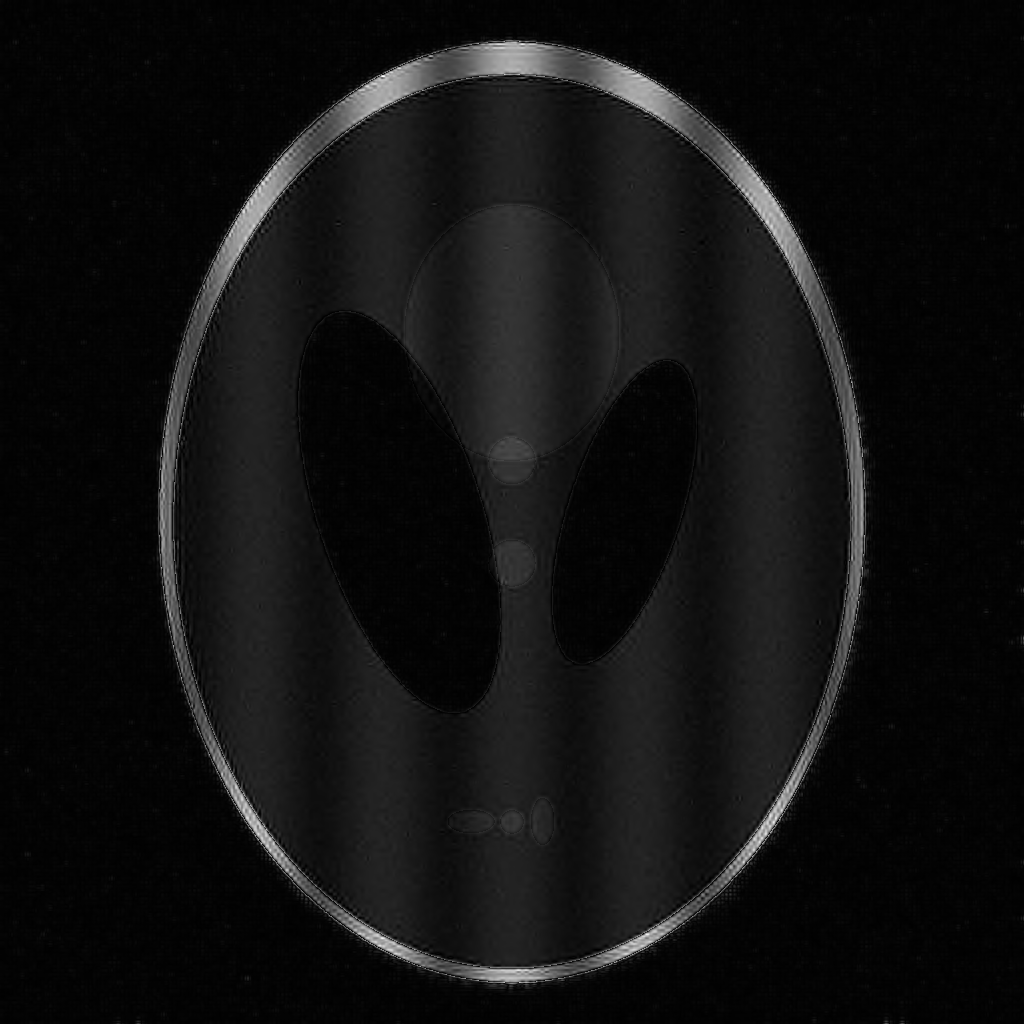}
\label{fig:phantom_sine:sos_error}
}
\subfigure[Pointwise error Fused $\ell^1$ analysis model]{
\includegraphics[width=0.30\linewidth]{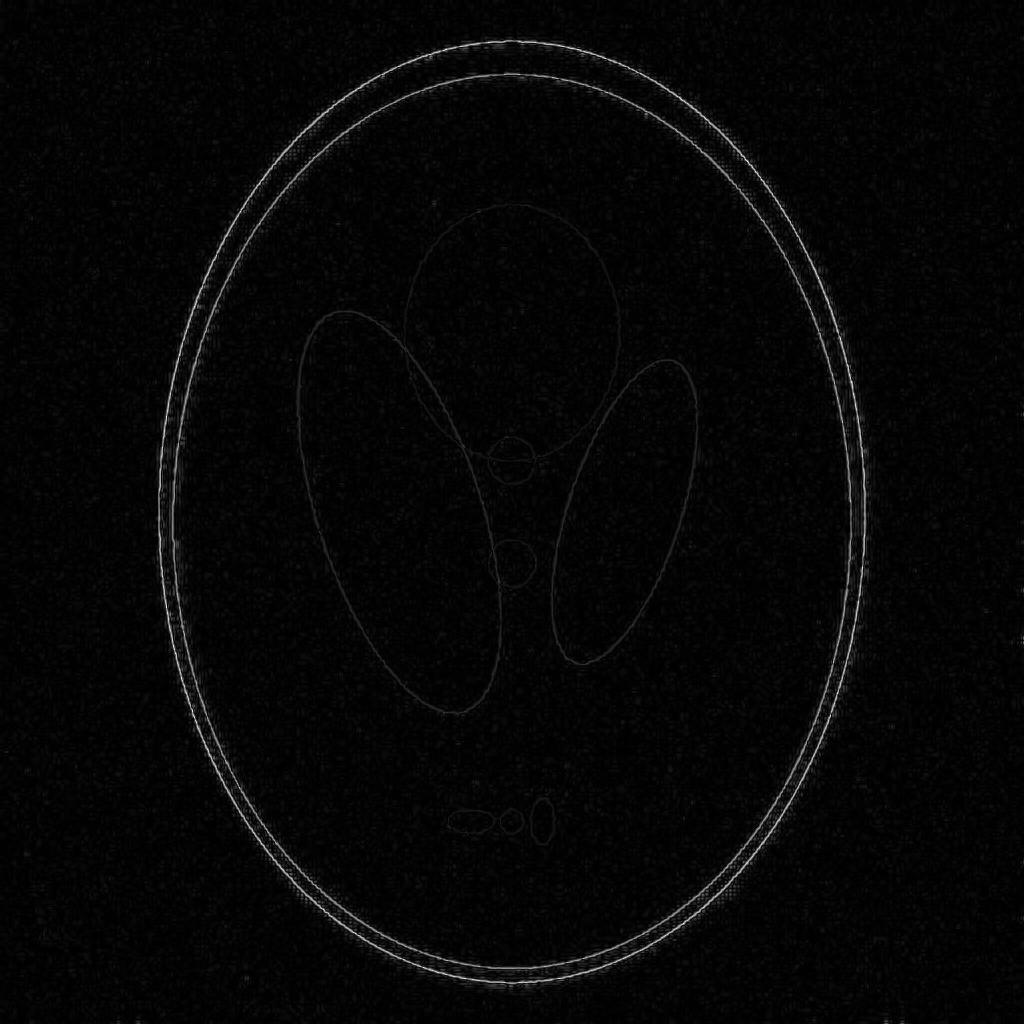}
\label{fig:phantom_sine:fcs_error}
}
\caption{Recovery of the phantom MR Image of size $1024 \times 1024$ from highly undersampled and noisy measurements. The sparse minimization is done on Daubechies 4 wavelet coefficients. Our approach allows~\ref{fig:phantom_sine:fcs} allows better recovery than other comparable methods~\ref{fig:phantom_sine:sos}.}
\label{fig:phantom_sine}
\end{figure}

The recovery is obtained from noisy measurements, in which some additive Gaussian noise with variance $0.05$ has been added. 
The measurements are obtained by subsampling $2.45\%$ ($25667$ samples) of the Fourier transform. 
We see here that the fused compressed sensing is better capable of handling a multi-channel problem with unusual illumination (compared to the Sum-of-Squares method). 
Moreover, more details are preserved, when compared to the single sensor $\ell^1$ analysis method. 

Some noise still appears in the image, but can easily be thresholded further if needed. 
One important aspect of MR Images that hasn't been considered in this research, is the fact that they are sparse in gradient. 
One usually prefers to minimize the Total Variation instead of the $\ell^1$ norm or another frame using $\ell^1$ analysis. 

Although we have not explicitly written the theory here, the recovered images in Fig.~\ref{fig:phantom_gauss} show the results when using a TV minimization instead of the $\ell^1$ analysis. 

\begin{table}[htb]
\centering
\begin{tabular}{l|ccc}
\hline
               & $\ell^1$ analysis  & Sum of Square fusion & Fused $\ell^1$ analysis \\ \hline
SSIM           &     $0.8022$       &       $0.4381$       &      $0.7672$           \\
PSNR           &     $27.917$       &       $14.893$       &      $28.484$           \\
$\ell^2$ error &     $41.159$       &       $184.36$       &      $38.558$           \\ \hline
\end{tabular}
\caption{Accuracy of the various recovery methods of the phantom image measured with the Structural SIMilarity index, Peak Signal to Noise Ratio, and pointwise $\ell^2$ error.}
\end{table}

\begin{figure}[htb]
\centering
\subfigure[Recovery using the single sensor $\ell^1$ analysis]{
\includegraphics[width=0.22\linewidth]{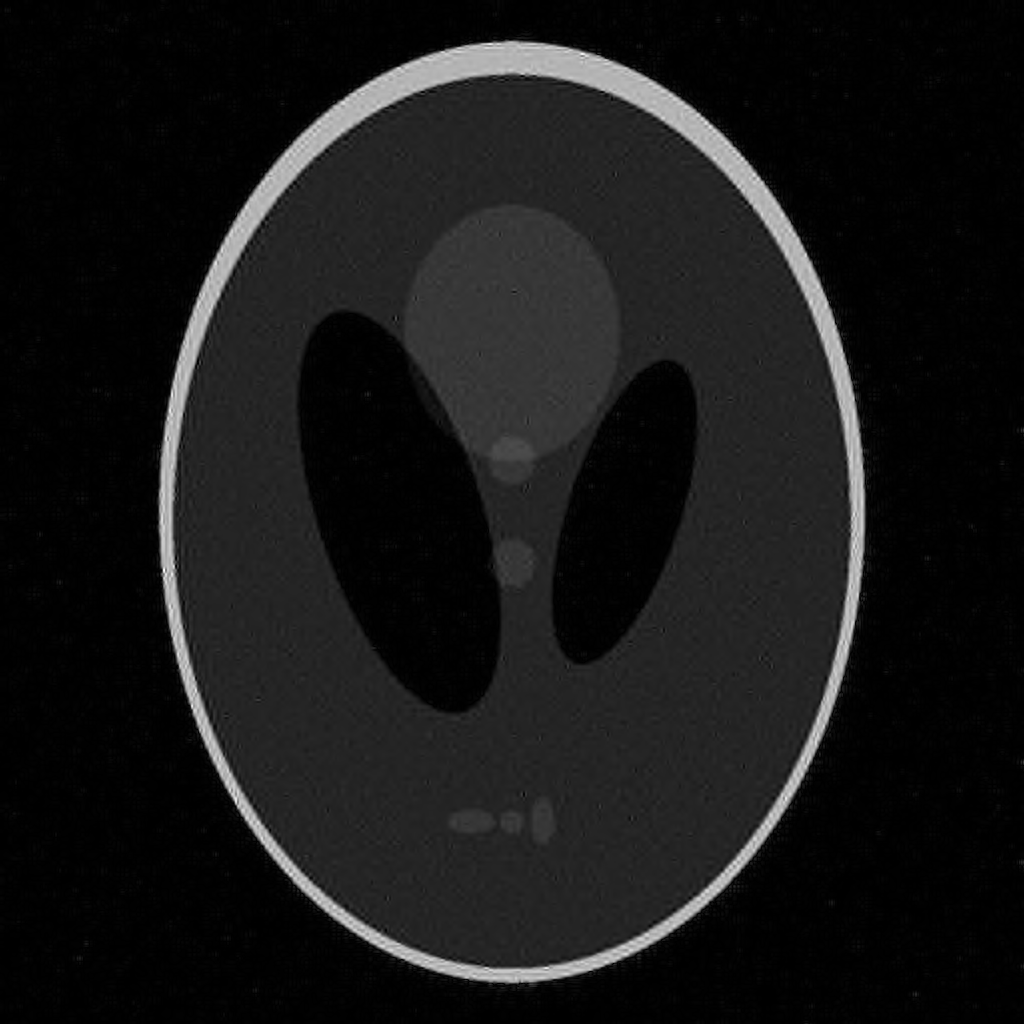}
\label{fig:phantom_gauss:l1A}
}
\subfigure[Recovery using the Sum-of-Squared approach]{
\includegraphics[width=0.22\linewidth]{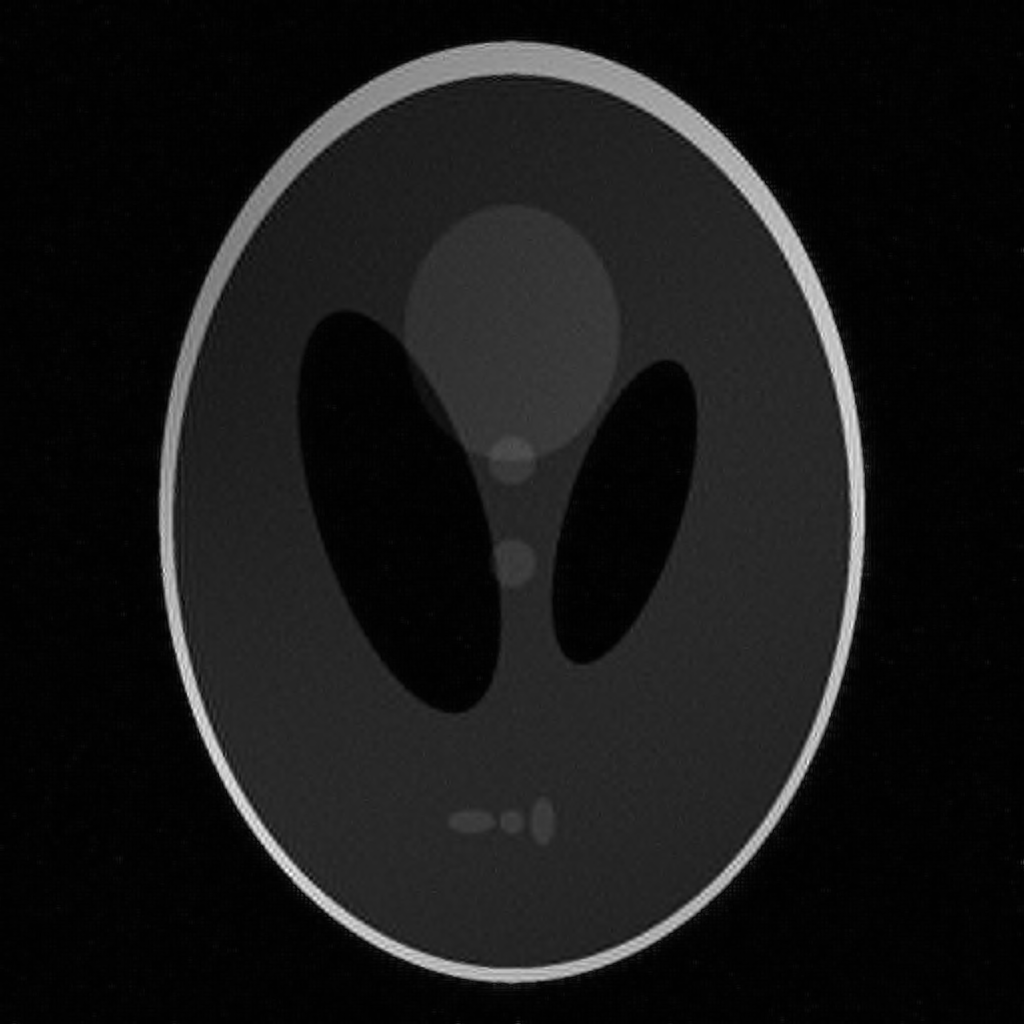}
\label{fig:phantom_gauss:sos}
}
\subfigure[Recovery using the fused $\ell^1$ recovery]{
\includegraphics[width=0.22\linewidth]{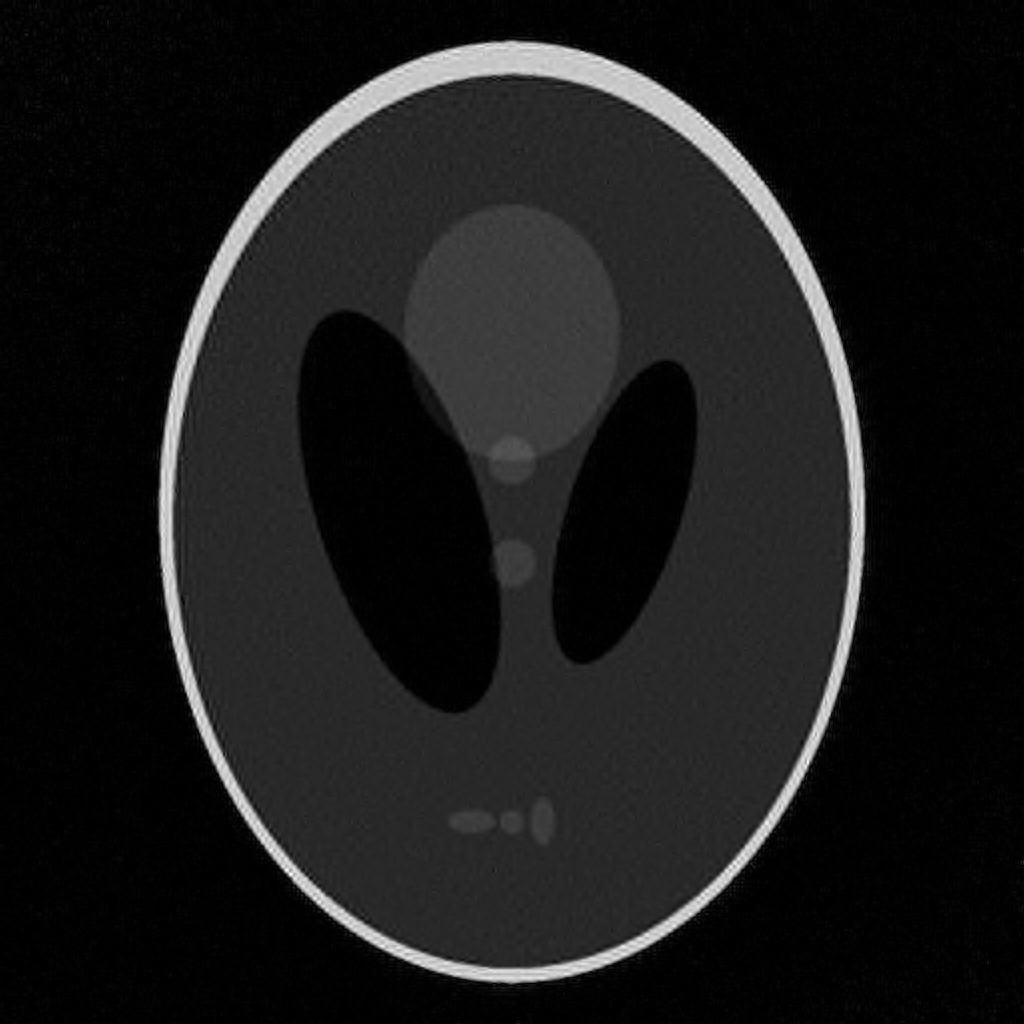}
\label{fig:phantom_gauss:fcs}
}
\subfigure[Recovery using a fused total variation minimization]{
\includegraphics[width=0.22\linewidth]{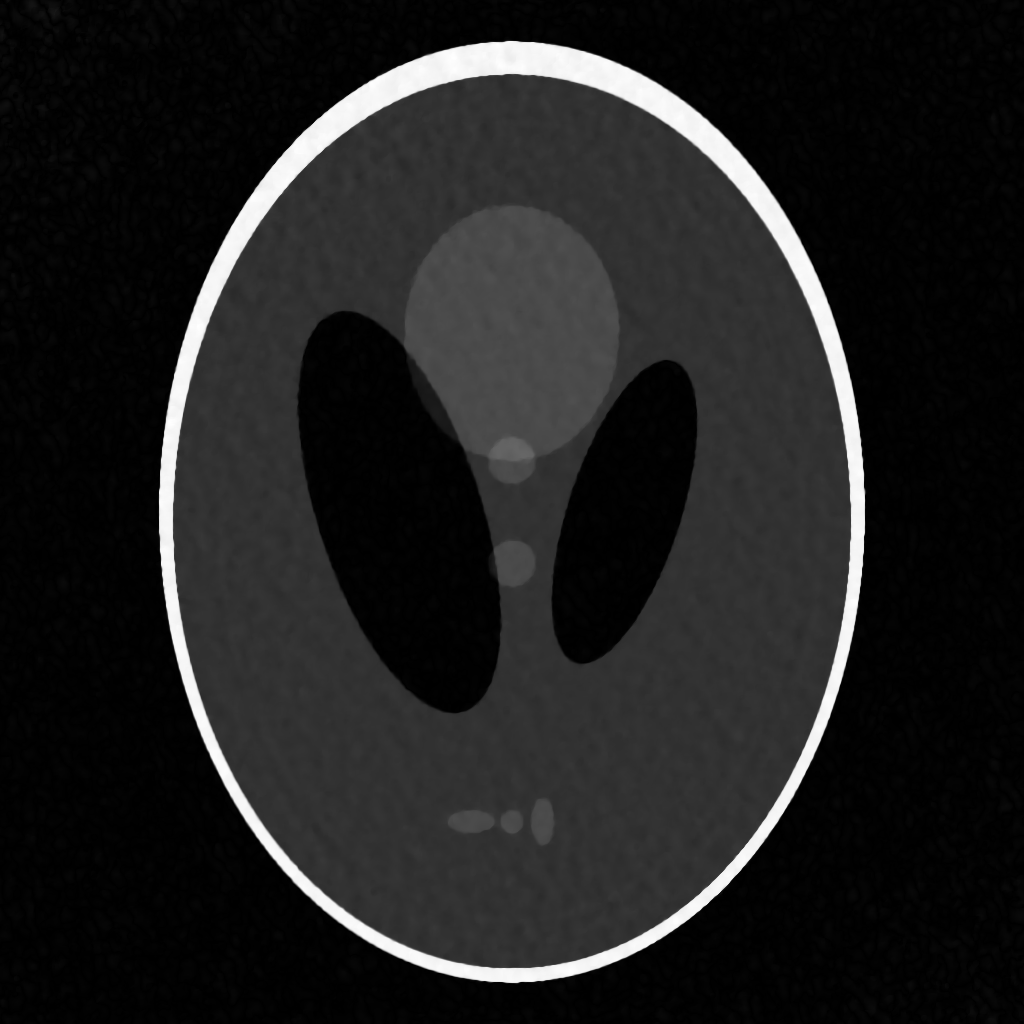}
\label{fig:phantom_gauss:fcstv}
} \\ 
\subfigure[Pointwise error single sensor $\ell^1$ analysis]{
\includegraphics[width=0.22\linewidth]{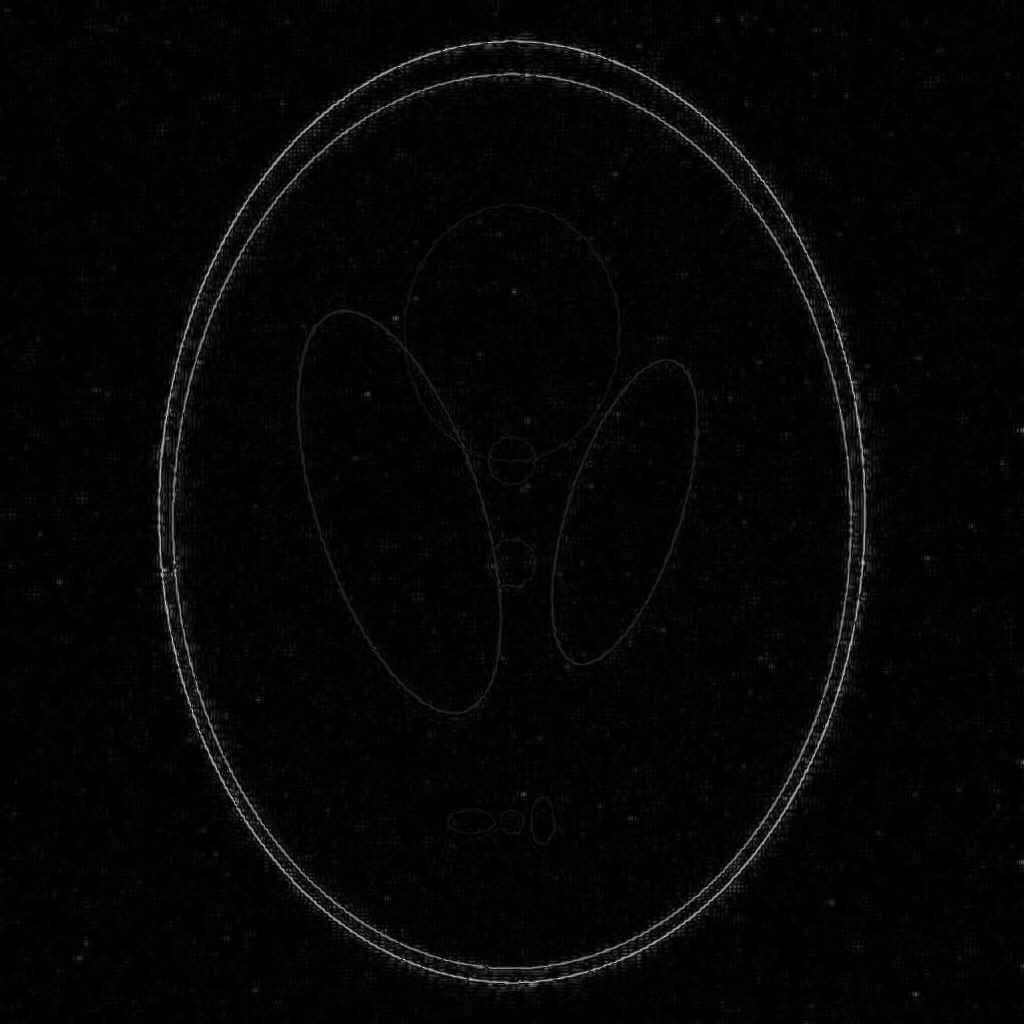}
\label{fig:phantom_gauss:l1_error}
}
\subfigure[Pointwise error using the Sum-of-Squared approach]{
\includegraphics[width=0.22\linewidth]{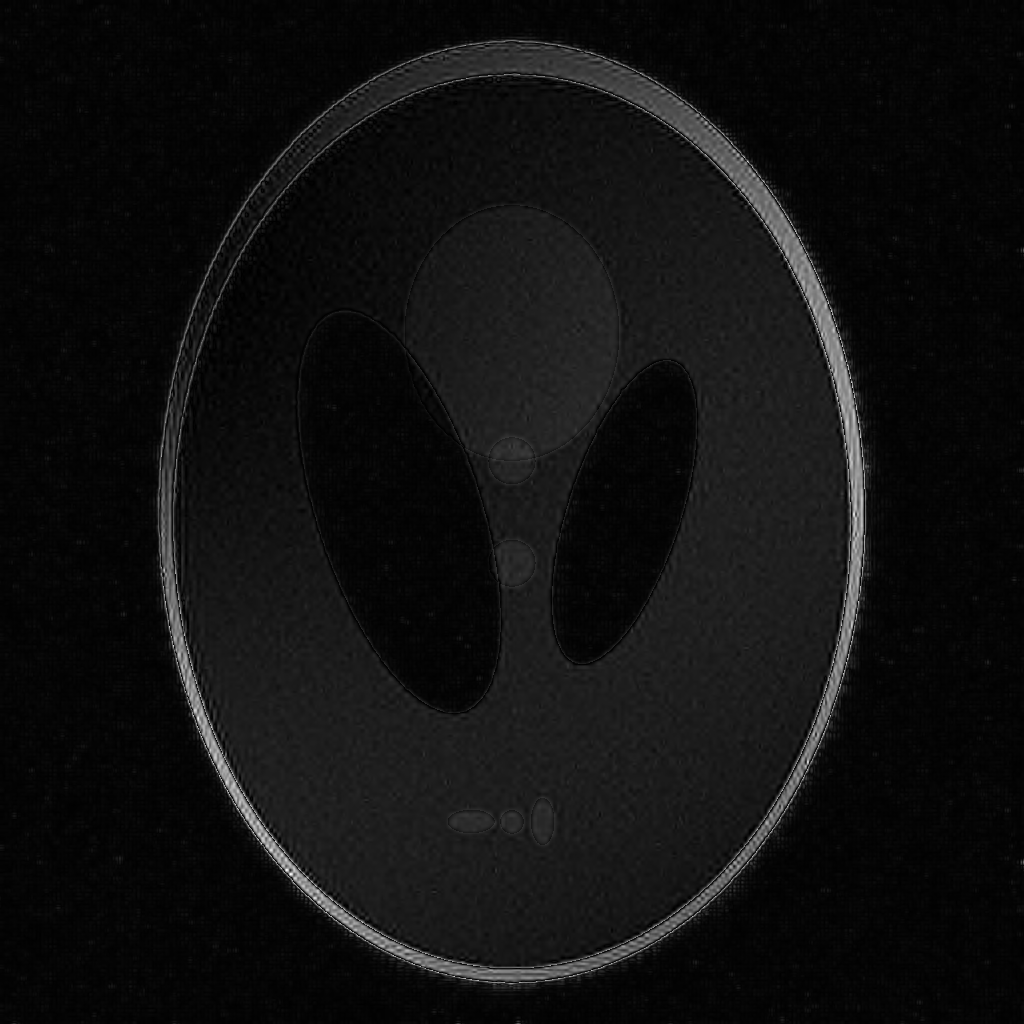}
\label{fig:phantom_gauss:sos_error}
}
\subfigure[Pointwise error using the fused $\ell^1$ analysis]{
\includegraphics[width=0.22\linewidth]{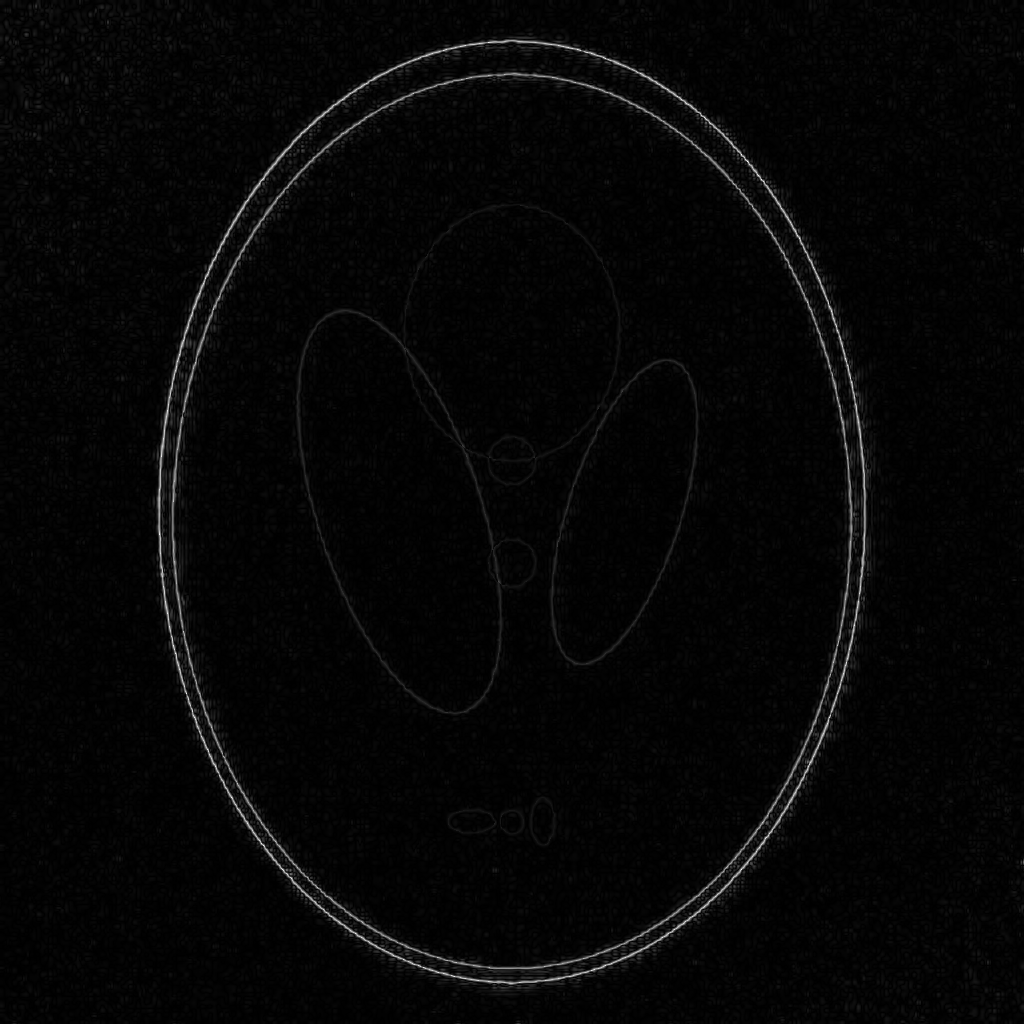}
\label{fig:phantom_gauss:fcs_error}
}
\subfigure[Pointwise error using a fused total variation minimization]{
\includegraphics[width=0.22\linewidth]{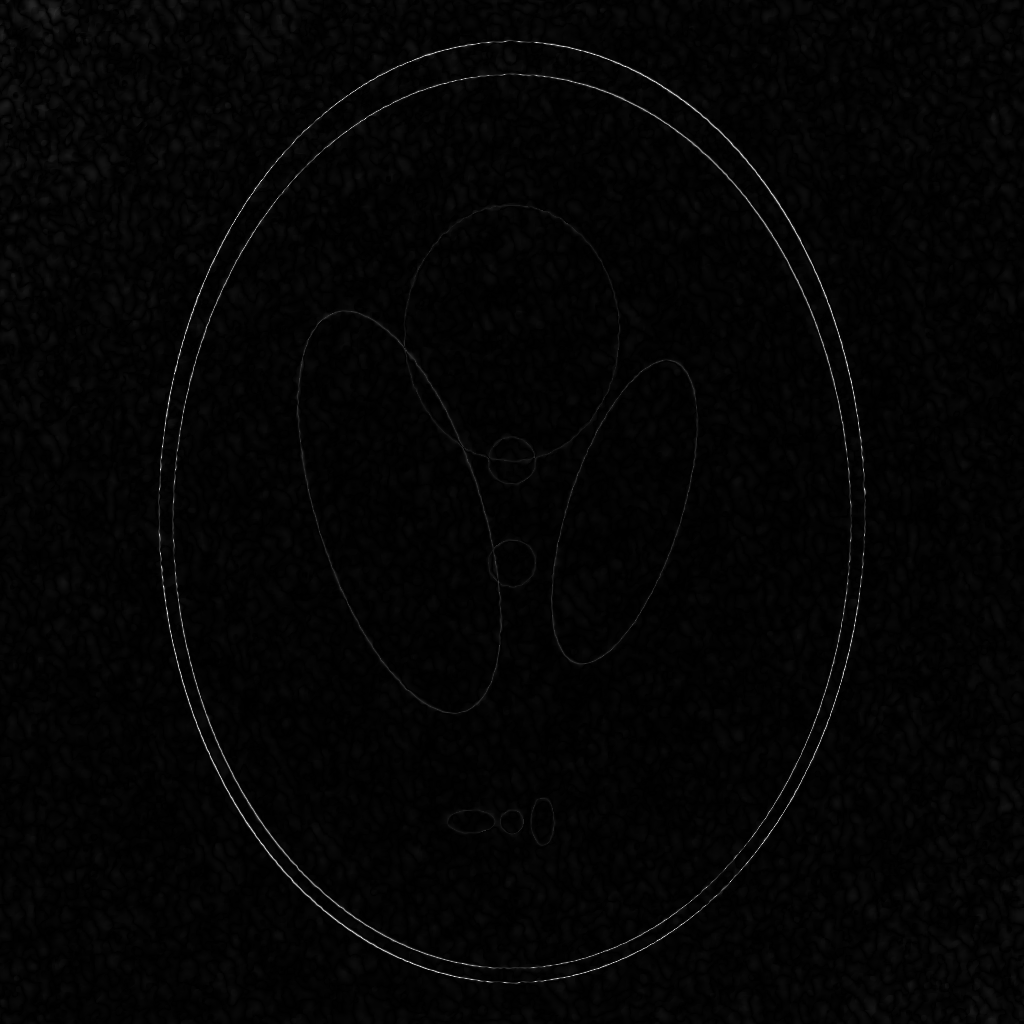}
\label{fig:phantom_gauss:fcstv_error}
}
\caption{Recovery of the phantom MR Image of size $1024 \times 1024$ from highly undersampled and noisy measurements (additive Gaussian noise, variance $0.05$). The sparse minimization is done on Daubechies 4 wavelet coefficients and on the gradient for the rightmost image. }
\label{fig:phantom_gauss}
\end{figure}
The results were obtained from a spherical beam illuminating the input image. 
The measurements are obtained by sampling $2.45\%$ ($25664$ samples) of the Fourier coefficients at random iid from a Gaussian distribution. 
Again, some Gaussian noise with variance $0.05$ is added to every measurements.

\begin{table}[htb]
\centering
\begin{tabular}{l|cccc}
\hline
               & $\ell^1$ analysis  & Sum of Square fusion & Fused $\ell^1$ analysis & Fused total variation \\ \hline
SSIM           &      $0.7972$      &      $0.4518$        &        $0.7872$         &       $0.8225$        \\
PSNR           &      $27.907$      &      $18.719$        &        $27.909$         &       $32.606$        \\
$\ell^2$ error &      $41.207$      &      $118.67$        &        $41.196$         &       $23.988$        \\ \hline
\end{tabular}
\caption{Accuracy of the various recovery methods of the phantom image measured with the Structural SIMilarity index, Peak Signal to Noise Ratio, and pointwise $\ell^2$ error.}
\end{table}

As presented in these examples, our method is capable of handling highly complex signals in potentially many dimensions yet keeping a very low number of samples taken. 
We can also empirically verify that the presented approach enjoys more robustness to noise and variations in scene illuminations than other known methods.

\section{Conclusion}
This paper introduces a novel way to look at signal sensing and reconstruction by utilizing fusion frames. 
The suggested method is very general by nature and can be applied in many situations, as illustrated in the previous section. 
Our research suggests that it is possible to recover \correction{high complexity} from many low dimensional pieces of information. 
This splitting of information can be either enforced by design of the sensors, or by the physical nature of the problem being considered. 

The present article proves that the number of measurements required locally is inversely proportional to the number of sensors, thereby allowing to use many cheap sensors. 
Empirically, we verified that our approach allowed to handle noise even in the case of strong undersampling of the scenes. 

This work can be further developed by integrating various sensors instead of a single one, repeated $n$ times. 
Further real-world applications, for instance in SAR systems, can also be implemented but the details are left for future research.

\appendix
\section{Further numerical results}
\label{appendix}

\subsection{Comparison with Douglas-Rachford optimization and variable density sampling (DR-VDS)}
In this set of tests, we numerically assess how our approach behaves in the presence of a variable density sampling scheme which contains a low-frequency deterministic part and different random sampling patterns. 
This approach is motivated by results in~\cite{chauffert2014vds,boyer2015compressed}, in which the sampling pattern is also highly structured. 
In particular, we look at basic variable density sampling scheme in which low frequencies are deterministically selected and the high-frequency components are sampled at random, according to either a radial pattern (see Figs.~\ref{fig:Chauffert_CScan_radial},\ref{fig:Chauffert_RealBrain_radial}) or a spiraling pattern (see Figs.~\ref{fig:Chauffert_CScan_spiral},\ref{fig:Chauffert_RealBrain_spiral}). 

\begin{figure}[htb]
\centering
\subfigure[Original image]{
\includegraphics[width=0.22\linewidth]{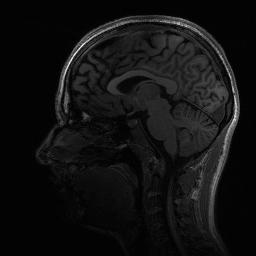}
\label{fig:sagital_image:radial:original}
}
\subfigure[Recovery using the Douglas-Rachford optimization provided in~\cite{chauffert2014vds}]{
\includegraphics[width=0.22\linewidth]{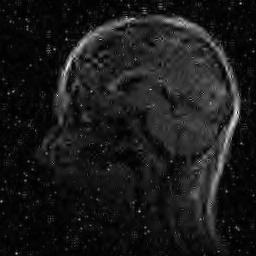}
\label{fig:sagital_image:radial:vds}
}
\subfigure[Recovery using the fused $\ell^1$ recovery]{
\includegraphics[width=0.22\linewidth]{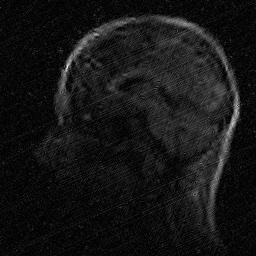}
\label{fig:sagital_image:radial:fcs}
}
\subfigure[Recovery using a fused total variation minimization]{
\includegraphics[width=0.22\linewidth]{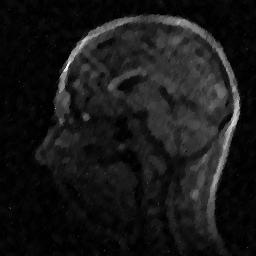}
\label{fig:sagital_image:radial:fcstv}
} \\ 
\subfigure[Sampling pattern of the $k$-space]{
\includegraphics[width=0.22\linewidth]{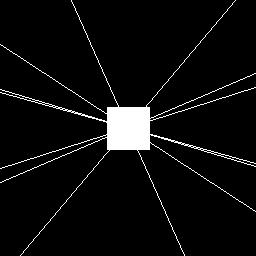}
\label{fig:sagital_image:radial:pattern}
}
\subfigure[Pointwise error using the Douglas-Rachford optimization provided in~\cite{chauffert2014vds}]{
\includegraphics[width=0.22\linewidth]{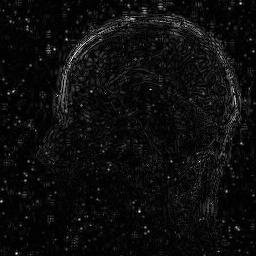}
\label{fig:sagital_image:radial:vds_error}
}
\subfigure[Pointwise error using the fused $\ell^1$ recovery]{
\includegraphics[width=0.22\linewidth]{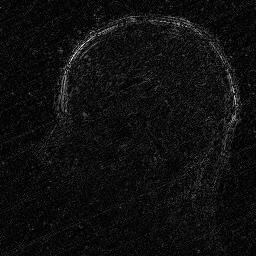}
\label{fig:sagital_image:radial:fcs_error}
}
\subfigure[Pointwise error using a fused total variation minimization]{
\includegraphics[width=0.22\linewidth]{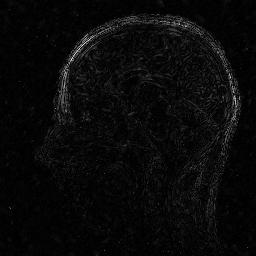}
\label{fig:sagital_image:radial:fcstv_error}
}
\caption{Recovery of a sagital brain MRI from 1/20 undersampled measurements taken as radial lines in the Fourier space and a deterministic low-frequency area. The sparse minimization is done on symmlet wavelet coefficients and on the gradient for the rightmost image. }
\label{fig:Chauffert_CScan_radial}
\end{figure}

\begin{figure}[htb]
\centering
\subfigure[Original image]{
\includegraphics[width=0.22\linewidth]{CScan_Original.jpg}
\label{fig:sagital_image:spiral:original}
}
\subfigure[Recovery using the Douglas-Rachford optimization provided in~\cite{chauffert2014vds}]{
\includegraphics[width=0.22\linewidth]{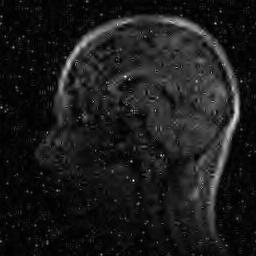}
\label{fig:sagital_image:spiral:vds}
}
\subfigure[Recovery using the fused $\ell^1$ recovery]{
\includegraphics[width=0.22\linewidth]{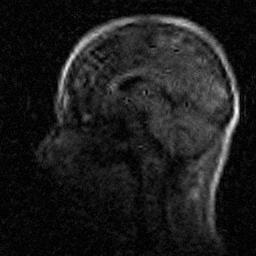}
\label{fig:sagital_image:spiral:fcs}
}
\subfigure[Recovery using a fused total variation minimization]{
\includegraphics[width=0.22\linewidth]{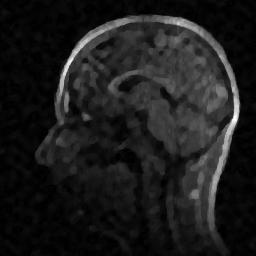}
\label{fig:sagital_image:spiral:fcstv}
} \\ 
\subfigure[Sampling pattern of the $k$-space]{
\includegraphics[width=0.22\linewidth]{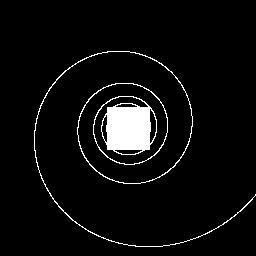}
\label{fig:sagital_image:spiral:pattern}
}
\subfigure[Pointwise error using the Douglas-Rachford optimization provided in~\cite{chauffert2014vds}]{
\includegraphics[width=0.22\linewidth]{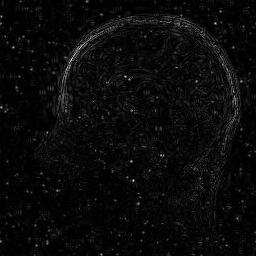}
\label{fig:sagital_image:spiral:vds_error}
}
\subfigure[Pointwise error using the fused $\ell^1$ recovery]{
\includegraphics[width=0.22\linewidth]{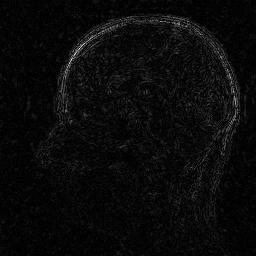}
\label{fig:sagital_image:spiral:fcs_error}
}
\subfigure[Pointwise error using a fused total variation minimization]{
\includegraphics[width=0.22\linewidth]{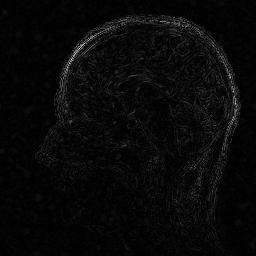}
\label{fig:sagital_image:spiral:fcstv_error}
}
\caption{Recovery of a sagital brain MRI from 1/20 undersampled measurements taken as a spirale in the Fourier space and a determinist low-frequency area. The sparse minimization is done on symmlet wavelet coefficients and on the gradient for the rightmost image. }
\label{fig:Chauffert_CScan_spiral}
\end{figure}

\begin{figure}[htb]
\centering
\subfigure[Original image]{
\includegraphics[width=0.22\linewidth]{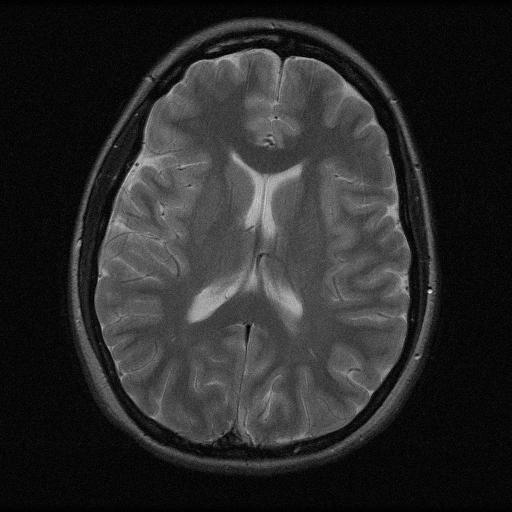}
\label{fig:sagital_image:radial:original}
}
\subfigure[Recovery using the Douglas-Rachford optimization provided in~\cite{chauffert2014vds}]{
\includegraphics[width=0.22\linewidth]{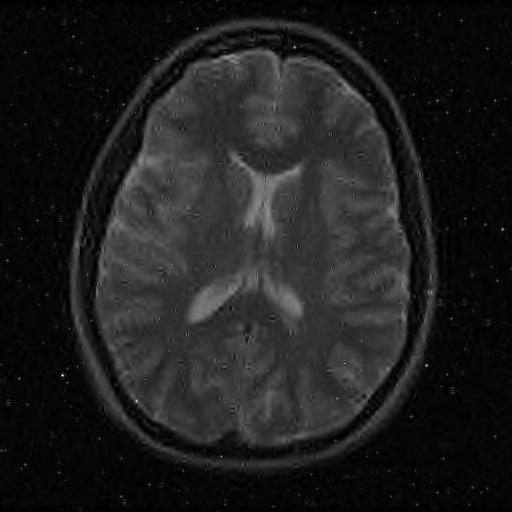}
\label{fig:sagital_image:radial:vds}
}
\subfigure[Recovery using the fused $\ell^1$ recovery]{
\includegraphics[width=0.22\linewidth]{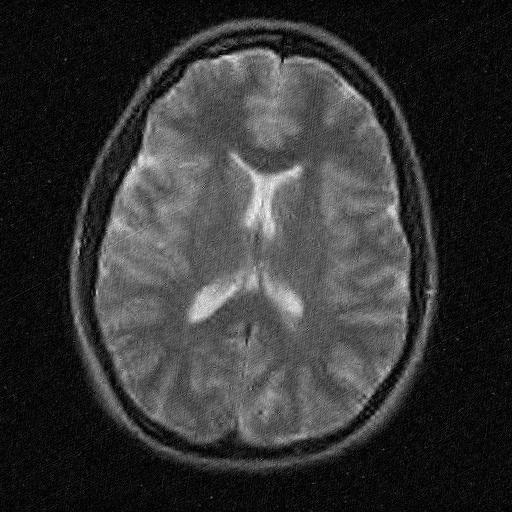}
\label{fig:sagital_image:radial:fcs}
}
\subfigure[Recovery using a fused total variation minimization]{
\includegraphics[width=0.22\linewidth]{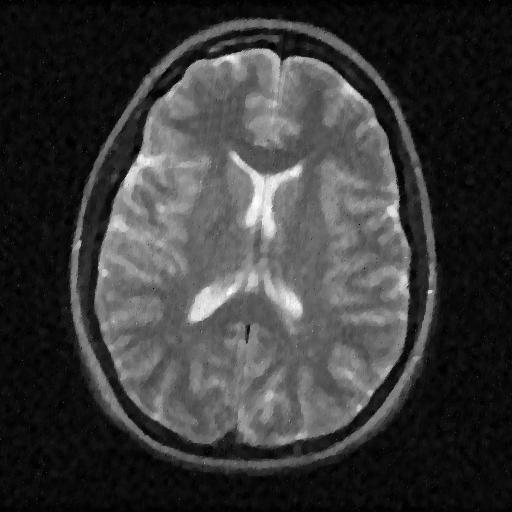}
\label{fig:sagital_image:radial:fcstv}
} \\ 
\subfigure[Sampling pattern of the $k$-space]{
\includegraphics[width=0.22\linewidth]{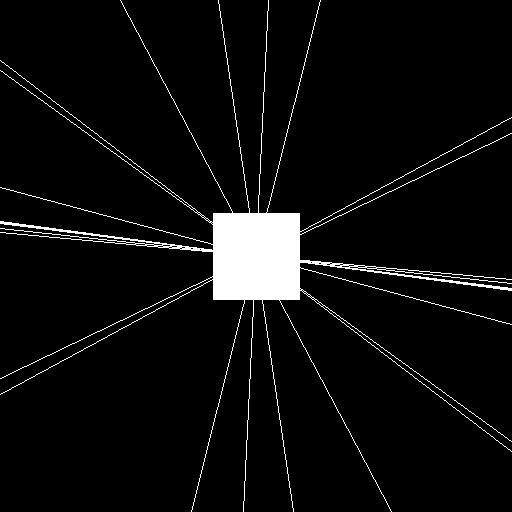}
\label{fig:sagital_image:radial:pattern}
}
\subfigure[Pointwise error using the Douglas-Rachford optimization provided in~\cite{chauffert2014vds}]{
\includegraphics[width=0.22\linewidth]{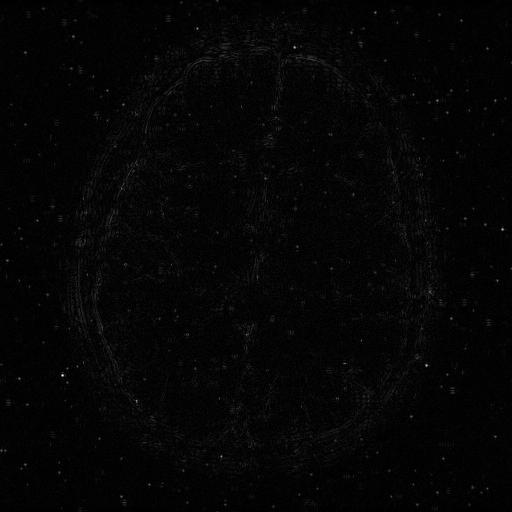}
\label{fig:sagital_image:radial:vds_error}
}
\subfigure[Pointwise error using the fused $\ell^1$ recovery]{
\includegraphics[width=0.22\linewidth]{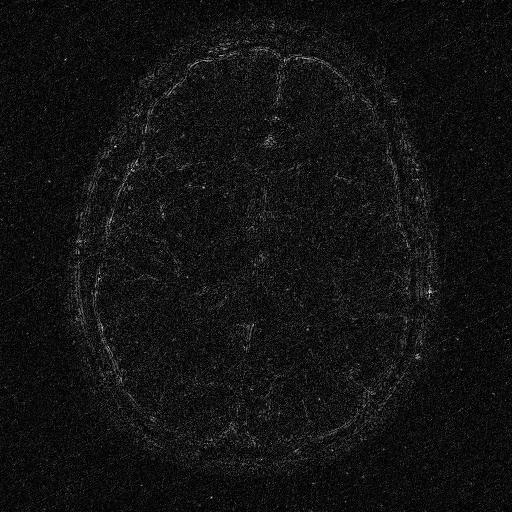}
\label{fig:sagital_image:radial:fcs_error}
}
\subfigure[Pointwise error using a fused total variation minimization]{
\includegraphics[width=0.22\linewidth]{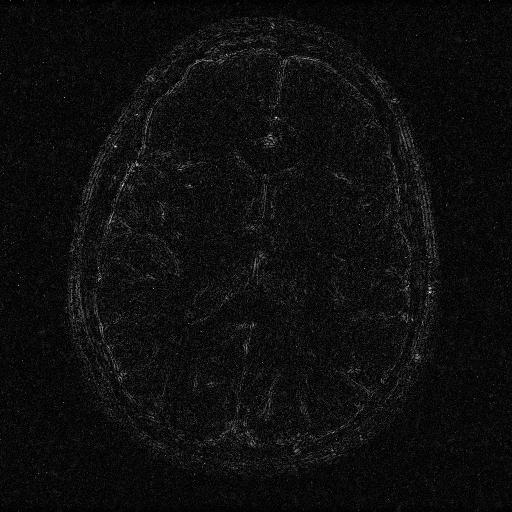}
\label{fig:sagital_image:radial:fcstv_error}
}
\caption{Recovery of an axial brain MRI from 1/20 undersampled measurements taken as radial lines in the Fourier space and a determinist low-frequency area. The sparse minimization is done on symmlet wavelet coefficients and on the gradient for the rightmost image. }
\label{fig:Chauffert_RealBrain_radial}
\end{figure}

\begin{figure}[htb]
\centering
\subfigure[Original image]{
\includegraphics[width=0.22\linewidth]{RealBrain_Original.jpg}
\label{fig:sagital_image:spiral:original}
}
\subfigure[Recovery using the Douglas-Rachford optimization provided in~\cite{chauffert2014vds}]{
\includegraphics[width=0.22\linewidth]{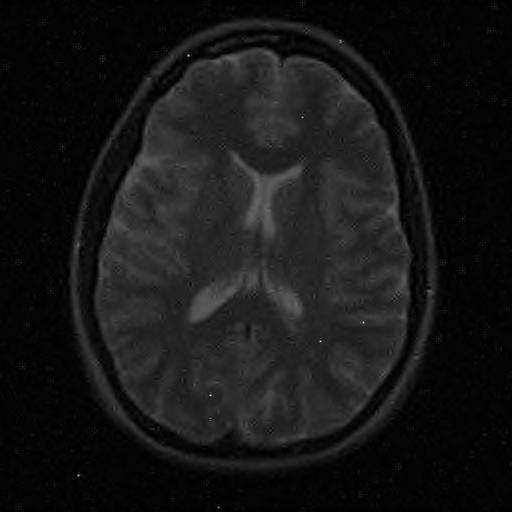}
\label{fig:sagital_image:spiral:vds}
}
\subfigure[Recovery using the fused $\ell^1$ recovery]{
\includegraphics[width=0.22\linewidth]{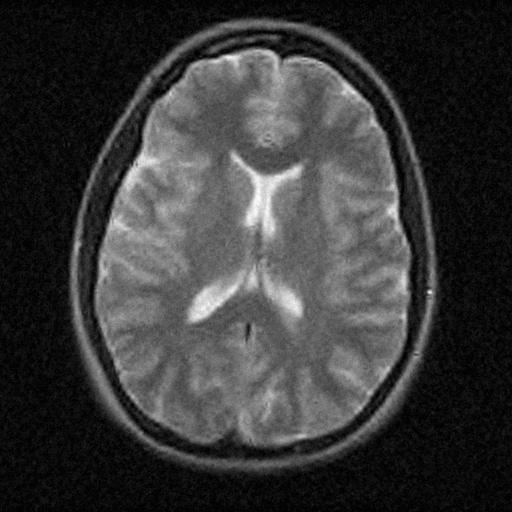}
\label{fig:sagital_image:spiral:fcs}
}
\subfigure[Recovery using a fused total variation minimization]{
\includegraphics[width=0.22\linewidth]{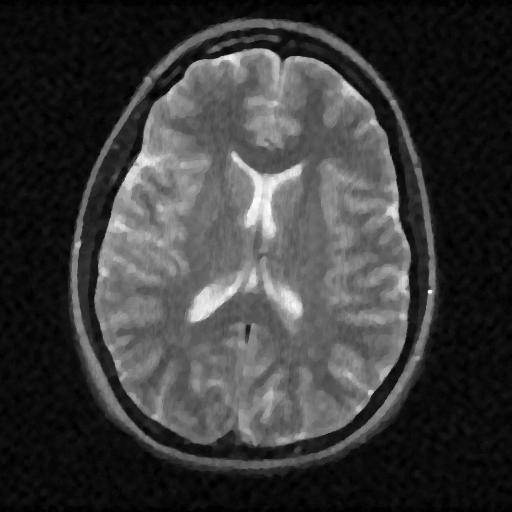}
\label{fig:sagital_image:spiral:fcstv}
} \\ 
\subfigure[Sampling pattern of the $k$-space]{
\includegraphics[width=0.22\linewidth]{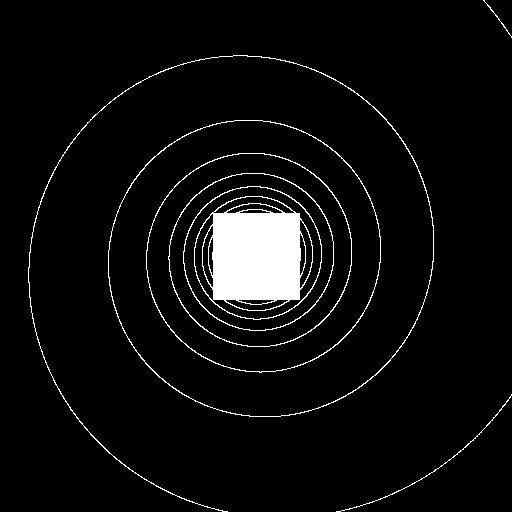}
\label{fig:sagital_image:spiral:pattern}
}
\subfigure[Pointwise error using the Douglas-Rachford optimization provided in~\cite{chauffert2014vds}]{
\includegraphics[width=0.22\linewidth]{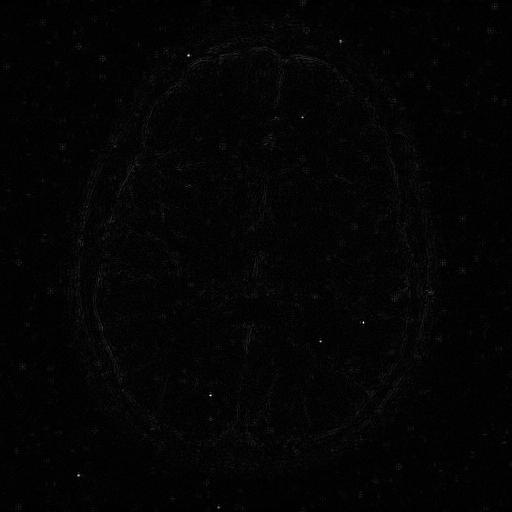}
\label{fig:sagital_image:spiral:vds_error}
}
\subfigure[Pointwise error using the fused $\ell^1$ recovery]{
\includegraphics[width=0.22\linewidth]{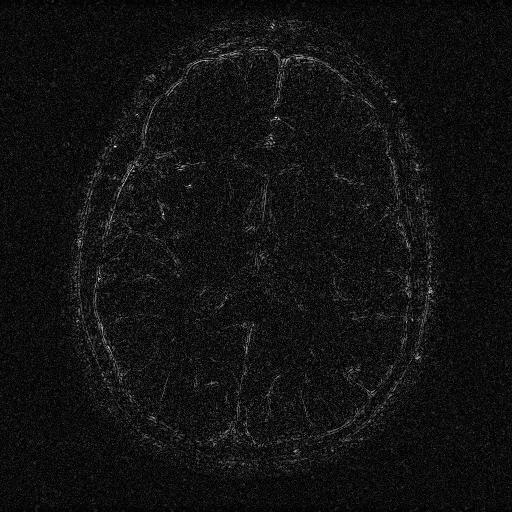}
\label{fig:sagital_image:spiral:fcs_error}
}
\subfigure[Pointwise error using a fused total variation minimization]{
\includegraphics[width=0.22\linewidth]{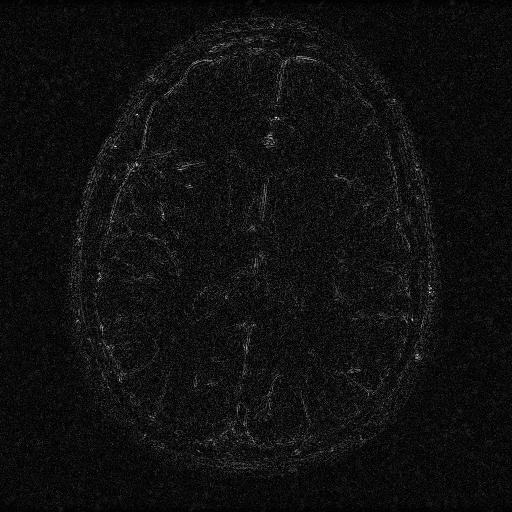}
\label{fig:sagital_image:spiral:fcstv_error}
}
\caption{Recovery of an axial brain MRI from 1/20 undersampled measurements taken as radial lines in the Fourier space and a determinist low-frequency area. The sparse minimization is done on symmlet wavelet coefficients and on the gradient for the rightmost image. }
\label{fig:Chauffert_RealBrain_spiral}
\end{figure}

For this test, we have reused the software already available from the authors of~\cite{chauffert2014vds}\footnote{At the time of writing, the code is available at \url{http://chauffertn.free.fr/codes.html}. }
The samples were obtained by subsampling $5$\% of the Fourier coefficients and adding normal noise with $0$ mean and $0.05$ variance. 
The accuracies can be found in Table~\ref{tab:cscan} for the results depicted in Figures~\ref{fig:Chauffert_CScan_radial} and \ref{fig:Chauffert_CScan_spiral} and in Table~\ref{tab:realbrain} for results depicted in Figures~\ref{fig:Chauffert_RealBrain_radial} and~\ref{fig:Chauffert_RealBrain_spiral}. 
As it is seen from the pointwise error figures, our fusion process is less subject to small local perturbations. 
Moreover, it is clear from the tables that fusing with local recovery obtained by total variation minimization is the most suitable approach.

\begin{table}[htb]
\centering
\begin{tabular}{ll|ccc}
\hline
Sampling                &                &  DR-VDS  & Fused $\ell^1$ & Fused TV								  \\ \hline
\multirow{3}{*}{Radial} & SSIM           & $0.6099$ &   $0.4366$     & $0.7691$                 \\
                        & PSNR           & $26.772$ &   $26.071$     & $29.193$                 \\
                        & $\ell^2$ error & $3.0294$ &   $3.0971$     & $2.4104$                 \\ \hline
\multirow{3}{*}{Spiral} & SSIM           & $0.6311$ &   $0.6830$     & $0.8034$                 \\
                        & PSNR           & $27.331$ &   $28.462$     & $30.336$                 \\ 
                        & $\ell^2$ error & $2.8132$ &   $2.3014$     & $1.7364$                 \\ \hline
\end{tabular}
\caption{Recovery results for the sagital brain MRI. The results are obtained by sampling 5\% of the $k$-space.}
\label{tab:cscan}
\end{table}

\begin{table}[htb]
\centering
\begin{tabular}{ll|ccc}
\hline
Sampling                &                &  DR-VDS  & Fused $\ell^1$ & Fused TV								  \\ \hline
\multirow{3}{*}{Radial} & SSIM           & $0.6336$ &   $0.4609$     & $0.6847$                 \\
                        & PSNR           & $30.917$ &   $26.807$     & $30.101$                 \\
                        & $\ell^2$ error & $3.1783$ &   $3.1700$     & $2.1869$                 \\ \hline
\multirow{3}{*}{Spiral} & SSIM           & $0.6645$ &   $0.6150$     & $0.7174$                 \\
                        & PSNR           & $34.695$ &   $29.283$     & $31.057$                 \\ 
                        & $\ell^2$ error & $2.8709$ &   $2.4314$     & $1.6922$                 \\ \hline
\end{tabular}
\caption{Recovery results for the axial slice brain MRI. The results are obtained by sampling 5\% of the $k$-space.}
\label{tab:realbrain}
\end{table}

This behavior is true for both images, which are real-world images of brain scans. 
It is worth keeping in mind that SSIM is a structural similarity which tries to emulate the human visual perception while the other metrics are pure machinery.
As can be seen both from the figures and the tables, the proposed method performs at least as well as previous approaches, and may behave better when considering total variation minimization.

\subsection{Single recovery from distributed measurements}
We investigate here use of a single recovery procedure from multiple measurements. 
In this settings, and following the ideas developed in~\cite{adcock2016CSparallel} and described in the previous sections, we are trying to recover the unknown vector $\xbf$ from the measurements $\ybf^{(i)}$ by solving a single sparse minimization problem.

\begin{figure}[htb]
\centering
\subfigure{
\includegraphics[width=0.12\linewidth]{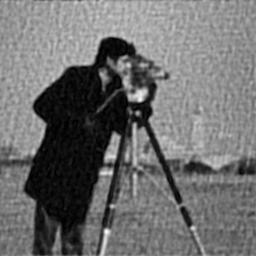}
\label{fig:adcock:adcock20}
}
\subfigure{
\includegraphics[width=0.12\linewidth]{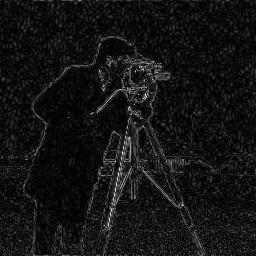}
\label{fig:adcock:adcock20_error}
}
\subfigure{
\includegraphics[width=0.12\linewidth]{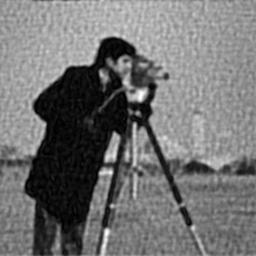}
\label{fig:adcock:adcock25}
}
\subfigure{
\includegraphics[width=0.12\linewidth]{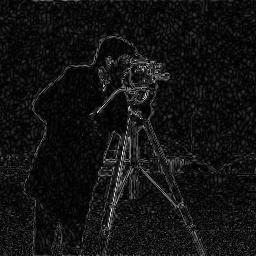}
\label{fig:adcock:adcock25_error}
}
\subfigure{
\includegraphics[width=0.12\linewidth]{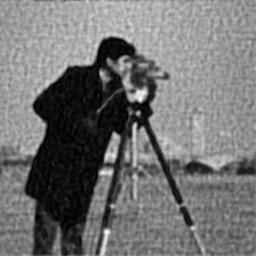}
\label{fig:adcock:adcock30}
}
\subfigure{
\includegraphics[width=0.12\linewidth]{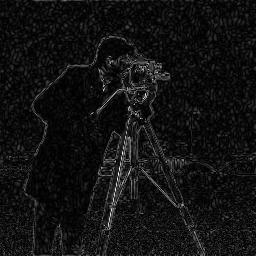}
\label{fig:adcock:adcock30_error}
} \\ 
\subfigure{
\includegraphics[width=0.12\linewidth]{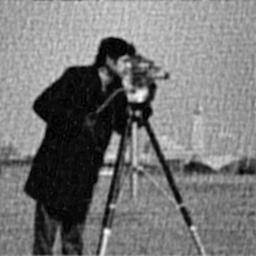}
\label{fig:adcock:cs20}
}
\subfigure{
\includegraphics[width=0.12\linewidth]{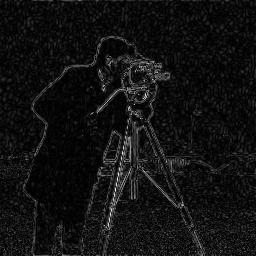}
\label{fig:adcock:cs20_error}
}
\subfigure{
\includegraphics[width=0.12\linewidth]{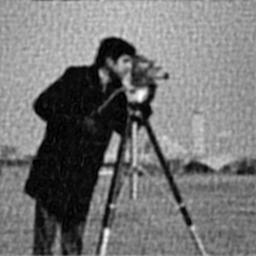}
\label{fig:adcock:cs25}
}
\subfigure{
\includegraphics[width=0.12\linewidth]{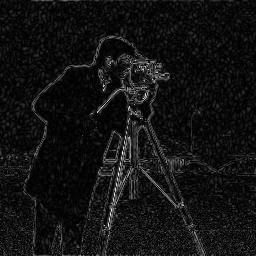}
\label{fig:adcock:cs25_error}
}
\subfigure{
\includegraphics[width=0.12\linewidth]{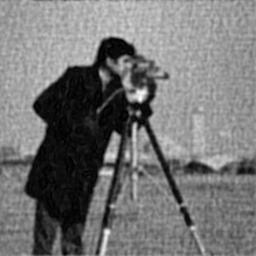}
\label{fig:adcock:cs30}
}
\subfigure{
\includegraphics[width=0.12\linewidth]{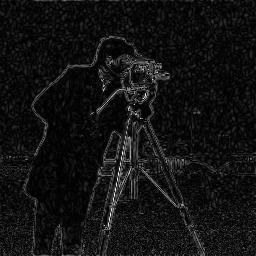}
\label{fig:adcock:cs30_error}
} \\ 
\subfigure{
\includegraphics[width=0.12\linewidth]{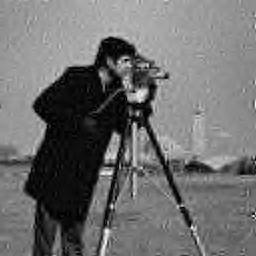}
\label{fig:adcock:l1a20}
}
\subfigure{
\includegraphics[width=0.12\linewidth]{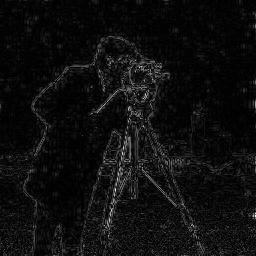}
\label{fig:adcock:l1a20_error}
}
\subfigure{
\includegraphics[width=0.12\linewidth]{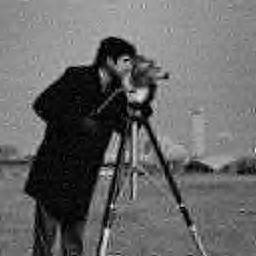}
\label{fig:adcock:l1a25}
}
\subfigure{
\includegraphics[width=0.12\linewidth]{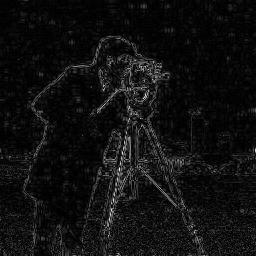}
\label{fig:adcock:l1a25_error}
}
\subfigure{
\includegraphics[width=0.12\linewidth]{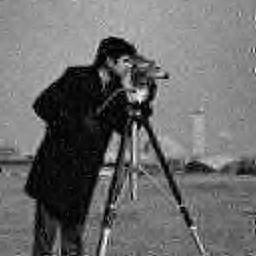}
\label{fig:adcock:l1a30}
}
\subfigure{
\includegraphics[width=0.12\linewidth]{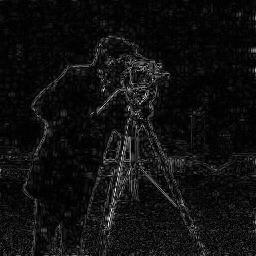}
\label{fig:adcock:l1a30_error}
} \\
\subfigure{
\includegraphics[width=0.12\linewidth]{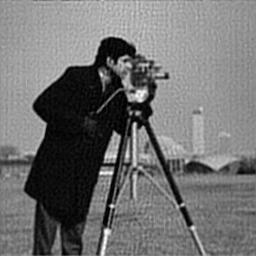}
\label{fig:adcock:fcs20}
}
\subfigure{
\includegraphics[width=0.12\linewidth]{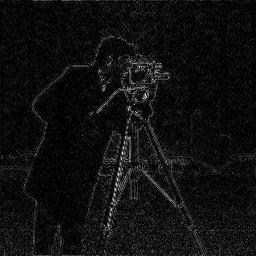}
\label{fig:adcock:fcs20_error}
}
\subfigure{
\includegraphics[width=0.12\linewidth]{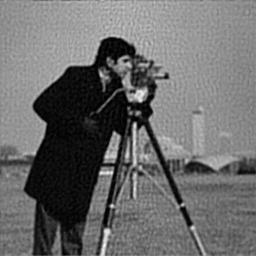}
\label{fig:adcock:fcs25}
}
\subfigure{
\includegraphics[width=0.12\linewidth]{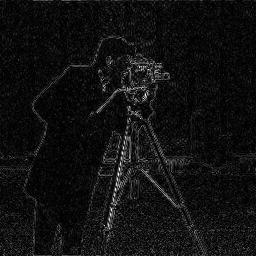}
\label{fig:adcock:fcs25_error}
}
\subfigure{
\includegraphics[width=0.12\linewidth]{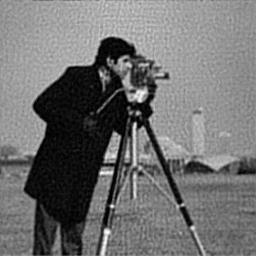}
\label{fig:adcock:fcs30}
}
\subfigure{
\includegraphics[width=0.12\linewidth]{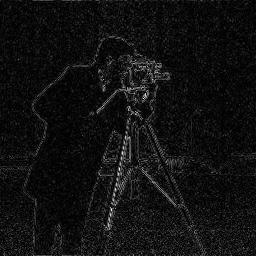}
\label{fig:adcock:fcs30_error}
}
\caption{Recovery of the camera man image when using a centralized recovery vs. local recoveries. The additive noise in the measurements is normal with variance 0.02 (left), 0.025 (central), and 0.03 (right). The first row depicts the recovery from a single, centralized, sparse approximation. The second row displays the traditional compressed sensing reocvery. The third row is the $\ell^1$ analysis model while the bottom row corresponds to the proposed fused $\ell^1$ framework.}
\label{fig:SingleCompute}
\end{figure}

\begin{table}[htb]
\centering
\begin{tabular}{ll|cccc}
\hline
Noise level             &                &   DSL1   &    CS    & L1 Analysis & Fused L1 analysis \\ \hline
\multirow{3}{*}{$0.020$}& SSIM           & $0.5247$ & $0.5693$ &   $0.6432$  &     $0.6776$      \\
                        & PSNR           & $23.028$ & $23.289$ &   $23.371$  &     $25.428$      \\
                        & $\ell^2$ error & $18.064$ & $17.530$ &   $17.365$  &     $13.703$      \\ \hline
\multirow{3}{*}{$0.025$}& SSIM           & $0.5279$ & $0.5711$ &   $0.6354$  &     $0.6478$      \\
                        & PSNR           & $23.000$ & $23.240$ &   $23.273$  &     $25.064$      \\ 
                        & $\ell^2$ error & $18.122$ & $17.629$ &   $17.562$  &     $14.290$      \\ \hline
\multirow{3}{*}{$0.030$}& SSIM           & $0.5200$ & $0.5625$ &   $0.6325$  &     $0.6231$      \\
                        & PSNR           & $22.912$ & $23.127$ &   $23.074$  &     $24.980$      \\ 
                        & $\ell^2$ error & $18.308$ & $17.860$ &   $17.970$  &     $14.429$      \\ \hline
\end{tabular}
\caption{Recovery results for the camera man image.}
\label{tab:adcock}
\end{table}

Fig.~\ref{fig:SingleCompute} shows the results obtained by various methods, namely, a single $\ell^1$ minimization computed on the multi-channel measurements (similar to the approach of~\cite{adcock2016CSparallel}) depicted on the first row, and traditional compressed sensing in the second row, an $\ell^1$ analysis model using Daubechies 4 wavelets in the third row, and our fused framework where the local recovery are done by minimizing the coefficients on a Daubechies 4 wavelet basis. 

The fusion of the local recovery shows the best of the results, which is also confirmed by the values for SSIM, PSNR, and pointwise $\ell^2$ error in Table~\ref{tab:adcock}.

\subsection{Robustness to erasure}
One advantage claimed for this framework is the robustness of the approach to local errors. 
This is believed to be a consequence of the redundancy inherent to frames. 
We demonstrate this point by considering the following setup. 
Assume given the camera man image as in the previous experiments. 
We emulate faulty measures by considering salt and pepper noise in the image prior to taking the Fourier measurements and filtering operations. 

\begin{figure}[htb]
\centering
\subfigure[$\ell^2$ error of the recovery]{
\includegraphics[width=0.25\linewidth]{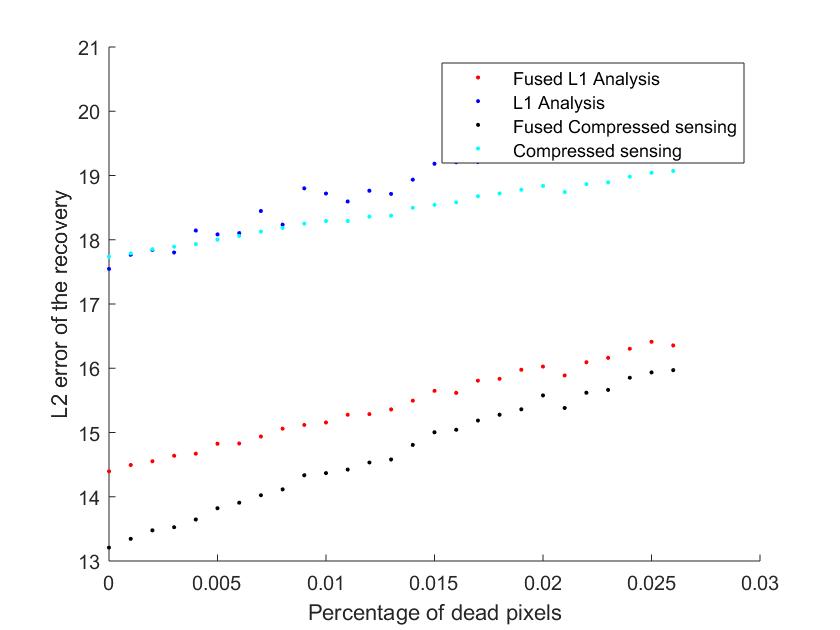}
\label{fig:SnP:l2}
}
\subfigure[PSNR of the recovery]{
\includegraphics[width=0.25\linewidth]{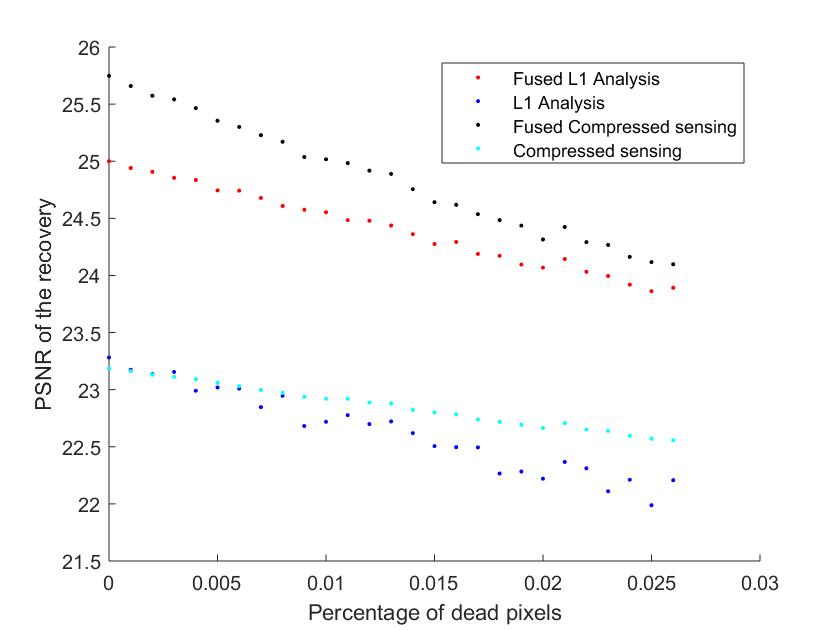}
\label{fig:SnP:psnr}
}
\subfigure[SSIM score of the recovery]{
\includegraphics[width=0.25\linewidth]{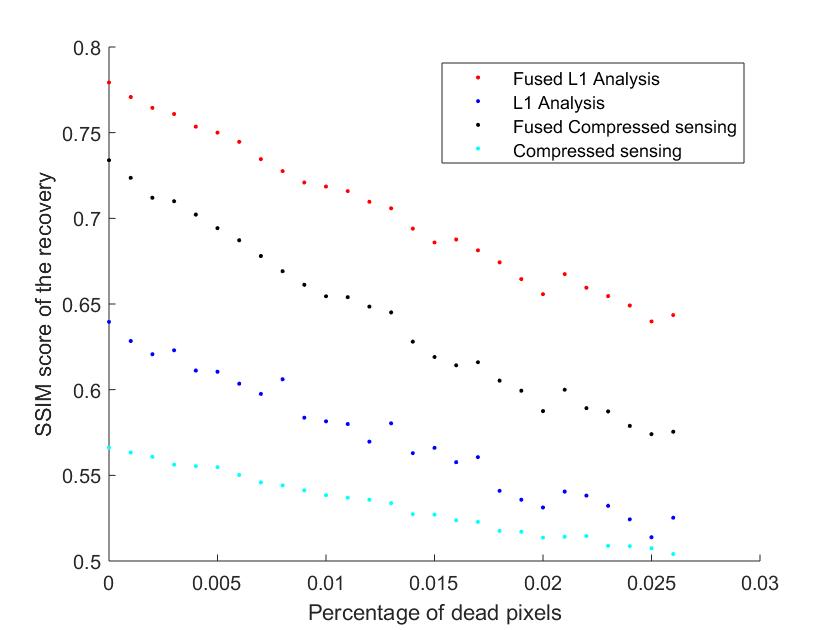}
\label{fig:SnP:ssim}
}
\caption{Recovery in presence of erasures in the input image.}
\label{fig:SnPnoise}
\end{figure}

The graphs in Fig.~\ref{fig:SnPnoise} show the evolution of the fidelity of the recovered images, as the density of faulty pixels increases. 
It is evident from these figures that the fusion process (red or black curves) favors robustness in the presence of erasure. 
Even in a regime where up to $30$\% of the pixels cannot be considered reliable, our fusion process is still able to achieve high fidelity with the original images. 
The results shown here were also obtained with some additive Gaussian noise post measurements.

\section*{Acknowledgment}
R. A. is partially supported by the BSU Aspire Research Grant ``Frame Theory and Modern Sampling Strategies'' .
J.-L. B. thanks the support of the European Research Council through the grant StG 258926 and the organizors as well as the Hausdorff Center for Mathematics for his participation in the Hausdorff Trimester Program on Mathematics of Signal Processing.
S. L. is partially supported by US NSF grants USA (DMS-1313490, DMS-1010058).


\end{document}